\documentclass[11pt]{article}

\usepackage{fullpage}

\usepackage{setspace}
\usepackage{amsmath, amsfonts, amssymb, amsthm}
\usepackage{natbib}
\usepackage{authblk}
\usepackage{xltabular}
\usepackage{booktabs}

\author[1]{Ron Berman}
\author[1]{Walter W. Zhang}
\author[2]{Hangcheng Zhao}
\affil[1]{The Wharton School, The University of Pennsylvania}
\affil[2]{Rutgers Business School}
\date{\today}

\usepackage{booktabs}
\usepackage[ruled]{algorithm2e} 

\SetAlFnt{\small}
\SetAlCapFnt{\small}
\SetAlCapNameFnt{\small}
\SetAlCapHSkip{0pt}
\IncMargin{-\parindent}

\usepackage{graphicx}
\usepackage{xcolor}
\usepackage{soul}
\usepackage{tikz}
\usepackage{enumitem}
\usepackage[table]{xcolor} 
\usepackage{threeparttable}

\usepackage{tabularx}
\usepackage{array}
\DeclareMathOperator*{\argmax}{arg\,max}

\usepackage{nicematrix}

\usepackage[hidelinks]{hyperref}
\usepackage{cleveref}
\usepackage{url}
\usepackage{subcaption}                                                       

\newtheorem{lemma}{Lemma}
\newtheorem{theorem}{Theorem}

\title{Valuing Winners: When and How to Correct for Selection Bias in Randomized Experiments\thanks{
Berman: \href{mailto:ronber@wharton.upenn.edu}{ronber@wharton.upenn.edu}, 
Zhang: \href{mailto:wwz@wharton.upenn.edu}{wwz@wharton.upenn.edu}, 
Zhao: \href{mailto:hangcheng.zhao@rutgers.edu}{hangcheng.zhao@rutgers.edu}.
We thank Eric T. Bradlow, Sanjog Misra, and Christophe Van den Bulte for their helpful comments.
}}

\begin{document}

\maketitle
\thispagestyle{empty} 
\vspace{-30pt}
\begin{abstract}

Decision-makers often deploy the best-performing treatment from a randomized experiment, creating a winner’s curse: selection favors treatments whose observed outcomes are high partly because of statistical noise, so the na\"{i}ve estimate of the winner is upward biased. We distinguish two forms of winner’s curse, bias relative to the true best treatment (global) and bias relative to the selected treatment’s true mean (selective), and link them to regret from deploying a suboptimal treatment. This framework defines seven decision-relevant evaluation targets: mean bias, mean squared error, and confidence interval coverage for the global and selective winner’s curse, and mean regret. 
We then show that methods that perform well on one target can perform poorly on others, so corrections should be matched to the manager’s objective. Across simulations with varying effect sizes, multiple-arm settings, and data calibrated to an online A/B testing platform, no method dominates uniformly: the plug-in estimator performs best when treatment differences are large, cross-fitting performs best when treatments are similar, and resampling methods often achieve low mean squared error for moderate differences. We also introduce an adaptive empirical likelihood procedure that delivers asymptotically valid confidence intervals across settings without the tuning sensitivity of resampling-based methods.
\end{abstract}
\textbf{Keywords}: A/B Testing, Selection Bias, Field Experiments, Empirical Likelihood, Randomized Control Trials
\newpage

\setcounter{page}{1}
\section{Introduction}
A common task for decision-makers after an experiment (e.g., randomized-controlled trial or A/B test) is to decide whether to implement the best-performing policy going forward. This decision usually appears in one of two variants. In ``champion-challenger'' experiments, an existing treatment (the champion) is already implemented and is used as a default treatment or control condition to determine if switching to the winner (the challenger) will provide enough gain \citep{simester2020efficiently}. In ``pick-the-best'' experiments, there is no pre-existing default treatment to compare to, and the goal is to implement the best treatment among those tested \citep{Feit2019}.

In both scenarios, a common approach is to compute the average outcome for each treatment, and then pick the one with the highest outcome for implementation. This approach, however, involves post-selection inference and leads to what is called the winner's curse \citep{andrews2024inference, xu2025} or the optimizer's curse \citep{SmithWinkler2006}: Making a prediction about future performance based on the observed average outcome is upwards biased because of selection. Intuitively, one does not know if the best performing treatment has a high average outcome in the data because it was lucky in the chosen sample (i.e., due to noise), or because it is truly better on average than other treatments in the population.

The literature proposes multiple solutions to this problem. These include sample splitting, conditional inference \citep{andrews2024inference} and resampling methods such as bootstrap-based methods \citep{xu2025}. Some of this research focuses on finite-sample bias reduction, while some focuses on inference, i.e., computing confidence intervals. Most of these papers, however, do not consider the decision maker's objective and whether their winner's curse correction approach supports the goal of the decision maker.

In this paper, we study how the winner's curse affects decision-makers who rely on experimental estimates to choose among treatments. We distinguish three goals a decision-maker may have: (1) Precisely predicting how a specific winning policy will perform if deployed, (2) Estimating how a specific winning policy will perform compared to the truly best treatment, and (3) Minimizing regret with respect to the performance of the truly best arm. We call the bias arising from the first goal the \emph{selective} winner's curse, which is the difference between the na\"{i}ve estimate and the selected treatment's true mean. We call the bias from the second goal the \emph{global} winner's curse, which is the difference between the na\"{i}ve estimate and the true mean of the best treatment. We then make the connection between these two concepts to the concept of regret in decision-theoretic analysis. 

Our paper makes four contributions. First, we derive an identity linking all three quantities in \Cref{eq:Winners-Curse-Identity} and show that this connection has practical consequences: methods that perform well on one objective can perform poorly on another. Second, for each objective, we analyze how proposed corrections perform and what their implications are for decision-makers.
We also show that in some cases, too much of a good thing can hurt the decision-maker: In the champion-challenger scenario, for example, it is sometimes important to know if the winner outperforms the control to make a decision. It is well-known that commonly used bias correction techniques increase the variance of the estimate due to the bias-variance tradeoff (e.g., bootstrap bias correction \citep{efron1993bootstrap}). As a result, for champion-challenger experiments, bias correction might increase the variance of the estimate so much that one can no longer confidently say if it is worthwhile to deploy the winner. 

Third, after analyzing the different solutions and their performance for different decision-making goals, we focus on three common scenarios that decision makers encounter: (1) null effects (Cohen's $d = 0$), (2) small treatment effects or small sample sizes (Cohen's $d = 0.2-0.5$) and (3) large treatment effects (Cohen's $d = 0.8$). Interestingly, across simulations, some of the proposed solutions in the literature perform inconsistently across these scenarios; many of the solutions hinge on a tuning parameter and its optimal value is context dependent. A method optimized for small effect regimes underperforms with null effects and vice-versa. These results persist when we extend the study to experiments with many treatments and in a simulation study calibrated to experiments from an online A/B testing platform \citep{berman2022false}. 

Consequently, our fourth contribution is a new inference method for experiments subject to the winner's curse. Building on empirical likelihood techniques \citep{owen2001empirical}, we develop an adaptive procedure (\Cref{alg:active_set}) that constructs valid confidence intervals across all settings we consider. We establish that this procedure achieves pointwise asymptotic coverage (\Cref{thm:pointwise-coverage-two-arms,thm:pointwise-coverage-many-arms}). Compared to resampling methods, our approach is both less sensitive to tuning parameters and more computationally efficient. In simulations, it consistently achieves nominal coverage across all scenarios.

\section{Related literature}
We distinguish among three different targets that a decision-maker might have with respect to the winner's curse and note that the literature has not identified an optimal method for many of these scenarios. This framework allows us to identify optimal correction procedures for realistic scenarios decision-makers might encounter.

We contribute to several interconnected literature streams on the winner's curse, post-selection inference, and optimal decision-making in experiments. \citet{SmithWinkler2006} introduced the optimizer's curse, showing that point estimates of selected alternatives are systematically upward biased. \citet{thaler2025} review similar reversion-to-the-mean phenomena in behavioral economics. In statistics, the selective inference literature tackles a related problem of conducting valid inference after model selection, where the choice of model parallels our setting of choosing the winner \citep{Kuchibhotla2022}. We extend this literature by distinguishing three objectives (global winner's curse, selected winner's curse, and regret) and showing that the optimal correction depends on the decision-maker's objective.

A related literature studies inference for rankings \citep{Bechhofer1954}. \citet{Gu2023} show that ranking and best-arm identification are directly connected and propose an empirical Bayes procedure. \citet{Mogstad2024} develop a resampling-based method for constructing statistical tests on ranks. While related, this literature focuses on identifying the best arm rather than estimating its value. Identifying the best arm addresses regret, but practitioners also need accurate estimates of the selected treatment's effect size to justify deployment and set expectations.

Estimating the maximum of means presents a parameter boundary problem: standard inference methods fail at the kink where $\mu_1 = \mu_2$.
\citet{andrews2000inconsistency} shows that standard resampling methods become inconsistent when parameters lie at boundaries, and proposes pre-test, shrinkage, and subsampling solutions. \citet{fang2019inference} and \citet{hong2018numerical} develop derivative-based bootstrap procedures for directionally differentiable statistics. The post-selection inference literature distinguishes simultaneous from conditional (selective) inference \citep{Kuchibhotla2022}. \citet{andrews2024inference} develop conditional inference specifically for winners, proposing a hybrid procedure that bounds interval length near ties. For bias correction, \citet{efron1993bootstrap} establishes the classical non-parametric bootstrap approach, and \citet{xu2025} apply it to the selected winner's curse. Our simulations reveal that these resampling methods, while asymptotically valid, are sensitive to tuning parameters in finite samples. The hybrid procedure proposed by \citet{andrews2024inference} is more stable but produces overly conservative confidence intervals in our simulation study. 

On the decision-making side, \citet{simester2020efficiently} and \citet{Feit2019} study experimental design for treatment selection, and \citet{bastani2025beating} propose inference-aware policy optimization. The minimax-regret literature (e.g., \citet{song2014local}, \citet{JooChiong2025}, and \citet{joo2025}) shows that the plug-in estimator is locally asymptotically minimax optimal; our regret analysis corroborates this finding. 

On the inference side, we contribute to the empirical likelihood (EL) literature \citep{owen2001empirical} by using chi-bar-squared critical values \citep{Dykstra1991} to achieve valid coverage of the confidence interval at the kink. Our adaptive procedure then builds on the pre-tests suggested by \citet{andrews2000inconsistency}, and we show that it provides pointwise asymptotic coverage of the confidence interval (\Cref{thm:pointwise-coverage-two-arms,thm:pointwise-coverage-many-arms} in \Cref{app:EL-adaptive}) both away from and at the kink. The EL approach is both less sensitive to tuning parameters and less computationally burdensome than resampling methods and thus offers both an efficient and robust alternative for constructing confidence intervals.

\section{Decision making scenarios with A/B testing}\label{sec:Three-Definitions}

Although A/B testing uses relatively standardized methods and analysis, it is applied to make decisions in distinct decision making scenarios, which imply that the same numerical results might yield different decisions under different scenarios and objectives of the experimenter. We provide two vignettes which we use to illustrate different solutions for the winner's curse below:

\paragraph{Vignette 1: Champion-Challenger with Optimized Default} 

An experimenter in this scenario already has a previously well-tested and optimized version which is the current default. They are testing new variations, nudges or interventions, to see if they improve on the existing default. Because the existing default is already optimized, it is likely that most new proposed challenger interventions will not improve on the champion and will have null effects. The goal of the experimenter is to select a new intervention that improves over the default with high confidence, and in addition properly estimate the amount of improvement over the default to assess whether it justifies the additional cost of switching to the challenger.

\paragraph{Vignette 2: One-shot pick the best}
An experimenter in this scenario runs a mega-study comparing many interventions, with a goal of selecting the best one going forward \citep{milkman2021megastudies}. The results of this study will be published and used to drive future public policy. It is anticipated that the best interventions will have a non-null effect compared to the other interventions \citep{Milkman2024Nature,Milkman2025PNAS}. The goal of the experimenter is to have a precise estimate of the expected outcome of the intervention to be deployed, in order to assess whether the outcomes outweigh the costs. This setting also appears in the behavioral economics literature \citep{thaler2025}, which emphasizes reversion to the mean for "winners" where the chosen winner may perform worse after being selected.

Although these two scenarios appear similar a priori, we present below how they impact decision-making when accounting for the winner's curse.

\subsection{Defining the winner's curse}

\definecolor{ECBlue}{RGB}{224,232,242}   
\definecolor{ECSand}{RGB}{241,235,225}   
\definecolor{ECSage}{RGB}{228,238,231}   
\definecolor{ECHeader}{RGB}{245,245,245}
\definecolor{ECNA}{RGB}{250,250,250}

\newcommand{\mCF}[1]{\cellcolor{ECBlue}\textbf{#1}}
\newcommand{\mBoot}[1]{\cellcolor{ECSand}\textbf{#1}}
\newcommand{\mPlugin}[1]{\cellcolor{ECSage}\textbf{#1}}
\newcommand{\mNA}[1]{\cellcolor{ECNA}#1}
\definecolor{ECLavender}{HTML}{D9D2E9} 
\definecolor{ECPeach}{HTML}{FCE5CD}

\newcommand{\mEL}[1]{\cellcolor{ECLavender}\textbf{#1}}
\newcommand{\mAKM}[1]{\cellcolor{ECPeach}\textbf{#1}}

\definecolor{ECRose}{HTML}{F4CCCC}     

  \newcommand{\mBayes}[1]{\cellcolor{ECRose}\textbf{#1}}

\newcommand{\twolinecell}[2]{
  \begingroup
  \setlength{\tabcolsep}{0pt}
  \renewcommand{\arraystretch}{0.85}
  \begin{tabular}[c]{@{}c@{}}
    #1\\[-0.55ex]#2
  \end{tabular}
  \endgroup
}

We present three alternative definitions for the magnitude of the winner's curse that have been used in the literature. We illustrate how they affect decision making in our three illustrative scenarios. We focus our analysis on experiments with two treatments (arms), and note that they naturally extend to more than two.

An experimenter runs an A/B test with two treatments/variations $k \in \{1,2\}$. The observed outcome of person $i$ who is exposed to treatment $k$, $y_{ik}$, is drawn from a distribution with mean $\mu_{k}$ and variance $\sigma^2$ where $\sigma$ is assumed known for simplicity. When treatment effects are null, we assume $\mu_1 = \mu_2$. When one treatment is better than the other, we assume it is treatment one, i.e.,  $\mu_1 = \mu_2 + \Delta \mu$ with $\Delta \mu >0$.

For each treatment group, the experimenter computes an estimate $\hat{\mu}_k$ of the true population mean $\mu_k$. A common way to compute the estimate is by taking the sample mean $\hat{\mu}_k = \frac{1}{n_k} \sum_i y_{ik}$, where $n_k$ is the number of individuals exposed to treatment $k$. However other methods (e.g., OLS regression) are also often used.

After the experiment, the experimenter picks one treatment to deploy in the future. If they pick the treatment with the highest $\hat{\mu}_k$ (denoted with $\hat{k}$), and use the value of $\hat{\mu}_{\hat{k}}$ (the plug-in estimator) to predict the future outcomes, then this value is expected to be upwards biased. Because the maximum function is convex, Jensen's inequality tells us that $E[\max(\hat{\mu}_1,\hat{\mu}_2) ] \ge \max(E[\hat{\mu}_1],E[\hat{\mu}_2)] ] = \max(\mu_1, \mu_2)$.
There is also some chance that the chosen treatment is not the truly optimal one, i.e., that $\hat{k} \neq \argmax(\mu_1, \mu_2)$.

This upward bias leads to a definition of the size of the winner's curse as the difference between the best estimated mean outcome and the truly best mean outcome in the population:
\begin{equation}
WC_{global} = \max(\hat{\mu}_1, \hat{\mu}_2) - \max(\mu_1, \mu_2).
\end{equation}
In other words, this estimate answers ''If one implements the treatment with the highest estimated mean, how far is the estimate from the truly best outcome''.

Although estimating (and minimizing) $WC_{global}$ is desirable, the true values of $\mu_k$ are unobserved and may not be precisely estimated. However, one can pick a benchmark to compare to (for example, a default control group). Let that benchmark be denoted as treatment $0$, we can rewrite 
\begin{equation}
WC_{global} = \max(\hat{\mu}_1, \hat{\mu}_2) - \hat{\mu}_0 + (\hat{\mu}_0 - \max(\mu_1, \mu_2))
\end{equation}
Minimizing $\max(\hat{\mu}_1, \hat{\mu}_2) - \hat{\mu}_0$ is equivalent to minimizing $WC_{global}$ because the last term in parenthesis is fixed and does not depend on the identity of the estimated winner $\hat{k}$ or the value $\max(\hat{\mu}_1, \hat{\mu}_2)$. This approach to measuring or correcting the winner's curse, which focuses on comparing to a control, addresses the objective of the experimenter in Vignette 1, by comparing to a champion treatment. An empirical application which focuses on this type of objective is the inference-aware policy optimization of \cite{bastani2025beating}. The focus in this analysis is often on having an improvement with a statistical guarantee of a positive effect, which implies that performing precise inference on $WC_{global}$ is desirable.

The literature on the winner's curse often focuses on a different objective, answering ''For the treatment selected for implementation, how far is the estimate from its true unobserved mean''. The winner's curse of the \emph{selected} treatment is defined as:
\begin{equation}
    WC_{select} = \hat{\mu}_{\hat{k}} - \mu_{\hat{k}} = \max(\hat{\mu}_1, \hat{\mu}_2) - \mu_{\argmax(\hat{\mu}_1, \hat{\mu}_2)}
\end{equation}
This metric is useful in Vignette 2, where a policy maker is interested in estimating (or minimizing) $WC_{select}$ in order to predict their returns when implementing the winning treatment \citep{xu2025}.

These two definitions lead to the third, which is regret, commonly used in decision theory \citep{bell1982regret}. A decision maker minimizing regret is not interested in estimates of outcomes, but only true outcomes. Hence, they might ask ''if I implement treatment $\hat{k}$, how much do I lose because I didn't implement the true optimal treatment?''. This question yields the definition of ex-post regret:
\begin{equation}
regret =  \max(\mu_1, \mu_2) -\mu_{\hat{k}}
\end{equation}

We observe that regret only depends on the selected arm $\hat k$ and not its estimate $\hat\mu$. Thus, the selecting the best arm is necessary and sufficient for zero ex-post regret. Further, since only arm selection matters, we can focus our efforts in estimating only $\hat k$ instead of both $\hat\mu_k$ and $\hat k$. This is why we expect the plug-in estimates to perform well in minimizing regret (as they are using the whole data set to just figure out $\hat k$).

\Cref{fig:WC-Def-Paper} illustrates the difference between these definitions. Noting that $\hat{\mu}_{\hat{k}} = \max(\hat{\mu}_1, \hat{\mu}_2)$, we can combine these definitions to write the identity:
\begin{equation}\label{eq:Winners-Curse-Identity}
   WC_{select} -WC_{global} =  regret.
\end{equation}
This identity implies a tradeoff between minimizing $WC_{global}$ and $WC_{select}$ if one keeps regret constant.

\begin{figure}[h]
    \begin{center}
    \begin{tikzpicture}[>=stealth, thick]

        \draw[->] (0,0) -- (6,0) node[right] {Arm choice}; 
        \draw[->] (0,0) -- (0,5) node[above] {Arm value}; 

        \coordinate (K_hat_tick) at (2,0);
        \coordinate (K_star_tick) at (5,0);
        
        \coordinate (Mu_hat_K_hat) at (2, 3);
        \coordinate (Mu_K_hat) at (2, 1.5);  
        \coordinate (Mu_K_star) at (5, 4.5); 
        \coordinate (Mu_hat_K_star) at (5, 2.5);

        \draw (K_hat_tick) -- ++(0,-0.15) node[below] {$\hat{k}$};
        \draw (K_star_tick) -- ++(0,-0.15) node[below] {$k^*$};

        \fill (Mu_hat_K_hat) circle (2pt);
        \fill (Mu_K_hat) circle (2pt);
        \fill (Mu_K_star) circle (2pt);
        \fill (Mu_hat_K_star) circle (2pt);

        \node[below right] at (Mu_K_hat) {$\mu_{\hat{k}}$};
        \node[left] at (Mu_hat_K_hat) {$\hat{\mu}_{\hat{k}}$};
        \node[right] at (Mu_K_star) {$\mu_{k^*}$};
        \node[left] at (Mu_hat_K_star) {$\hat{\mu}_{k^*}$};

        \draw[-] (Mu_K_hat) -- (Mu_K_star) node[midway, below, sloped, font=\small] {Regret};

        \draw[-] (Mu_hat_K_hat) -- (Mu_K_star) node[midway, above, sloped, font=\small] {Global};

        \draw[-] (Mu_hat_K_hat) -- (Mu_K_hat) node[midway, left, font=\small] {Select};

    \end{tikzpicture}
    \caption{Three definitions of winner's curse}
    \label{fig:WC-Def-Paper}
\end{center}

\footnotesize
Note: The selected arm is $\hat{k}$ while the truly best arm is $k^*$. Although the true means ensure $\mu_{\hat{k}}<\mu_{k^*}$ in the population, the estimated means in the sample may lead to
$\hat{\mu}_{\hat{k}} > \hat{\mu}_{k^*}$
The lines illustrate the values of $WC_{global}$, $WC_{select}$ and the Regret and are linked in our identity \Cref{eq:Winners-Curse-Identity}.
\end{figure}
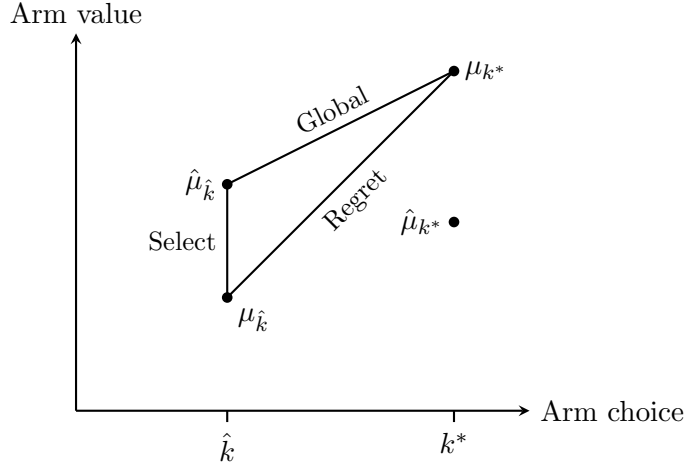

Given these three definitions and their applicability in different experimentation scenarios and objectives, we investigate if the correction procedures offered in the literature properly help the experimenter minimize the winner's curse. Because these definitions do not coincide, we would not expect one correction procedure to be the best for all scenarios and all objectives.\footnote{There is no tradeoff when all of the values are zero, which means both the correct treatment has been selected and the value of the treatment has been correctly estimated.} 
Also, we note that the correction procedures sometimes target different objectives such as correcting bias or conducting inference, and these are sometimes conflicting. For example, in the case of $WC_{select}$, a debiasing procedure will aim to minimize $E[WC_{select}]$, which will allow an experimenter in Scenario 2 to have the most precise prediction of the future performance of the selected intervention.
For scenario 1, however, debiasing estimates is expected to increase their variance; a larger variance makes statistical inference about the probability of improving over the optimized default harder. Some methods might also be too conservative, over estimate the variance, and  lead to overcovering confidence intervals, which may lead managers to decide against deployment despite it being beneficial.

In our analysis, we will compare the different proposed methods on both their ability to debias, and on providing tight confidence intervals. We will also consider a mean square error (MSE) metric that trades off bias and variance, and allows a manager to consider how much might be lost by debiasing or choosing a method that is overly conservative. Table \ref{tab:Winner-Curse-Def} summarizes the different definitions and provides the notation for each one of the metrics we will analyze. 

\begin{table}[ht]
\centering
\caption{Definitions of Winner’s Curse and Associated Statistical Properties}
\label{tab:Winner-Curse-Def}
\renewcommand{\arraystretch}{1.15}

\begin{tabular}{c|c|c|c}
\hline
\rowcolor{ECHeader}
 & \textbf{Bias}  &\textbf{MSE} & \textbf{Coverage}\\
\hline
WC Global  & $E[\hat\mu_{\hat k} -  \mu_{k^*}]$ & $E[\left(\hat\mu_{\hat k} -  \mu_{k^*}\right)^2]$ & $Pr\left(\mu_{k^*} \in CI_{\alpha}(\hat\mu_{\hat k})\right)$\\
\hline
WC Select & $E[\hat\mu_{\hat k} - \mu_{\hat k}]$ & $E[\left(\hat\mu_{\hat k} - \mu_{\hat k}\right)^2]$
& $Pr\left(\mu_{\hat k} \in CI_{\alpha}(\hat\mu_{\hat k})\right)$ \\
\hline
Regret & \multicolumn{3}{c}{$\mu_{k^*} - \mu_{\hat k}$} \\ 
\hline
\end{tabular}

\vspace{0.5em}
\footnotesize
Note: $CI_\alpha(\hat \mu$) indicates the $\alpha$-level two-sided confidence interval around the estimate $\hat \mu$. For regret, we do not target the MSE and Coverage as they are less interpretable and typically less useful in decision making. Minimizing regret is equivalent to selecting the correct arm $\hat k = \arg\max_k \mu_k$.
\end{table}

In the remainder of the paper we first describe and then analyze the different solutions offered to correct the winner's curse by comparing their performance across the different objectives. Our analysis shows that the appropriate correction procedure should be matched to the target objective of the decision maker and the expected effect sizes measured by Cohen's $d$. Table \ref{tab:wc_scenarios} previews the recommended procedure for each scenario for reducing bias and Table \ref{tab:wc_scenarios_mse} contains our full recommendation for all seven objectives.

\begin{table}[!ht]
\centering
\caption{Winner’s Curse Scenarios and Recommended Methods}
\vspace{0.5em}
\label{tab:wc_scenarios}
\renewcommand{\arraystretch}{1.15}
\setlength{\tabcolsep}{8pt}
\begin{NiceTabular}{c|cc|c}
    \CodeBefore
      \rectanglecolor{ECBlue}{2-2}{4-2}
      \rectanglecolor{ECBlue}{2-3}{2-3}
      \rectanglecolor{ECSand}{3-3}{4-3}
      \rectanglecolor{ECSage}{3-4}{5-4}
    \Body
    \rowcolor{ECHeader}
      \textbf{Cohen's $d$}
      & \textbf{Bias (Select)}
      & \textbf{Bias (Global)}
      & \textbf{Regret} \\
    \hline
    $d = 0$
      & \mCF{CF}
      & \mCF{CF}
      & \mNA{Does not matter} \\
    $d = 0.2$
      & \Block{2-1}{\textbf{CF}}
      & \Block{2-1}{\textbf{Bootstrap}}
      & \Block{3-1}{\textbf{Plug-in}} \\
    $d = 0.5$
      & & & \\
    $d = 0.8$
      & \mPlugin{Plug-in}
      & \mPlugin{Plug-in}
      & \\
  \end{NiceTabular}

\vspace{0.5em}
\footnotesize
Note: Each cell reports the method recommended by the simulation results for the corresponding performance criterion and target. CF denotes cross-fitting. Bootstrap denotes bootstrap bias correction. Plug-in denotes the uncorrected plug-in estimator. \Cref{tab:wc_scenarios_mse} provides the full set of results for our simulation, reporting additional recommendations when the focus is not on the mean but on the MSE and the coverage of the confidence interval.
\end{table}

\section{Family of solutions}

This section introduces the families of solutions that address the winner's curse. We first formally set up the problem and outline the mathematical properties of the maximum of means, which is central to understanding the winner's curse. The discussion highlights the challenges with standard inferential techniques when dealing with the non-differentiability of the maximum function. Then, we describe the proposed solutions, and categorize them by their underlying machinery. One downside of these prior methods is that they require tuning parameters, which we show can lead to substantial degradation in performance if picked incorrectly. Hence, we introduce a novel techniques for constructing confidence intervals using empirical likelihood, which is less sensitive to underlying tuning parameters.

\subsection{Plug-in estimator}

The winner's curse can be recast as a problem studying the maximum of means. We focus on the two-arm setting and the parameter $\theta = \max(\mu_1, \mu_2) = \phi(\mu_1, \mu_2)$, where $\mu_1$ and $\mu_2$ are the means of two independent populations. We let $\{X_{1i}\}_{i=1}^{n_1}$ and $\{X_{2i}\}_{i=1}^{n_2}$ be independent and identically distributed (i.i.d.) samples from distributions with means $\mu_1, \mu_2$ and variances $\sigma_1^2, \sigma_2^2$ respectively. The sample estimators for the means are $\hat\mu_1 = \bar{X}_1$ and $\hat\mu_2 = \bar{X}_2$. 

A natural estimator for this statistic is the plug-in estimator $\hat{\theta}_n = \max(\bar{X}_1, \bar{X}_2) = \phi(\hat\mu_1,\hat\mu_2)$ that plugs the sample mean estimates into the max function. It is straightforward to show this function is continuous and Lipschitz in its arguments. Further, it is differentiable except at the \emph{kink} where $\hat\mu_1=\hat{\mu_2}$. The maximum function is convex, and by Jensen's inequality, we have $E[\hat{\theta}_n] \ge \max(\mu_1, \mu_2) = \theta$. This upward bias is precisely the "winner's curse".

In \Cref{app:Plugin-Asymptotics}, we establish the consistency of the estimator and derive its asymptotic distribution. We show that the asymptotic distribution is a maximum of two normal distributions, which behaves irregularly when at the kink ($\mu_1 = \mu_2$) due to the non-differentiability of $\phi$ (the maximum function) at this point when the means of the two groups are equal. As a result, standard approaches that rely on the Delta method  will fail at the kink \citep{hirano2012impossibility}, and the vanilla bootstrap methods for bias correction and uncertainty quantification are inconsistent for this statistic \citep{mammen1992bootstrap}.

\subsection{Outline of solutions to the Winner's Curse}

We categorize the several families of proposed solutions into five categories: sample-splitting, resampling, conditional inference, empirical Bayes, and empirical likelihood. In our simulations, we will evaluate the performance of these methods in terms of the different objectives outlined in \Cref{sec:Three-Definitions}.

\subsubsection{Sample-splitting}
Sample splitting decouples the selection of the winner from the estimation of its value. The data is partitioned into two independent folds; one fold is used to select the arm with the highest sample mean, and the other fold is used to estimate the mean of that selected arm. While this eliminates winner's curse, it reduces the effective sample size and can be inefficient. Sample splitting methods are reviewed in \citet{Kuchibhotla2022} and are considered as a general, model-agnostic method for de-biasing. However, the splitting of the data will widen the confidence interval and accurate selection of the winner is not guaranteed. 

Cross-fitting can recover some efficiency by averaging estimates over exchanged folds. The procedure recycles the sample-splitting folds for selection and estimation and then averages the results across the folds to provide a more data-efficient estimator.

\subsubsection{Resampling methods}

Resampling methods, such as the bootstrap, are widely used for bias correction and inference \citep{efron1993bootstrap}. 
\citet{fang2019inference} show that a statistic needs to be fully Hadamard differentiable for the bootstrap to be consistent. In \Cref{app:Plugin-Hadamard}, we show that the maximum operator $\max(\cdot)$ would not yield consistent bootstrap estimates because it is only directionally (and not fully) Hadamard differentiable.

\citet{andrews2000inconsistency} provides four classes of adjustments for resampling methods to be consistent for estimating a parameter on the boundary of its domain.
We adapt these methods for our setting where we have two arms and the statistic of interest is the maximum of the two arms. The boundary is at the kink when the two arms yield equal means.

First, \citet{andrews2000inconsistency} proposes using a pre-test that initially tests whether the two arms' means are equal or not for each bootstrap iteration. 
If the pre-test fails to reject the null hypothesis that the two means are equal, then the asymptotic distribution of the max of two normals is used. Otherwise, the limiting distribution of the higher arm is used. The critical value of the pre-test is chosen to scale with the number of observations which ensures asymptotic consistency.

Second, \citet{andrews2000inconsistency} proposes a shrinkage parametric bootstrap. The parametric approaches uses the limiting distribution of the problem. Like in the first approach, for each bootstrap iteration, a pre-test is run to see if the means are equal.
If the test fails to reject, the means are pooled and draws from a normal distribution with the pooled means are used. Otherwise, the maximum of the draws of two normals centered around the sample means of each arm are used. The pre-test critical value is also chosen to ensure asymptotic consistency. 

Intuitively, the pre-test in both approaches will ''snap'' the estimates to the kink if the sample means are close enough and smaller than the pre-test critical value. In practice, the choice of the pre-test value is a tuning parameter, and the adjustment leads to pointwise asymptotically consistent performance when at or away from the kink, but possibly worse performance in finite samples when close to the kink.

As the third and fourth solutions, \citet{andrews2000inconsistency} proposes two subsampling procedures: the $m$-out-of-$n$ subsampling estimator \citep{Politis1994} and the $m$-out-of-$n$ or rescaled bootstrap. Both use a smaller subsample size to smooth out the limit distribution and can be generally applied to i.i.d. settings. The subsampling size, $m$, acts as the tuning parameter and as long as $m\to\infty$ and $m/n\to 0$, the estimator is asymptotically consistent. As $m$ gets close to $n$, for the $m$-out-of-$n$ bootstrap, we recover the standard nonparametric bootstrap.

\citet{fang2019inference} and \citet{hong2018numerical} offer two additional resampling procedures based on the Hadamard derivative of $\phi$ and are similar to the pre-test procedures outlined by \citet{andrews2000inconsistency}. \citet{fang2019inference} uses a pre-test for each bootstrap iteration to determine if the means are equal and then uses the analytical Hadamard derivative at the kink if it fails to reject. Just like the pre-test approaches, the effective critical value of the test is a tuning parameter that needs to be chosen during implementation. \citet{hong2018numerical} expand on this approach and numerically compute the Hadamard derivative for each bootstrap iteration, and the step size of the numerical derivative becomes the tuning parameter. The numerical derivative approach is readily generalizable to many arms because it does not require a pre-test across the different arms. 

Separately, these resampling methods can be used for bias correction \citep{efron1993bootstrap}, which includes the winner's curse or bias of the plug-in estimator \citep{xu2025}. Any of these resampling methods that explicitly adjust for the kink will provide consistent bias correction. 

These approaches have two drawbacks: sensitivity to tuning parameters and heavy computational cost, with the latter especially acute for bootstrap-within-bootstrap confidence intervals around bootstrap bias corrections. While theory guarantees asymptotic consistency, in finite samples the choice of tuning parameter will substantially shift the sampling distribution. Each method has its own tuning parameter: $m$ in $m$-in-$n$ resampling, the critical value for pre-test and shrinkage estimators, and the step size for the numerical Hadamard derivative. Our simulations confirm that results depend meaningfully on these choices.

\subsubsection{Conditional inference}

A related literature examines constructing confidence intervals that are valid conditional on the selection event \citep{Kuchibhotla2022}. By explicitly conditioning on the fact that a specific arm was chosen (e.g., $\hat{\mu}_1 > \hat{\mu}_2$), these methods adjust the inference to account for the truncation of the sampling distribution induced by selection and are known in the literature as selective inference. In contrast, the sample-splitting and resampling approaches provide simultaneous inference, which covers both the arm selection and the mean estimation in our example. 

\citet{andrews2024inference} studies the problem of conducting inference conditional on choosing the best performing arm. In their framework, a manager faces a set of normally distributed arms and chooses the "winner" with the highest estimate. To remedy the winner's curse, they propose a conditional inference procedure based on the distribution that is conditional on the selection. By conditioning on the identity of the winning arm, they show that the residual distribution of the winner's estimator follows a truncated normal distribution rather than a standard normal. Inference proceeds by inverting this truncated distribution, which yields estimators that are median-unbiased and confidence intervals with valid coverage conditional on the specific arm selected. However, a limitation of this approach is that when the "best" arms are close in performance (close to the kink), the conditioning event becomes highly restrictive, leading to potentially infinite-length confidence intervals. 

To resolve this concern, \citet{andrews2024inference} propose a hybrid inference procedure that combines the conditional inference procedure with a simultaneous "projection" interval, which is an overly conservative interval valid for all arms simultaneously. Specifically, the projection interval is constructed to cover the true parameters of all potential winners with a high probability, ensuring unconditional validity regardless of which arm is ultimately selected. The hybrid procedure conditions on the selection event but restricts the parameter space to the projection interval, effectively using the simultaneous interval to bound the length of the conditional interval. The authors show that this approach prevents the confidence intervals from becoming excessively wide and ensures valid unconditional coverage. Our simulation study evaluates the performance of these two approaches for our outcomes of interest.

\subsubsection{Empirical Bayes}

\citet{SmithWinkler2006} propose a Bayesian correction in which the posterior mean of each arm is a convex combination of its sample mean and a prior mean. We implement an empirical Bayes version that shrinks each arm toward the grand mean, with the prior variance estimated from the data by method of moments:
\begin{equation}
    \hat{v}_k = \alpha_k \bar{Y}_k + (1-\alpha_k)\bar{\mu}, \qquad \alpha_k = \frac{\hat{\sigma}_\mu^2}{\hat{\sigma}_\mu^2 + \widehat{se}_k^2},
\end{equation}
where $\bar{\mu} = K^{-1}\sum_k \bar{Y}_k$ is the grand mean and $\hat{\sigma}_\mu^2$ is a method-of-moments estimator of the between-arm variance. The winner is selected as $\hat{k} = \arg\max_k \hat{v}_k$ and its shrunk mean $\hat{v}_{\hat{k}}$ is reported as the point estimate.

The intuition is that the naive estimator $\bar{Y}_{\hat{k}}$ is unbiased conditional on the true means but becomes upward biased once the selection step is introduced, because extreme sample means are disproportionately likely to come from arms whose noise realizations happened to be large. Shrinking toward the grand mean reduces these extremes: an arm whose observed mean is far above the others is partially attributed to noise rather than to a genuinely superior true mean, and its estimate is pulled back accordingly. The amount of shrinkage is governed by the signal-to-noise ratio $\alpha_k$. When sampling noise is large relative to the spread of true means, $\alpha_k$ is small and estimates are shrunk heavily; when the spread of true means dominates, $\alpha_k$ approaches one and the shrinkage vanishes. This approach is the Bayesian analogue of the James--Stein estimator \citep{JamesStein1961, EfronMorris1973}.

The Empirical Bayes approach is most advantageous when the true means are plausibly drawn from a common distribution and when $K$ is large enough to estimate $\sigma_\mu^2$ reliably. It is straightforward to implement, requires no resampling, and has no tuning parameter beyond the variance components, which are estimated from the data. Its main limitation in our setting is that the estimate of $\sigma_\mu^2$ is noisy when $K$ is small, and shrinkage toward the grand mean can introduce downward bias when one arm is genuinely far better than the rest.

\subsubsection{Empirical likelihood}

Empirical likelihood (EL) offers a computationally efficient, non-parametric approach to constructing confidence intervals. We incorporate EL for our problem and introduce an adaptive procedure based on pre-tests that accounts for the kink and ensures proper pointwise asymptotic coverage of its confidence interval. 

EL constructs confidence intervals by inverting a profile likelihood ratio test: For a candidate value $\theta$, one computes the maximum of the empirical likelihood subject to the constraint that the parameter equals $\theta$, and rejects when the resulting deviance exceeds a $\chi^2$ critical value \citep{owen2001empirical}. The likelihood ratio test holds under smooth, interior constraints. 

However, the constraint $\max(\mu_1, \mu_2) = \theta$ introduces a constraint at the boundary. When $\mu_1 = \mu_2 = \theta$, the constraint set has a kink, and the standard $\chi^2$ calibration for the critical value is incorrect. In \Cref{app:Chi-bar-square-adaption}, we show the profile deviance at the kink instead converges to a chi-bar-squared distribution, which is a mixture of $\chi^2$ variates, and accounts for the boundary constraint \citep{Dykstra1991}. We also show that applying the standard $\chi^2$ critical value at the kink leads to under-coverage, but using the chi-bar-squared critical value everywhere leads to over-coverage away from the kink.

In practice, we do not know if we are at the kink or not. We therefore propose an adaptive procedure based on pre-tests to adaptively use the $\chi^2$ or the chi-bar-squared critical value. In \Cref{app:EL-adaptive}, we show that this approach is asymptotically consistent and provides pointwise asymptotic coverage \Cref{lem:Consistency-2-arm} and \Cref{thm:pointwise-coverage-two-arms}  in \Cref{app:EL-adaptive}. 

\Cref{lem:Consistency-2-arm} establishes that the pre-test consistently classifies whether the data generating process is at the kink ($\mu_1 = \mu_2$) or away from it.\footnote{$\hat K_n$ is a binary variable that represents the sample classification (See \Cref{eq:K_n_definition}).}

\begin{quote}
\noindent\textbf{Lemma 1 (Consistency of the pre-test decision).}
Assume $n_1,n_2\to\infty$, $n_1/(n_1+n_2)\to\pi\in(0,1)$, $s_g^2\to_p \sigma_g^2\in(0,\infty)$.
If $\Delta=0$, then $\Pr(\hat{K}_n=1)\to 1$.
If $\Delta\neq 0$, then $\Pr(\hat{K}_n=1)\to 0$.
\end{quote}

\noindent The intuition behind the proof relies on the behavior of the scaled difference between the sample means. At the kink (when the true means are equal), this difference is driven purely by random noise and remains bounded. Because our pre-test threshold $\kappa_n$ grows with the sample size, it eventually exceeds this noise, correctly identifying the tie. Away from the kink, the true difference between the means causes the scaled statistic to grow very quickly (at a rate proportional to $\sqrt{n}$). Because our threshold is chosen to grow slower than this rate, the test statistic easily exceeds it, ensuring the pre-test correctly detects that the means are different.

Building on \Cref{lem:Consistency-2-arm}, \Cref{thm:pointwise-coverage-two-arms} shows that the adaptive EL confidence interval achieves correct pointwise asymptotic coverage.

\begin{quote}
\noindent\textbf{Theorem 1 (Pointwise asymptotic coverage of the adaptive EL CI).}
Let $\theta_0=\max(\mu_1,\mu_2)$ denote the true value. Under the same conditions as Lemma~1, for any fixed $\alpha\in(0,1)$, $\Pr\big(\theta_0\in \mathcal{C}^{EL,ad}_{1-\alpha}(\theta)\big)\to 1-\alpha$.
\end{quote}

\noindent The proof evaluates the coverage probability by splitting it into two scenarios based on the pre-test's decision. If the true means are equal (at the kink), the pre-test consistently identifies the tie. This triggers the use of the chi-bar-squared critical value, which matches the nonstandard distribution of the test statistic at the boundary, ensuring correct coverage. Conversely, if the true means differ (away from the kink), the pre-test consistently detects the difference. The procedure then applies the standard $\chi^2$ critical value, which provides correct coverage because the test statistic follows a standard, smooth distribution in this regime.

We then extend our analysis $J$ arms. We derive the associated chi-bar-squared critical value (\Cref{sec:Many-arms-chi-bar-squared-adaption}) and propose an active-set approach (\Cref{alg:active_set}), in which we only consider comparisons of the arms to the empirically best performing arm. \Cref{lem:Consistency-J-arm} and \Cref{thm:pointwise-coverage-many-arms} respectively show that the procedure is asymptotically consistent for identifying the active-set of best performing arms and provides pointwise asymptotic coverage (\Cref{sec:Adaptive-many-arms}).

The advantages of our EL method are twofold: compared to resampling techniques, its tuning parameter is more rigorously justified by theory and, in practice, demonstrates lower sensitivity to its tuning parameter.
In finite samples, the performance of $m$-out-of-$n$ resampling methods depends heavily on the chosen subsample size, $m$. As our simulations demonstrate, varying this tuning parameter causes the method to systematically over- or under-correct. Further, there is no established theoretical guidance for selecting an optimal $m$.
The EL approach requires a pre-test cutoff $\kappa_n$ that is guided by asymptotic theory, requiring $\kappa_n \to \infty$ and $\kappa_n = o(\sqrt{n})$. With theoretically guided cutoff $\kappa_n = \sqrt{\log\log(\min\{n\})}$, the EL method avoids the finite-sample volatility observed in resampling procedures in our simulation study; the tuning parameter for EL also only impacts the critical value and not the sampling distribution, which provides stability in its estimate. 

While our EL method uses a pre-test like the adaptive resampling techniques of \citet{andrews2000inconsistency} and \citet{fang2019inference}, the impact of the pre-test is fundamentally different. In adaptive bootstrap methods, the pre-test determines which distribution to simulate. A misclassification forces the estimator to draw from the wrong distribution, causing significant volatility in finite samples. In contrast, our EL procedure uses the pre-test only to select the critical value (toggling between $\chi_1^2$ and $\bar{\chi}^2$), while the underlying test statistic remains unchanged. Therefore, a pre-test misclassification in the EL framework merely applies a slightly different threshold to the correct test statistic and the EL approach is structurally more stable.

Like all pre-test procedures, our EL method achieves pointwise, rather than uniform, asymptotic coverage \citep{andrews2000inconsistency}. Because methods relying on a tuning parameter to adjust for the kink inherently suffer from this limitation, we advocate the EL approach primarily through its finite sample stability and its significant computational advantages. While resampling methods require computationally expensive bootstrap-with-bootstraps procedures to construct valid confidence intervals, our EL method completely sidesteps this burden. Instead, it relies on a highly efficient one-dimensional root-finding procedure, making it vastly more scalable in practice.

\section{Analysis of winner's curse correction methods}\label{sec:Results-Comparison}

Our Monte Carlo simulations are designed to compare the four categories of solutions at representative treatment effect values that match the three scenarios from \Cref{sec:Three-Definitions}. We focus on the seven different objectives: bias correction, MSE, and coverage for WC select and WC global as well as regret.

\subsection{Simulation setup}
\label{subsec:sim_settings}

\paragraph{Data generating process}
The simulation generates data from a two-arm randomized experiment with balanced assignment. In each Monte Carlo replication, we draw an i.i.d.~sample of size
$N$ with sample sizes $n_1=\lfloor N/2\rfloor$ and $n_2=N-n_1$, where $N \in \{20, 40, 60, 80, 100\}$.
Individual outcomes for unit $i$ are from 
$Y_i \overset{iid}{\sim} \mathcal{N}(\mu_k,\sigma^2)$, 
with $\mu_1 = 1$ and $\mu_2 = 1+\Delta\mu$. 
Throughout, we report Cohen's $d \equiv \Delta\mu / \sigma$ as a standardized effect-size summary, so that the grid spans $d \in \{0, 0.1, 0.2, 0.5, 0.8, 1.0\}$ when varying $\sigma$.
For each $(N, d ,\sigma)$ tuple, we run $R=20,000$ replications.
Resampling-based procedures use $B=200$ resamples per replication.

\paragraph{Objectives for comparison}
We evaluate all methods on the following criteria:
\begin{enumerate}
    \item \textbf{Bias of point estimates (Winner's curse).}
    In each replication, we compute arm sample means $\bar{Y}_1,\bar{Y}_2$, select the winner
    $\hat{k}=\arg\max_{k\in\{1,2\}}\bar{Y}_k$, and form the plug-in estimator $\hat{\phi}=\max\{\bar{Y}_1,\bar{Y}_2\}=\bar{Y}_{\hat{k}}$.
    
    We evaluate winner's-curse errors for the plug-in estimator as well as for each debiased estimator produced by the methods described below.
    We report errors under two targets: (i) the \emph{select} target $\mu_{\hat{k}}$ (the true mean of the selected arm) and (ii) the \emph{global} target $\mu_{k^*}$ where $k^*=\arg\max_{k\in\{1,2\}}\mu_k$.

    \item \textbf{Mean squared error (MSE).}
    For each method and each target, we compute the MSE by averaging squared estimation errors across replications.

    \item \textbf{Coverage of 95\% confidence intervals.}
    For methods that directly output 95\% intervals, we directly compute empirical coverage of the true value. For bootstrap-corrected estimators, we construct 95\% intervals using a \emph{bootstrap-of-bootstraps} procedure: within each Monte Carlo replication, we repeatedly resample from the realized dataset, re-run the full bias-correction procedure on each resample to obtain a bootstrap distribution of the corrected estimator, and then construct a percentile interval from that distribution. Coverage is reported as the fraction of replications in which the resulting interval contains the target parameter (select or global).

    \item \textbf{Regret.}
    We also measure regret, $\mathrm{Regret}=\mu_{k^*}-\mu_{\hat{k}}$.
    
\end{enumerate}

\paragraph{List of methods}

\Cref{tab:methods_summary} summarizes the procedures we analyze in our simulations. Here, we report results of analyses for the following methods:
(1) standard nonparametric bootstrap; 
(2) $m$-out-of-$n$ bootstrap with $m=\lfloor N^{\gamma}\rfloor$ and $\gamma\in\{0.45,0.95\}$; 
(3) parametric shrinkage bootstrap; 
(4) FS19 derivative bootstrap; 
(5) sample splitting; 
(6) cross fitting; 
(7) AKM24 hybrid;
(8) empirical likelihood (EL);  and 
(9) Bayesian shrinkage. 
\Cref{app:Simulation-Details} reports the full set of simulation results for all methods we analyze.

For the bootstrap-based bias-correction methods, we use two correction amounts that mirror our two evaluation targets (select vs.\ global).
Let $D=\{(T_i,Y_i)\}_{i=1}^N$ denote the realized dataset, and define the arm sample means
$\bar{Y}_k(D)$ and the plug-in winner estimator
$\hat{\phi}(D)=\max\{\bar{Y}_1(D),\bar{Y}_2(D)\}$.
Each bootstrap procedure generates bootstrap resamples $D^{*(b)}$ for $b=1,\ldots,B$ and computes the corresponding bootstrap winner
$\hat{\phi}^{*(b)}=\hat{\phi}(D^{*(b)})$, from which we form two bias estimates.

The \emph{global} correction  estimates the upward bias of the plug-in winner by centering each bootstrap winner at the original plug-in value,
\[
\widehat{WC}_{\mathrm{glo}}
=\frac{1}{B}\sum_{b=1}^B\big(\hat{\phi}^{*(b)}-\hat{\phi}(D)\big).
\]
The \emph{select} correction  re-centers each bootstrap winner around the mean of the arm selected within that bootstrap sample, evaluated on the original data,
\[
\widehat{WC}_{\mathrm{sel}}
=\frac{1}{B}\sum_{b=1}^B\big(\hat{\phi}^{*(b)}-\bar{Y}_{\hat{k}^{*(b)}}(D)\big),
\qquad
\hat{k}^{*(b)}=\arg\max_{k\in\{1,2\}}\bar{Y}_k(D^{*(b)}).
\]
For either choice $\diamond\in\{\mathrm{sel},\mathrm{glo}\}$, the corrected estimator has the common form $\hat{\phi}_{\diamond}=\hat{\phi}(D)-\widehat{WC}_{\diamond}$,
where $\widehat{WC}_{\diamond}$ is computed using $B=200$ bootstrap resamples.\\

\newcolumntype{Y}{>{\raggedright\arraybackslash}X}

\newcommand{\algcell}[1]{
  \parbox[t]{\linewidth}{
    \scriptsize
    \raggedright
    \setlength{\parindent}{0pt}
    \setlength{\parskip}{0pt}
    #1
  }
}

\newlength{\methodcolw}
\setlength{\methodcolw}{0.22\textwidth} 

\begingroup
\linespread{1.0}
\setlength{\tabcolsep}{3pt}
\renewcommand{\arraystretch}{1.02}
\setlength{\aboverulesep}{0pt}
\setlength{\belowrulesep}{0pt}

\begin{xltabular}{\textwidth}{@{}>{\raggedright\arraybackslash}p{\methodcolw}Y@{}}
\caption{Summary of methods analyzed in simulation}
\label{tab:methods_summary}\\
\toprule
\textbf{Method} & \textbf{Procedure} \\
\midrule
\endfirsthead

\multicolumn{2}{l}{\textit{Table \ref{tab:methods_summary} continued from previous page}}\\
\toprule
\textbf{Method} & \textbf{Procedure} \\
\midrule
\endhead

\midrule
\multicolumn{2}{r}{\textit{Continued on next page}}\\
\endfoot

\bottomrule
\endlastfoot
\textbf{(1) Standard nonparametric bootstrap (NPB)} &
\algcell{
\textbf{Input:} $D=\{(T_i,Y_i)\}_{i=1}^N$, $B=200$. Compute $\bar{Y}_k(D)$ and $\hat{\phi}(D)=\max_k \bar{Y}_k(D)$.\\
\textbf{For} $b=1,\dots,B$ \textbf{do}\\
\hspace*{1em}Sample $N$ rows from $D$ with replacement until both arms appear; call resample $D^{*(b)}$.\\
\hspace*{1em}Compute $\hat{\phi}^{*(b)}=\max_k \bar{Y}_k(D^{*(b)})$ and $\hat{k}^{*(b)}=\arg\max_k \bar{Y}_k(D^{*(b)})$.\\
\textbf{End for.}\\[2pt]
\textbf{Corrections:}\;
$\displaystyle
\widehat{WC}_{\mathrm{sel}}
=\frac{1}{B}\sum_{b=1}^B\Big(\hat{\phi}^{*(b)}-\bar{Y}_{\hat{k}^{*(b)}}(D)\Big),\quad
\widehat{WC}_{\mathrm{glo}}
=\frac{1}{B}\sum_{b=1}^B\Big(\hat{\phi}^{*(b)}-\hat{\phi}(D)\Big).$\\
\textbf{Output:}\; $\hat{\phi}_{\diamond}=\hat{\phi}(D)-\widehat{WC}_{\diamond}$, $\diamond\in\{\mathrm{sel},\mathrm{glo}\}$.
} \\\hline
\textbf{(2) $m$-out-of-$n$ bootstrap (w/ repl.)} &
\algcell{
\textbf{Input:} $D$, $B=200$, $\gamma\in\{0.45,0.95\}$. For each $\gamma$, set $m(\gamma)=\lfloor N^{\gamma}\rfloor$.\\
\textbf{For each} $\gamma\in\{0.45,0.95\}$ \textbf{do}\\
\hspace*{1em}\textbf{For} $b=1,\dots,B$ \textbf{do}\\
\hspace*{2em}Sample $m(\gamma)$ rows from $D$ with replacement until both arms appear; call it $D_{m(\gamma)}^{*(b)}$.\\
\hspace*{2em}Compute $\hat{\phi}_{m(\gamma)}^{*(b)}=\max_k \bar{Y}_k(D_{m(\gamma)}^{*(b)})$ and
$\hat{k}_{m(\gamma)}^{*(b)}=\arg\max_k \bar{Y}_k(D_{m(\gamma)}^{*(b)})$.\\
\hspace*{1em}\textbf{End for.}\\[2pt]
\hspace*{1em}\textbf{Corrections (for this $\gamma$):}\;
$\displaystyle
\widehat{WC}_{\mathrm{sel}}(\gamma)
=\frac{1}{B}\sum_{b=1}^B\Big(\hat{\phi}_{m(\gamma)}^{*(b)}-\bar{Y}_{\hat{k}_{m(\gamma)}^{*(b)}}(D)\Big),\quad
\widehat{WC}_{\mathrm{glo}}(\gamma)
=\frac{1}{B}\sum_{b=1}^B\Big(\hat{\phi}_{m(\gamma)}^{*(b)}-\hat{\phi}(D)\Big).$\\
\hspace*{1em}\textbf{Output (for this $\gamma$):}\; $\hat{\phi}_{\diamond}(\gamma)=\hat{\phi}(D)-\widehat{WC}_{\diamond}(\gamma)$,
$\diamond\in\{\mathrm{sel},\mathrm{glo}\}$.\\
\textbf{End for.}
} \\\hline
\textbf{(3) Parametric shrinkage bootstrap (PB)} &
\algcell{
\textbf{Input:} $D$, $B=200$, $c_{\eta}=1.1$. Compute $\hat{\sigma}_k$ and $\hat{se}(\bar{Y}_k)=\hat{\sigma}_k/\sqrt{n_k}$.\\
Compute $\hat{se}=\sqrt{\hat{se}(\bar{Y}_1)^2+\hat{se}(\bar{Y}_2)^2}$ and $\eta_N=c_{\eta}\hat{se}\sqrt{2\log\log(N)}$.\\
Shrink-to-tie means: if $|\bar{Y}_1-\bar{Y}_2|\le \eta_N$, set $\tilde{\tau}_1=\tilde{\tau}_2=\tfrac12(\bar{Y}_1+\bar{Y}_2)$; else set $(\tilde{\tau}_1,\tilde{\tau}_2)=(\bar{Y}_1,\bar{Y}_2)$.\\
\textbf{For} $b=1,\dots,B$ \textbf{do}\\
\hspace*{1em}Keep $\{T_i\}$ fixed; draw $Y_i^{*(b)}\sim\mathcal N(\tilde{\tau}_{T_i},\hat{\sigma}_{T_i}^2)$ independently to form $D^{*(b)}$.\\
\hspace*{1em}Compute $\hat{\phi}^{*(b)}=\max_k \bar{Y}_k(D^{*(b)})$ and $\hat{k}^{*(b)}=\arg\max_k \bar{Y}_k(D^{*(b)})$.\\
\textbf{End for.}\\[2pt]
\textbf{Corrections:}\;
$\displaystyle
\widehat{WC}_{\mathrm{sel}}
=\frac{1}{B}\sum_{b=1}^B\Big(\hat{\phi}^{*(b)}-\tilde{\tau}_{\hat{k}^{*(b)}}\Big),\quad
\widehat{WC}_{\mathrm{glo}}
=\frac{1}{B}\sum_{b=1}^B\Big(\hat{\phi}^{*(b)}-\hat{\phi}(D)\Big).$\\
\textbf{Output:}\; $\hat{\phi}_{\diamond}=\hat{\phi}(D)-\widehat{WC}_{\diamond}$, $\diamond\in\{\mathrm{sel},\mathrm{glo}\}$.
} \\\hline
\textbf{(4) FS19 derivative bootstrap (FS)} &
\algcell{
\textbf{Input:} $D$ and bootstrap arm means $\{\bar{Y}_1^{*(b)},\bar{Y}_2^{*(b)}\}_{b=1}^B$, $B=200$.\\
Compute $\hat{se}=\sqrt{\hat{se}(\bar{Y}_1)^2+\hat{se}(\bar{Y}_2)^2}$, $t=(\bar{Y}_1-\bar{Y}_2)/\hat{se}$, and $\kappa_n=\sqrt{\log(\min\{n_1,n_2\})}$. (Chosen following \cite{fang2019inference}; see that paper for details.)\\
\textbf{For} $b=1,\dots,B$ \textbf{do}\\
\hspace*{1em}$h_1^{*(b)}=\sqrt{n_1}(\bar{Y}_1^{*(b)}-\bar{Y}_1)$,\;
$h_2^{*(b)}=\sqrt{n_2}(\bar{Y}_2^{*(b)}-\bar{Y}_2)$.\\
\hspace*{1em}Set $d^{*(b)}=h_1^{*(b)}$ if $t>\kappa_n$; $d^{*(b)}=h_2^{*(b)}$ if $t<-\kappa_n$; else $d^{*(b)}=\max\{h_1^{*(b)},h_2^{*(b)}\}$.\\
\textbf{End for.}\\
Set $n_{\mathrm{eff}}=\min\{n_1,n_2\}$ and
$\displaystyle \widehat{WC}^{\mathrm{FS19}}=\Big(\frac{1}{B}\sum_{b=1}^B d^{*(b)}\Big)\big/\sqrt{n_{\mathrm{eff}}}$.\\
\textbf{Output:}\; corrected point estimate $\hat{\phi}(D)-\hat{WC}^{\mathrm{FS19}}$.
} \\ \hline
\textbf{(5) Sample split (SS)} &
\algcell{
\textbf{Input:} $D$, split proportion $1/2$. Randomly split indices into folds $A$ (train) and $B$ (eval). Select $\hat{k}_A=\arg\max_k \bar{Y}_k(D_A)$ and estimate $\hat{\theta}_{\mathrm{SS}}=\bar{Y}_{\hat{k}_A}(D_B)$.\\
\textbf{95\% CI:} $t$-interval based on the evaluation-fold variance of the chosen arm.
} \\\hline
\textbf{(6) Cross fit (CF)} &
\algcell{
\textbf{Input:} $D$, split proportion $1/2$. Split into folds $A,B$. 
Select on $A$, estimate on $B$: $\hat{k}_A=\arg\max_k \bar{Y}_k(D_A)$, \ $\hat{\theta}_{A\to B}=\bar{Y}_{\hat{k}_A}(D_B)$.
Select on $B$, estimate on $A$: $\hat{k}_B=\arg\max_k \bar{Y}_k(D_B)$, \ $\hat{\theta}_{B\to A}=\bar{Y}_{\hat{k}_B}(D_A)$.\\
Set $\hat{\theta}_{\mathrm{CF}}=\tfrac12(\hat{\theta}_{A\to B}+\hat{\theta}_{B\to A})$. \textbf{95\% CI:} let $\hat{\mathrm{Var}}_{B}=s^2_{\hat{k}_A,B}/n_{\hat{k}_A,B}$ and $\hat{\mathrm{Var}}_{A}=s^2_{\hat{k}_B,A}/n_{\hat{k}_B,A}$. $\hat{se}_{\mathrm{CF}}=\tfrac12\sqrt{(\hat{\mathrm{Var}}_{A}+\hat{\mathrm{Var}}_{B})}$, and report
$\hat{\theta}_{\mathrm{CF}}\pm t_{N-1,0.975}\hat{se}_{\mathrm{CF}}$.
} \\\hline
\textbf{(7) AKM24 Hybrid} &
\algcell{
\textbf{Input:} $x_W=\hat{\phi}(D)$, $x_L$, $sd_W$, $\alpha=0.05$, $\beta=\alpha/10$, $K=2$.\\
Let $c_\beta=\Phi^{-1}\!\left(\frac{1+(1-\beta)^{1/K}}{2}\right)$ and define
$L(\mu)=\max\{x_L,\mu-c_\beta sd_W\}$, $U(\mu)=\mu+c_\beta sd_W$.\\
\textbf{Point estimate:} solve for $\hat{\mu}_{\mathrm{hyb}}$ such that
$F_{\mathrm{TN}}(x_W;\mu,sd_W,L(\mu),U(\mu))=\tfrac12$.\\
\textbf{95\% CI:} set $\gamma=(\alpha-\beta)/(2(1-\beta))$ and find $(\mu^L,\mu^U)$ solving
$F_{\mathrm{TN}}(x_W;\mu^L,\cdot)=1-\gamma$ and $F_{\mathrm{TN}}(x_W;\mu^U,\cdot)=\gamma$.
} \\\hline
\textbf{(8) Empirical likelihood (EL)} &
\algcell{
\textbf{Input:} arm samples $y_1=\{Y_i:T_i=1\}$ and $y_2=\{Y_i:T_i=2\}$, $\alpha=0.05$.\\
\textbf{Adaptive cutoff:} compute $\hat{se}=\sqrt{\hat{se}(\bar{Y}_1)^2+\hat{se}(\bar{Y}_2)^2}$ and $\eta_N=c_{\eta}\hat{se}\sqrt{2\log\log(N)}$ (same rule as PB).\\
If $|\bar{Y}_1-\bar{Y}_2|\le \eta_N$, use chi-bar-squared cutoff $q_{\bar{\chi}^2}(1-\alpha)$; otherwise use $\chi^2_1(1-\alpha)$.\\
\textbf{Profile EL deviance:} for candidate $\theta$, enforce $\max(\mu_1,\mu_2)=\theta$ by profiling along two paths: (A) $\mu_1=\theta,\mu_2\le\theta$; (B) $\mu_2=\theta,\mu_1\le\theta$.\\
Compute one-sample EL deviances $D_j(\theta)=-2\log LR(y_j;\mu=\theta)$ (via \texttt{el.test}).\\
Set $\mathrm{dev}_A(\theta)=D_1(\theta)+\mathbf{1}\{\bar{y}_2>\theta\}D_2(\theta)$,\;
$\mathrm{dev}_B(\theta)=D_2(\theta)+\mathbf{1}\{\bar{y}_1>\theta\}D_1(\theta)$, and $\mathrm{dev}(\theta)=\min\{\mathrm{dev}_A(\theta),\mathrm{dev}_B(\theta)\}$.\\
\textbf{95\% CI:} invert the test; find the smallest/largest $\theta$ with $\mathrm{dev}(\theta)\le$ cutoff (1D root finding around $\hat{\phi}(D)$). 
} \\\hline
\textbf{(9) Bayesian shrinkage} &
  \algcell{
  \textbf{Input:} $D=\{(T_i,Y_i)\}_{i=1}^N$. Compute $\bar{Y}_k$,
  $\hat{se}_k=\hat{\sigma}_k/\sqrt{n_k}$, and $\hat{\phi}(D)=\max_k
  \bar{Y}_k$.\\
  \textbf{Mean:} $\bar{\mu}=\tfrac12(\bar{Y}_1+\bar{Y}_2)$.\textbf{Prior variance:} $\hat{\sigma}_\mu^2=\max\!\Big(0,\;\frac{(\bar{
  Y}_1-\bar{Y}_2)^2-\hat{se}_1^2-\hat{se}_2^2}{2}\Big)$.\\
  \textbf{Shrinkage weights:}
  $\alpha_k=\hat{\sigma}_\mu^2\big/(\hat{\sigma}_\mu^2+\hat{se}_k^2)$,
  $k=1,2$.\\
  \textbf{Posterior means:}
  $\hat{v}_k=\alpha_k\bar{Y}_k+(1-\alpha_k)\bar{\mu}$.\\
  \textbf{Selection \& point estimate:} $\hat{k}_{\mathrm{Bayes}}=\arg\max_k
  \hat{v}_k$. Output: $\hat{\phi}_{\mathrm{Bayes}}=\hat{v}_{\hat{k}}$.\\
  \textbf{95\% CI:} Let $\sigma_{\bar{\mu}}^2=\tfrac12\hat{\sigma}_\mu^2+\tfrac1
  4(\hat{se}_1^2+\hat{se}_2^2)$ and $\sigma_{\mathrm{post}}^2=(1-\alpha_{\hat{k}
  })\hat{\sigma}_\mu^2+(1-\alpha_{\hat{k}})^2\sigma_{\bar{\mu}}^2$. Report
  $\hat{v}_{\hat{k}}\pm z_{1-\alpha/2}\sqrt{\sigma_{\mathrm{post}}^2}$.
  } \\
\bottomrule
\end{xltabular}
\endgroup
\linespread{1.5}

\subsection{Results}

This section presents the results of our simulations and provides an analysis for our seven objectives of interest.

\begin{figure}[ht]
     \begin{center}
        \caption{Winner's Curse Bias (Global)}\label{fig:wc_global}
\includegraphics[width=0.9 \textwidth]{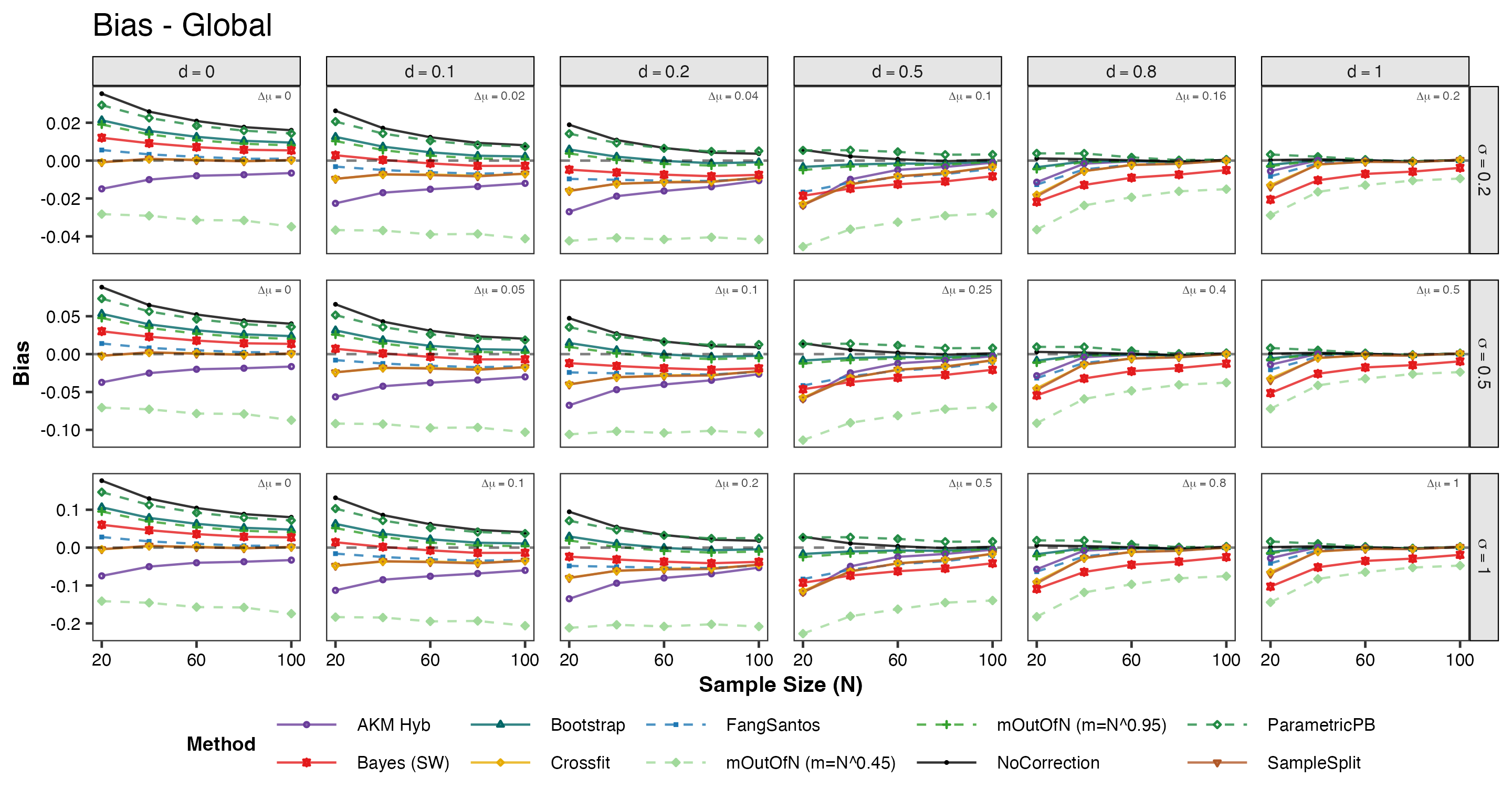}
    \end{center}
\footnotesize
    Notes: WC Global as a function of sample size. Panels vary by $\Delta\mu$ (columns) and $\sigma$ (rows). Methods: (1) standard nonparametric bootstrap; 
(2) $m$-out-of-$n$ bootstrap with $m=\lfloor N^{\gamma}\rfloor$ and $\gamma\in\{0.45,0.95\}$; 
(3) parametric shrinkage bootstrap (ParametricPB); 
(4) \citep{fang2019inference} Hadamard-derivative based bootstrap; 
(5) sample splitting; 
(6) cross-fitting; 
(7) AKM24 hybrid and (8) plug-in (no correction).
\end{figure}

\begin{figure}[!ht]
    \begin{center}
        \caption{Winner's Curse Bias (Select)}\label{fig:wc_select}
\includegraphics[width=0.9 \textwidth]{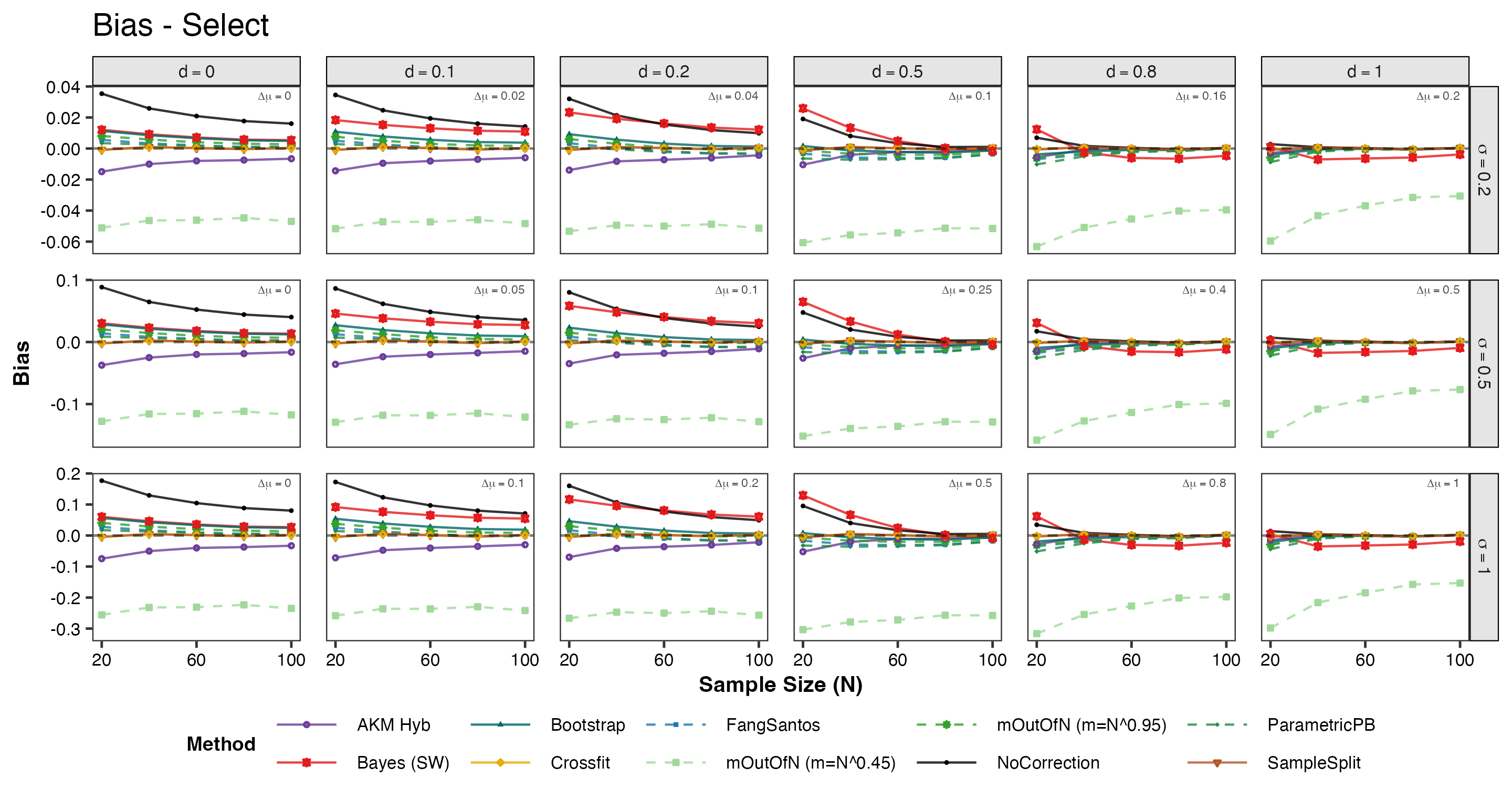}
    \end{center}
    \footnotesize
    Notes: WC Select bias as a function of sample size $N$. Panels vary by $\Delta\mu$ (columns) and $\sigma$ (rows). The legend matches that of Figure~\ref{fig:wc_global}.
\end{figure}

\begin{figure}[!ht]
    \begin{center}
        \caption{Mean Squared Error (Global)}
\includegraphics[width=0.9 \textwidth]{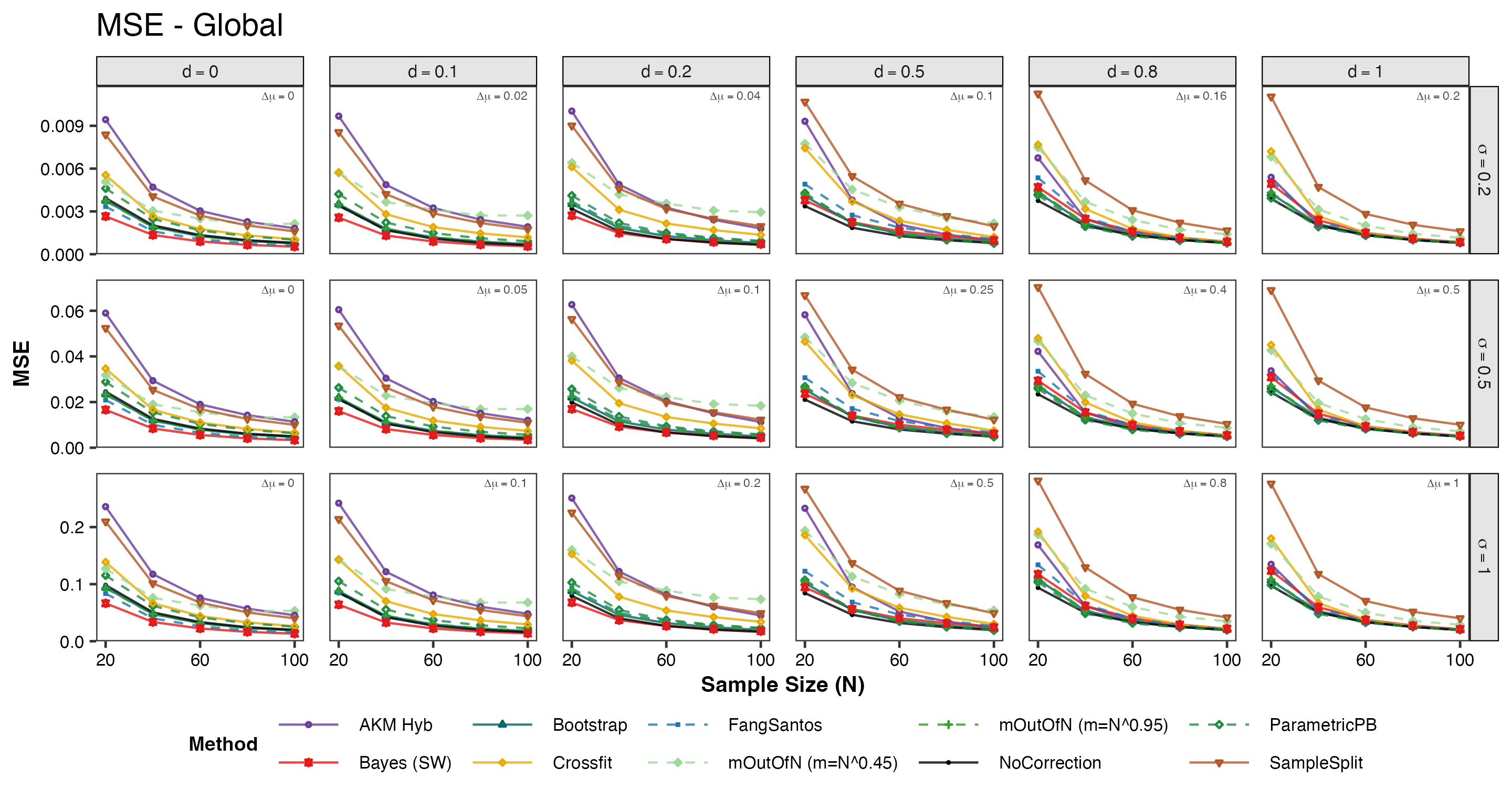}
    \label{fig:mse_global}
    \end{center}
     \footnotesize
    Notes: Mean squared error (MSE) for WC Global as a function of sample size $N$. Panels vary by $\Delta\mu$ (columns) and $\sigma$ (rows).
\end{figure}
\begin{figure}[!ht]
    \centering
        \caption{Mean Squared Error (Select)}
\includegraphics[width=0.9 \textwidth]{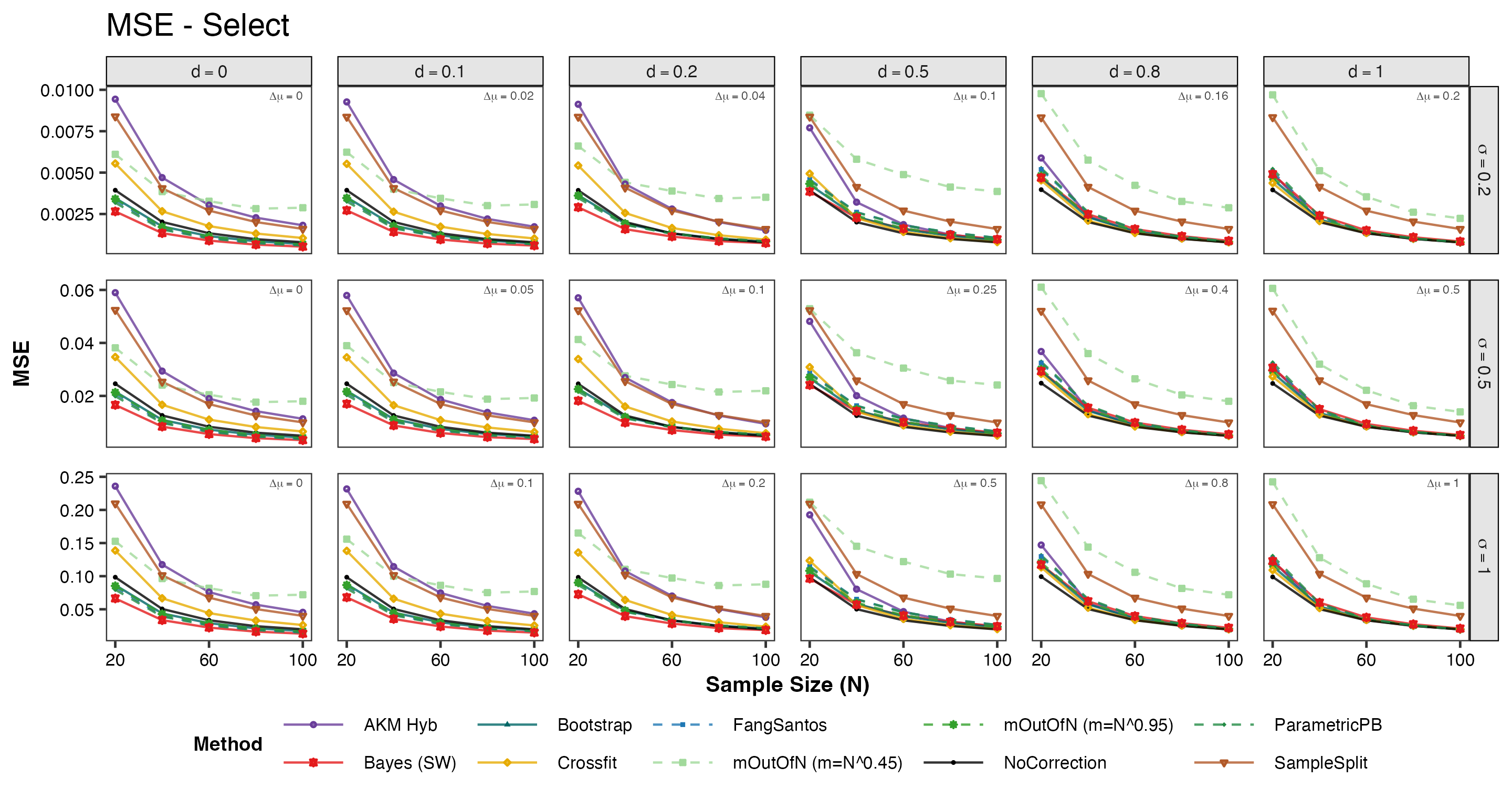}
    \label{fig:mse_select}
    
    \footnotesize
    Notes: Mean squared error (MSE) for WC Select as a function of sample size $N$. Panels vary by $\Delta\mu$ (columns) and $\sigma$ (rows).
\end{figure}

\begin{figure}[!ht]
    \begin{center}
    \caption{Coverage (Global)}
    \includegraphics[width=0.9 \textwidth]{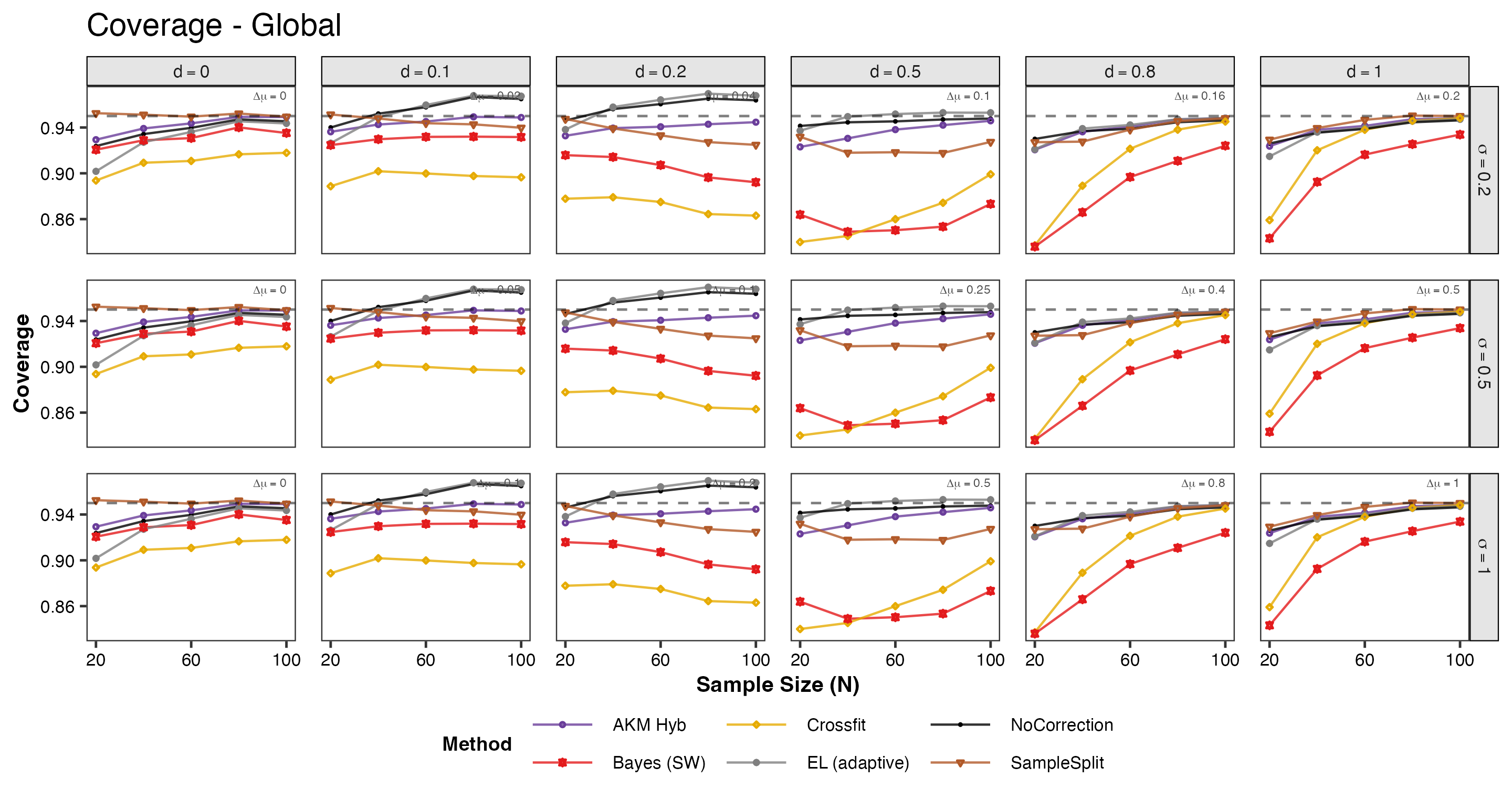}
    \label{fig:cov_global}
\end{center}
    \footnotesize
Notes: Empirical coverage of $\alpha$-level two-sided confidence intervals for $\mu_{k^*}$, i.e., $\Pr\left(\mu_{k^*}\in CI_\alpha(\hat\mu_{\hat k})\right)$, as a function of sample size $N$. Panels vary by $\Delta\mu$ (columns) and $\sigma$ (rows).   EL denotes empirical likelihood. 
\end{figure}

\begin{figure}[!ht]
    \begin{center}
    \caption{Coverage (Select)}
    \includegraphics[width=0.9 \textwidth]{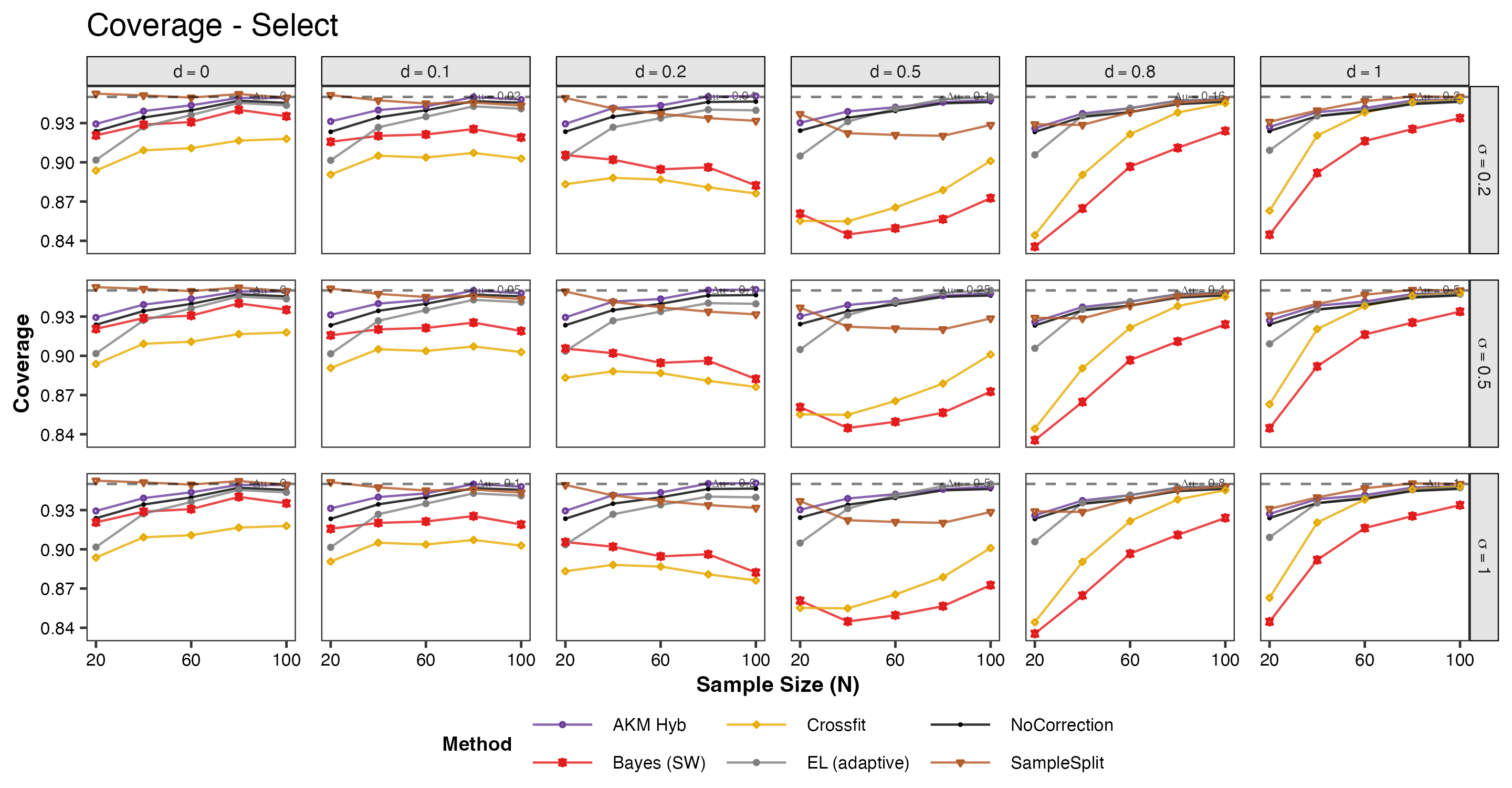}
    \label{fig:cov_select}
\end{center}
\footnotesize
     Notes: Empirical coverage of $\alpha$-level two-sided confidence intervals for $\mu_{\hat k}$, i.e., $\Pr\left(\mu_{\hat k}\in CI_\alpha(\hat\mu_{\hat k})\right)$, as a function of sample size $N$. Panels vary by $\Delta\mu$ (columns) and $\sigma$ (rows). 
\end{figure}

\normalsize

\paragraph{Winner’s-curse bias.}
Figures \ref{fig:wc_global} and \ref{fig:wc_select} report the average estimation error for the global and  select targets.
The magnitude of winner’s-curse bias is largest in regimes with low signal-to-noise ratios: small sample sizes, large noise $\sigma$, and near-ties (small $\Delta\mu$), and it diminishes as $N$ increases or the gap $\Delta\mu$ grows.
The direction of bias reflects how aggressively a method corrects for selection: The plug-in estimator  exhibits a large upward bias in tie/near-tie settings, while some resampling-based corrections can over-shoot and become negatively biased in these same regimes (most visibly for the more aggressive $m$-out-of-$n$ choice with smaller $m$).

Methods that explicitly decouple selection from estimation such as sample splitting and cross-fitting produce biases close to zero for the select target, especially under ties or at the kink, at the cost of reduced effective sample size.

\paragraph{Mean squared error (MSE)}
Figures~\ref{fig:mse_global} and \ref{fig:mse_select} summarize the MSE trade-off.
The ranking of methods depends on whether the problem is a tie/near-tie or a clear-winner regime.
In ties and near-ties, bias reduction is valuable, yet overly aggressive correction can inflate MSE through increased variance and/or negative bias.
This is most apparent when using a small $m$ in the $m$-out-of-$n$ bootstrap, which can lead to systematic negative bias and larger MSE.
In contrast, when $\Delta\mu$ is large (clear separation), the winner’s-curse bias is already small and additional correction is typically unnecessary; in these regimes, the plug-in estimator often dominates in the MSE because it avoids the extra variability introduced by correction or sample splitting.

\paragraph{Regret}
Regret is driven by selection mistakes choosing the inferior arm, rather than by post-selection bias correction. From the simulation, the expected regret has the closed form
\[
\mathbb{E}[\mathrm{Regret}]
=
|\mu_1-\mu_2|\,
\Phi\!\left(
-\frac{|\mu_1-\mu_2|}{\sigma\sqrt{1/n_1+1/n_2}}
\right),
\]
where $\Phi(\cdot)$ is the standard normal CDF. Procedures that devote only part of the sample to selection (e.g., sample splitting or cross-fitting) increase the variance of the comparison $\hat\mu_1-\hat\mu_2$, thereby raising the probability of selecting the inferior arm and increasing regret relative to selecting with the full sample. For a sample split with equal size folds, selection uses $n_k/2$ observations per arm, implying
\[
\mathbb{E}[\mathrm{Regret}_{\mathrm{SS}}]
=
|\mu_1-\mu_2|\,
\Phi\!\left(
-\frac{|\mu_1-\mu_2|}{\sigma\sqrt{1/(n_1/2)+1/(n_2/2)}}
\right)
=
|\mu_1-\mu_2|\,
\Phi\!\left(
-\frac{|\mu_1-\mu_2|}{\sigma\sqrt{2(1/n_1+1/n_2)}}
\right).
\]

Bootstrap-based bias corrections have the same regret as the plug-in rule because the selection decision still uses the full-sample comparison $\hat\mu_1-\hat\mu_2$. 

In clear-winner regimes, the plug-in rule is therefore preferred when regret is the primary objective, whereas under ties ($\mu_1=\mu_2$) regret is uninformative since $\mathrm{Regret} = 0$.

\paragraph{Coverage}
\Cref{fig:cov_global} and  \Cref{fig:cov_select} report the empirical coverage of $95\%$ confidence intervals for the global and select targets.
In exact ties, sample splitting yields close to nominal coverage for both targets, whereas cross-fitting can under-cover at small $N$.\footnote{\Cref{app:Cross-Fit-Coverage} provides a detailed discussion of why cross-fitting can under-cover in this setting.}
When Cohen's $d$ is moderate, the empirical likelihood and AKM24 are the preferred methods.
When Cohen's $d$ is large, the plug-in approach yields intervals whose coverage converges to the target level as $N$ grows.

The bootstrap-of-bootstrap is the most computationally expensive method, yet it underperforms relative to empirical likelihood. Empirical likelihood requires root-finding, which can be slower than closed-form alternatives, but it remains relatively fast compared to resampling methods.

\subsection{Scenario-based recommendations}

\begin{table}[!ht]
\centering
\caption{Winner's Curse Scenarios and Recommended Methods}
\vspace{0.5em}
\label{tab:wc_scenarios_mse}
\renewcommand{\arraystretch}{1.15}
\setlength{\tabcolsep}{6pt}
\resizebox{.98\textwidth}{!}{
 \begin{NiceTabular}{c|cc|cc|cc|c}                                             
  \CodeBefore                                                                   
    \rectanglecolor{ECBlue}{3-2}{5-2}                                           
    \rectanglecolor{ECRose}{3-4}{4-5}         
    \rectanglecolor{ECSage}{5-4}{6-5}                                           
    \rectanglecolor{ECPeach}{4-6}{5-6}                                          
    \rectanglecolor{ECLavender}{4-7}{5-7}                                       
    \rectanglecolor{ECSage}{4-8}{6-8}                                           
    \rectanglecolor{ECSand}{4-3}{5-3}                                           
  \Body      
  \rowcolor{ECHeader}  
   & \multicolumn{2}{c|}{\textbf{Bias}}
    & \multicolumn{2}{c|}{\textbf{MSE}} 
    & \multicolumn{2}{c|}{\textbf{Coverage}}            
    & \Block{2-1}{\textbf{Regret}} \\     
  \rowcolor{ECHeader}
    \textbf{Cohen's $d$} 
    & \textbf{Select}  
    & \textbf{Global}  
    & \textbf{Select}          
    & \textbf{Global} 
    & \textbf{Select}
    & \textbf{Global}  
    & \\                      
    \hline    
  $d=0$  
    & \Block{3-1}{\textbf{CF}}
    & \mCF{CF} 
    & \Block{2-2}{\textbf{Bayes}}
    &             
    & \mCF{SS}                                                 
    & \mCF{SS}
    & \mNA{Does not matter} \\ 
    \cline{3-3}\cline{6-8}
  $d = 0.2$  
    &                          
    & \Block{2-1}{\textbf{Bootstrap}}
    &           
    &    
    & \Block{2-1}{\textbf{AKM24}} 
    & \Block{2-1}{\textbf{EL}}
    & \Block{3-1}{\textbf{Plug-in}} \\
    \cline{4-5} 
  $d = 0.5$                  
    &
    &          
    & \Block{2-2}{\textbf{Plug-in}}
    &          
    &                          
    &
    & \\         
    \cline{2-3}\cline{6-7}
  $d = 0.8$  
    & \mPlugin{Plug-in}        
    & \mPlugin{Plug-in}  
    &     
    &
    & \mPlugin{Plug-in} 
    & \mPlugin{Plug-in}
    & \\      
        \hline\hline
  A/B Testing Platform
    & \mCF{CF}
    & \mCF{CF} / \mBayes{Bayes}
    & \mBayes{Bayes}
    & \mBayes{Bayes}
    & \mEL{EL}
    & \mEL{EL}
    & \mPlugin{Plug-in} \\
  \end{NiceTabular}
}
\vspace{0.5em}

\footnotesize
Note: Each cell reports the method recommended by the simulation for the corresponding metric. SS denotes sample splitting, CF denotes cross-fitting. Bootstrap denotes bootstrap bias correction. Plug-in denotes the uncorrected plug-in estimator. EL denotes the adaptive empirical likelihood procedure. AKM24 refers to \citep{andrews2024inference}. The A/B Testing platform is in \Cref{sec:AB-Testing-Platform}.
\end{table}

Table~\ref{tab:wc_scenarios_mse} summarizes a practical rule-of-thumb for experimenters who might need to choose a winner's curse correction method.

\begin{itemize}
    
    \item When $d=0$ (ties), selection-induced bias is most severe. In this case, which is common in Scenario 1,  cross-fitting is preferable for controlling the mean (bias) under both targets, while bootstrap-style correction can be recommended when prioritizing MSE.
    
    \item When $d \in (0.2, 0.5)$ but the gap is moderate, which might be common in Scenarios 2 and 3, cross-fitting or bootstrap correction are strong choices for controlling mean bias, whereas MSE performance often favors either Bayes method or the plug-in estimator depending on $(N,\sigma)$. For Regret, the plug-in estimator dominates.
    
    \item When $d\ge 0.8$, the problem is effectively “easy”: selection is reliable, the winner’s-curse bias is small, and the plug-in estimator is recommended across objectives, including regret. 
    
    \item For coverage, sample splitting is recommended under ties, empirical likelihood and AKM24 under moderate gaps, and plug-in intervals in the large-gap regime. Bootstrap methods are generally less recommended.

\end{itemize}

\subsection{Robustness/Sensitivity Analysis}
We provide additional simulation, scenario-based recommendations for (1) when arms are binomial and (2) for many (normals) arms. The details are in Appendix \ref{app:Simulation-Details}.

\subsubsection{Binomial arms}

To verify that the results do not depend  on the Gaussian noise assumption, we re-run the two-arm Cohen's $d$ grid with Bernoulli outcomes $Y_{ik}\sim\text{Bernoulli}(p_k)$. We sweep base rates $p_1 \in \{0.1, 0.3, 0.5\}$ and target effect sizes  $d \in \{0, 0.1, 0.2, 0.5, 0.8, 1.0\}$, with $p_2$ chosen so that the Cohen's $d$ matches the column label: $  d \;=\; \frac{p_2 - p_1}{\sqrt{\tfrac{1}{2}\big(p_1(1-p_1) +  p_2(1-p_2)\big)}}.$ Figures~\ref{fig:dgrid-bias-global}--\ref{fig:dgrid-cov-select}  report the same simulation with Bernoulli outcomes. Bias and MSE rankings are essentially unchanged. Coverage is the dimension where the noise distribution matters: Sample Splitting and \cite{andrews2024inference} under-cover at small $N$, while the empirical-likelihood interval remains  near nominal under both DGPs and is the safer default for binary outcomes.      
  
\subsubsection{Many arms}

We extend the simulation to $K \in \{2, 3, 4\}$ arms with true means drawn from a Gaussian prior,  $\tau_k \overset{iid}{\sim} \mathcal{N}(1, \sigma_0^2)$, and observations $Y_i \mid T_i = k \sim \mathcal{N}(\tau_k, \sigma^2)$. The prior standard deviation  $\sigma_0 \in \{0, 0.05, 0.1, 0.2, 0.5, 1\}$ controls treatment-effect  heterogeneity, $\sigma \in \{0.1, 0.2, 0.5, 1\}$, and $N = 100$ is fixed and balanced    across arms. Method-by-method generalizations are summarized in  Appendix Table~\ref{tab:k-arm-changes}. Figures~\ref{fig:prior_K_bias_global_app}--\ref{fig:prior_K_cov_select_app} report the  full results faceted by $\sigma_0$ (columns) and $\sigma$ (rows), with  $K$ on the $x$-axis.   The qualitative patterns extend consistently to $K > 2$.

\section{A/B Testing Platform Calibration}\label{sec:AB-Testing-Platform}

To study the effectiveness of the methods in real-world settings, we calibrate the simulation to data from the Optimizely A/B experimentation platform from \citet{berman2022false}. The simulation mimics how the methods would perform if deployed on the platform.
There are two main differences: (1) we consider many arms and (2) the underlying data generating process is more complicated and contains a mix of small effects and null effects.

 \begin{figure}[!ht]                                                             
      \centering                                                                
      \caption{Winner's Curse Bias (Binomial, Fixed Conversion Rate)}           
      \begin{subfigure}{0.48\textwidth}                                      
          \centering                                                            
          \includegraphics[width=0.9 \textwidth]{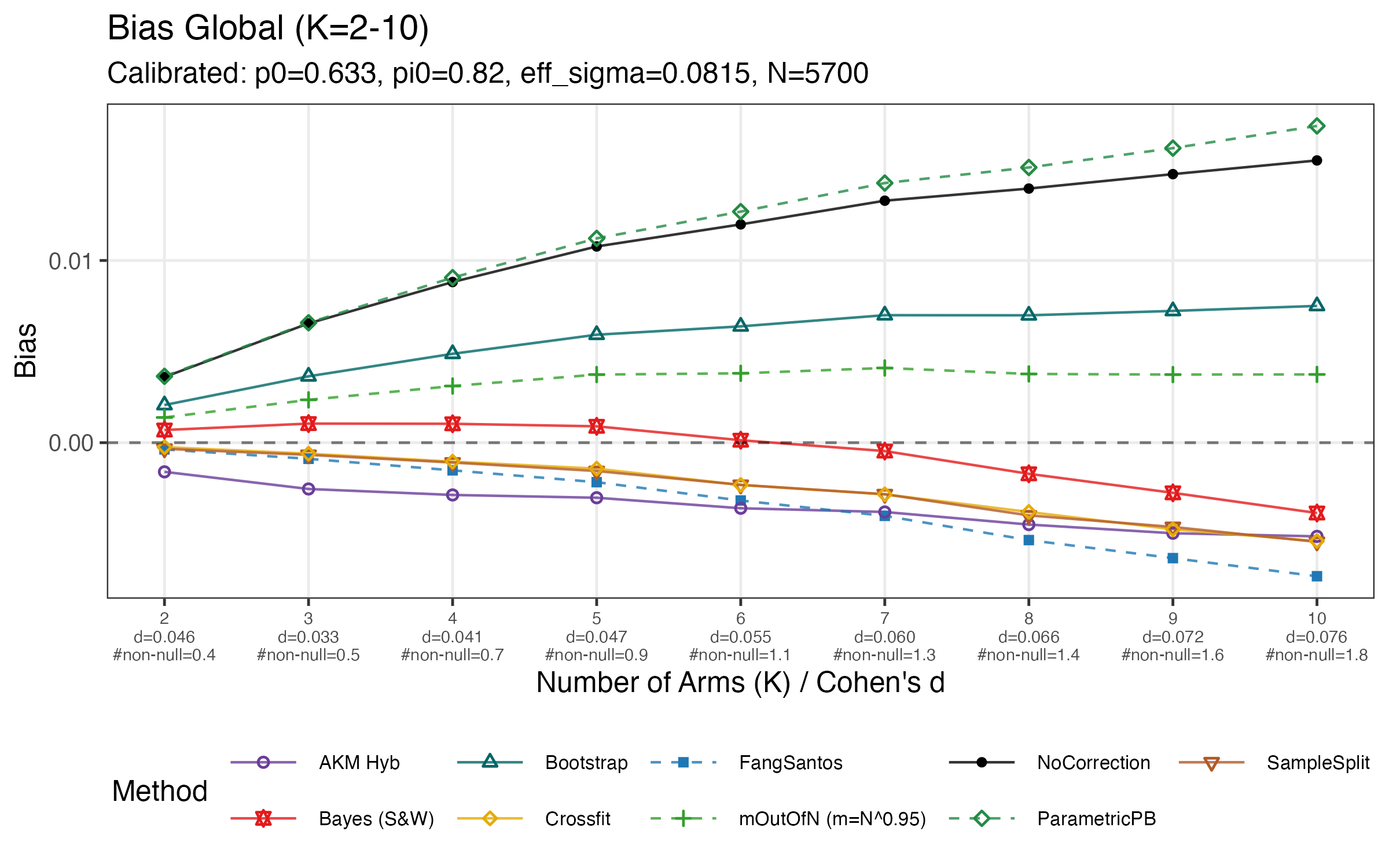}   
          \caption{Global}                                                      
          \label{fig:binom_fixed_bias_global}
      \end{subfigure}                                                           
      \hfill      
      \begin{subfigure}{0.48\textwidth}                                      
          \centering
          \includegraphics[width=0.9 \textwidth]{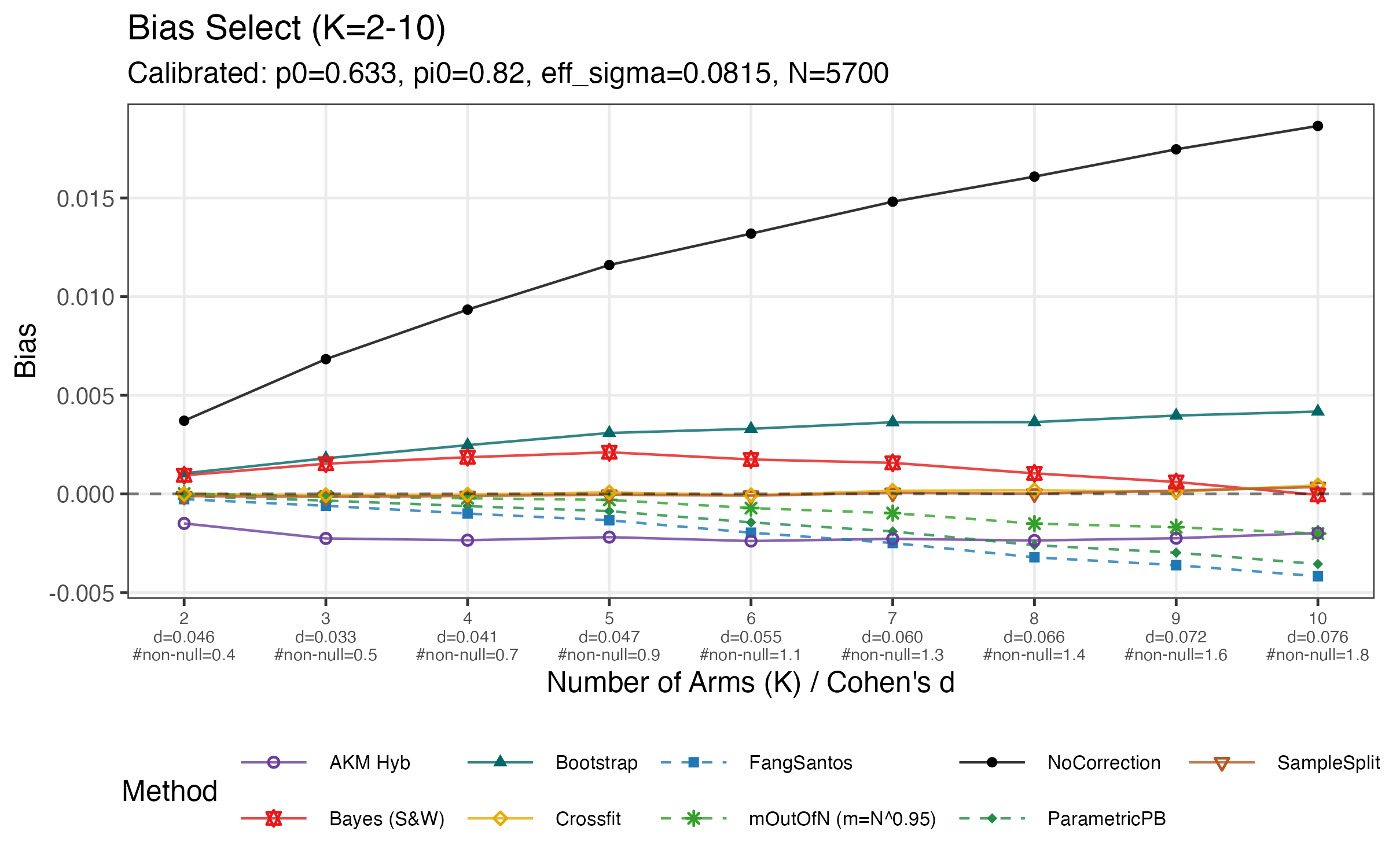}   
          \caption{Select}
          \label{fig:binom_fixed_bias_select}                                   
      \end{subfigure}
      \label{fig:binom_fixed_bias}                                              
      \footnotesize
      Notes: WC bias as a function of sample size $N$ under the binomial DGP
  with fixed base conversion rate. Panels vary by $K$ (number of arms).         
  \end{figure}
                                            
  \begin{figure}[h]
      \centering
      \caption{MSE (Binomial, Fixed Conversion Rate)}                           
      \begin{subfigure}[b]{0.48\textwidth}
          \centering                                                            
          \includegraphics[width=0.9 \textwidth]{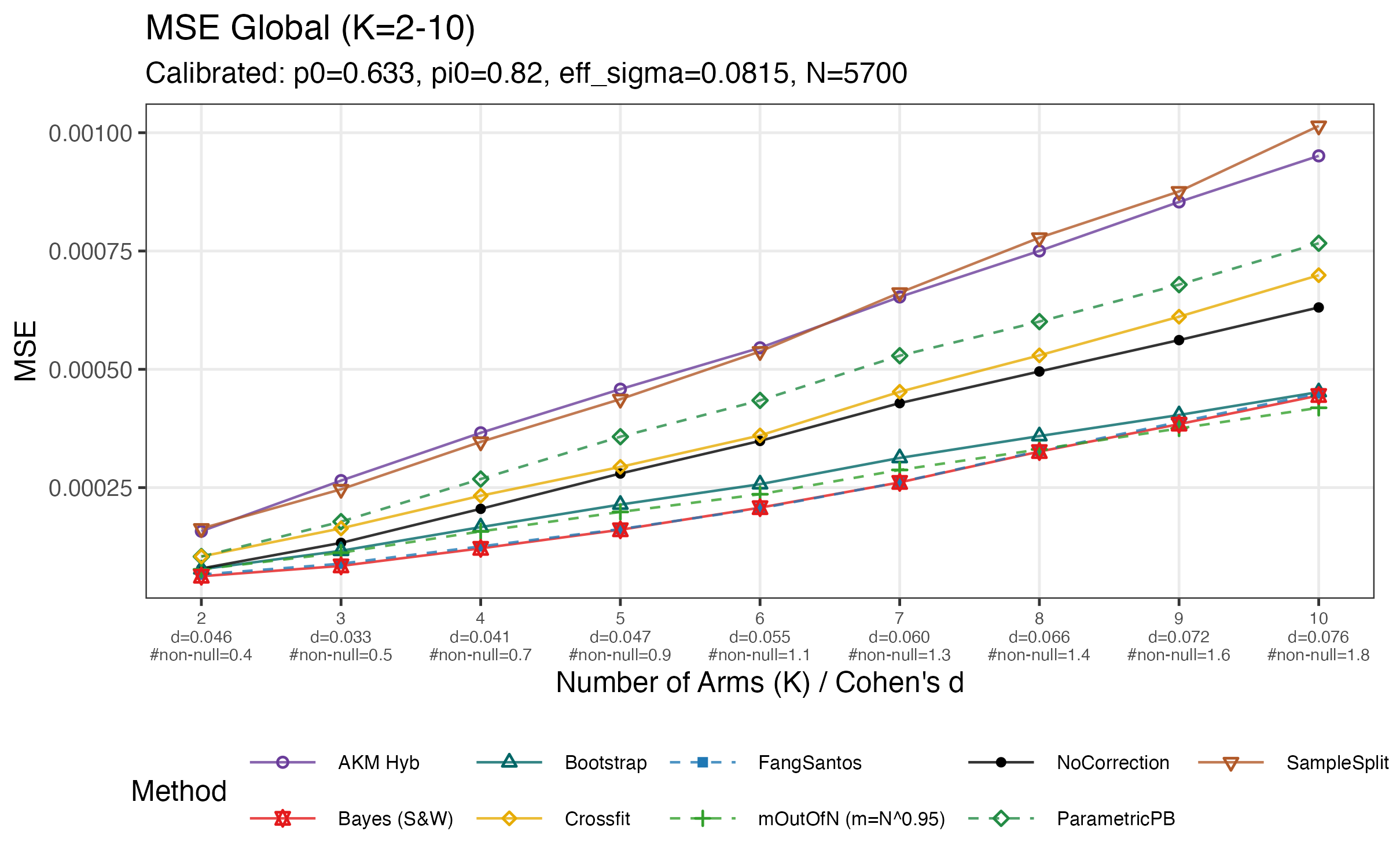}
          \caption{Global}                                                      
          \label{fig:binom_fixed_mse_global}
      \end{subfigure}                                                           
      \hfill      
      \begin{subfigure}[b]{0.48\textwidth}
          \centering                                                            
          \includegraphics[width=0.9 \textwidth]{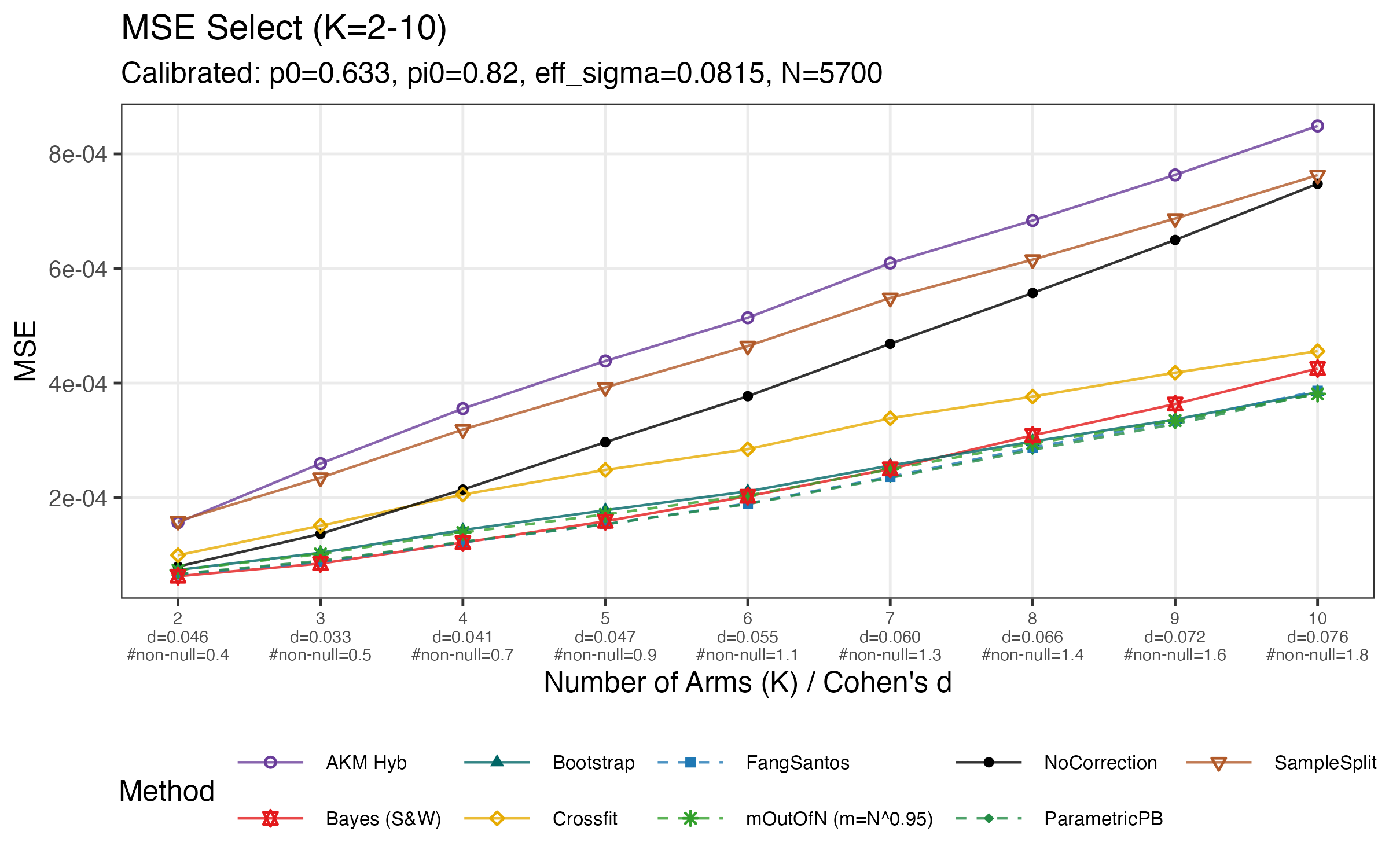}
          \caption{Select}                                                      
          \label{fig:binom_fixed_mse_select}
      \end{subfigure}
      \label{fig:binom_fixed_mse}                                               
      \footnotesize
      Notes: MSE as a function of sample size $N$ under the binomial DGP with   
  fixed base conversion rate. Panels vary by $K$.                               
  \end{figure}
                                                  
  \begin{figure}[h]
      \centering
      \caption{Coverage (Binomial, Fixed Conversion Rate)}
      \begin{subfigure}[b]{0.48\textwidth}
          \centering
          \includegraphics[width=0.9 \textwidth]{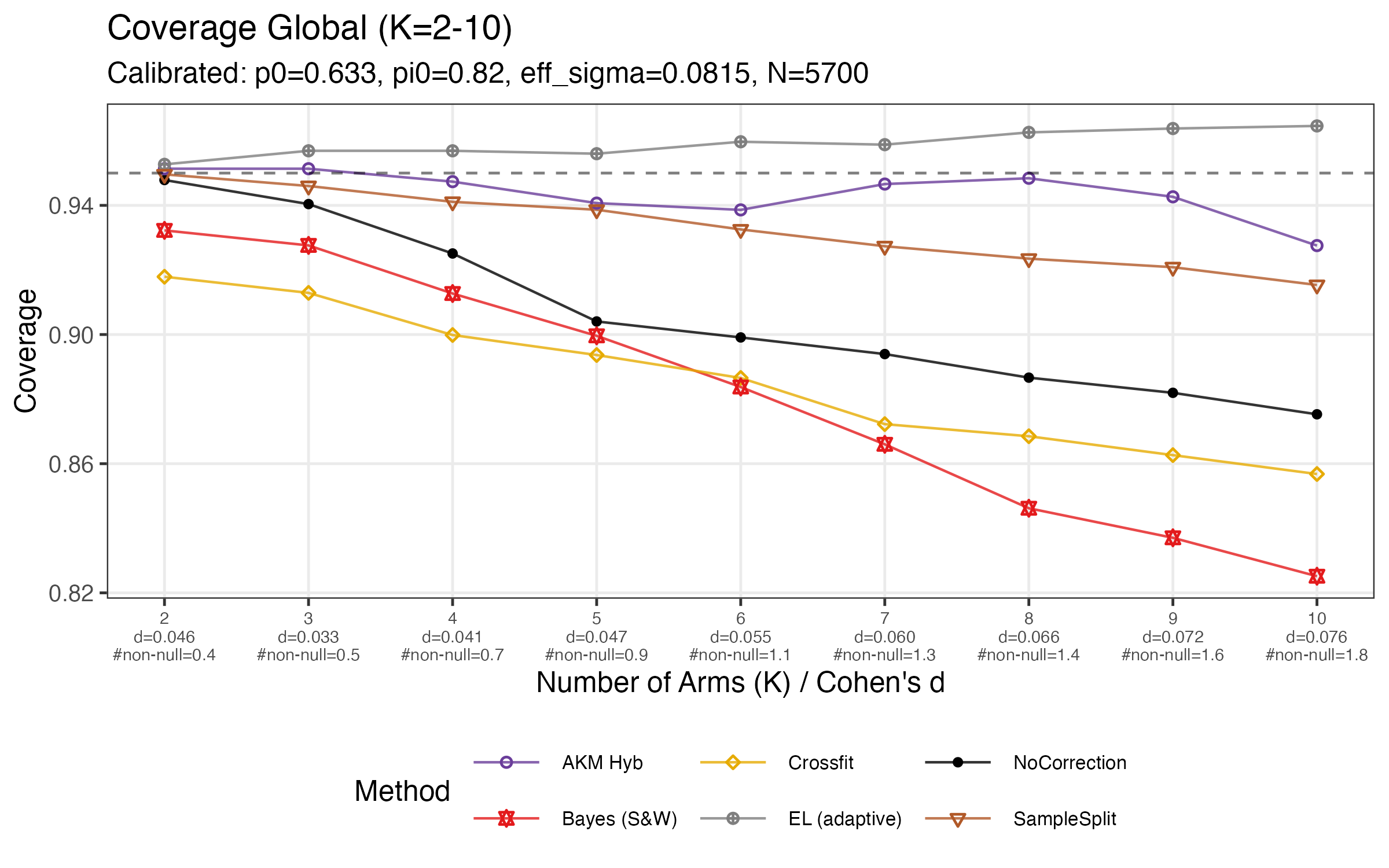}
          \caption{Global}                                                      
          \label{fig:binom_fixed_cov_global}
      \end{subfigure}                                                           
      \hfill      
      \begin{subfigure}[b]{0.48\textwidth}                                      
          \centering
          \includegraphics[width=0.9 \textwidth]{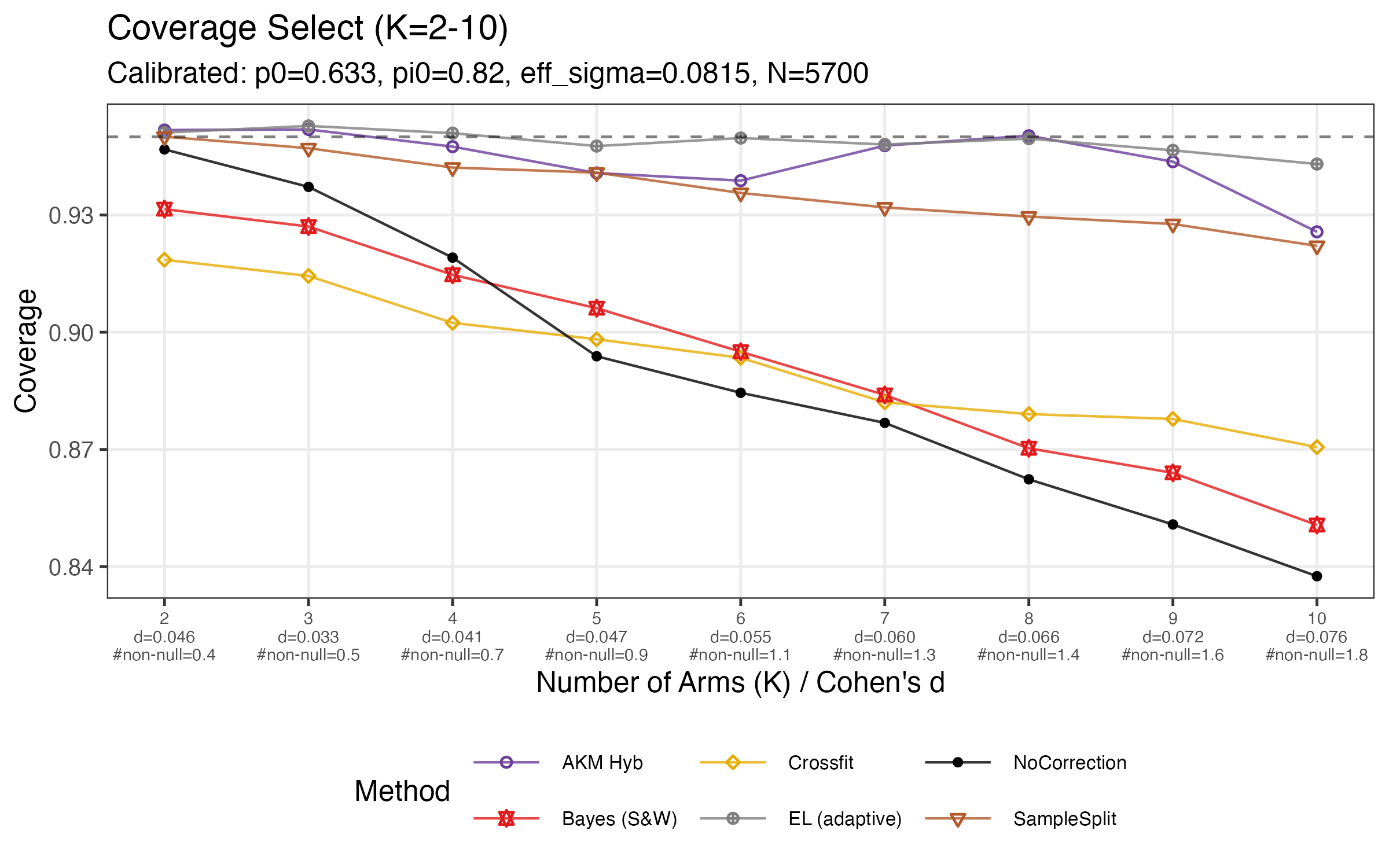}    
          \caption{Select}
          \label{fig:binom_fixed_cov_select}
      \end{subfigure}                                                           
      \label{fig:binom_fixed_cov}
      \footnotesize                                                             
      Notes: Coverage as a function of sample size $N$ under the binomial DGP
  with fixed base conversion rate. Panels vary by $K$. The dashed line indicates
   the nominal 95\% level.
  \end{figure}                                                                                
  \begin{figure}[h]
      \centering
      \caption{Regret (Binomial, Fixed Conversion Rate)}
      \includegraphics[width=0.5\textwidth]{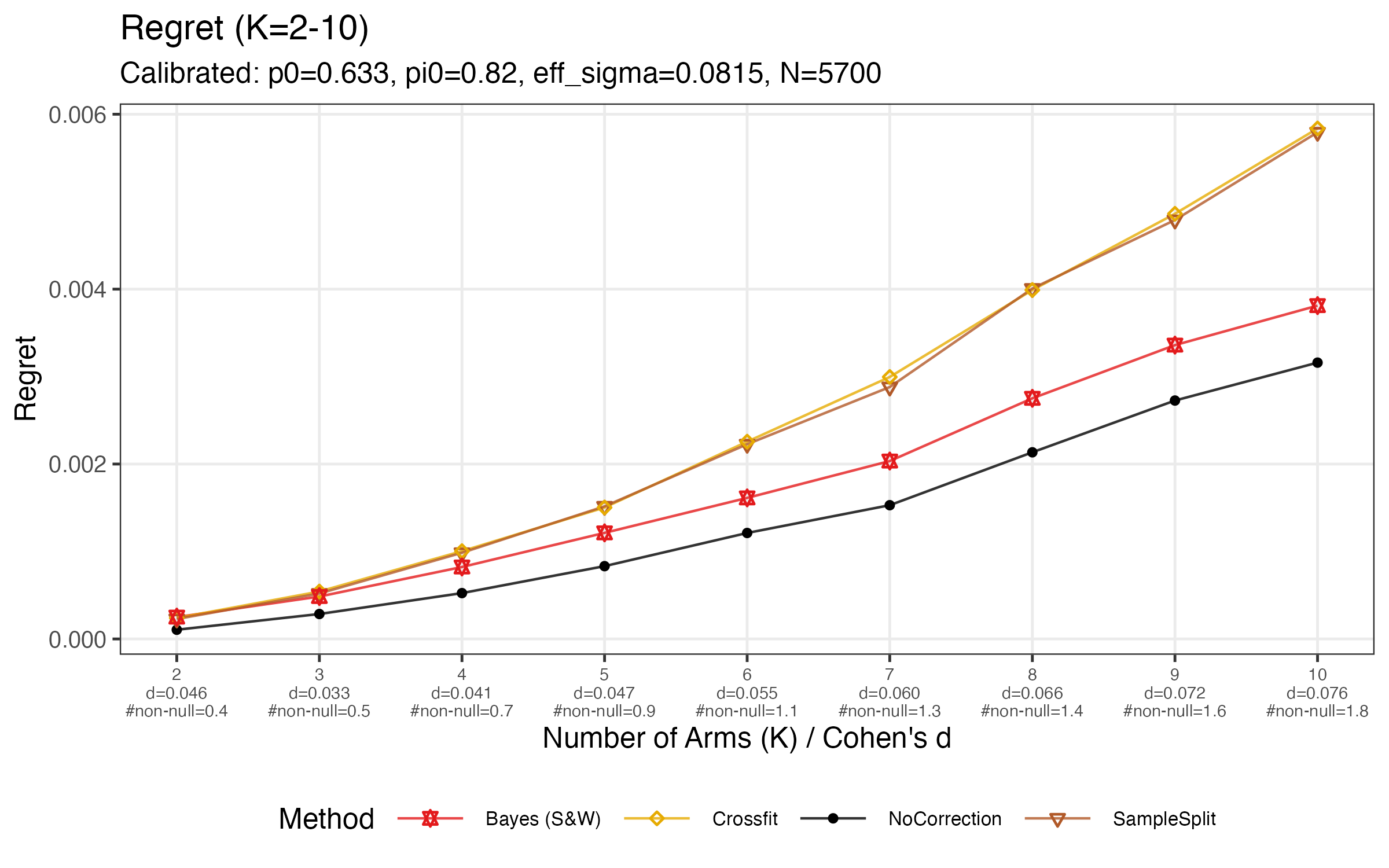}        
      \label{fig:binom_fixed_regret}
      
      \footnotesize                                                             
      Notes: Regret as a function of sample size $N$ under the binomial DGP with
   fixed base conversion rate. Panels vary by $K$.                              
  \end{figure}   
  
In contrast to our prior simulation's data generating process, outcomes here are binary: each visitor $i$ in arm $k$ produces a Bernoulli outcome $Y_{ik} \sim \text{Bernoulli}(p_k)$, representing a conversion event. The conversion probability $p_k$ is constructed from a null/non-null mixture. Further, each arm is independently designated as null with probability $\pi_0$ and non-null arms receive a treatment effect $\tau_k \sim \mathcal{N}(\mu_e, \sigma_e^2)$, so that $p_k = p_0 + D_k \cdot \tau_k$ where $D_k \sim \text{Bernoulli}(1 - \pi_0)$. 

We set $p_0 =0.633$, $\pi_0 = 0.82$, $\mu_e = -0.0032$, and $\sigma_e = 0.0815$ to match the empirical distribution of treatment effects observed on the Optimizely  platform, and fix $N = 5,700$ with balanced allocation across arms.\footnote{$\pi_0 = 0.82$ implies 82\% of the arms are null effects.} We vary the number of arms from $K = 2$ to $K = 10$. As $K$ grows, selection bias and the Winner's Curse increase. We report standardized effect size as Cohen's $d$. 
  Across Monte Carlo replications $r = 1, \ldots, R$, let $\mu_{(1)}^{(r)} \geq 
  \mu_{(2)}^{(r)}$ denote the top-two true arm rates and $\hat{s}_{(1)}^{(r)},  
  \hat{s}_{(2)}^{(r)}$ their sample standard deviations.   Cohen's $d$ increases with the number of arms and ranges from $0.04-0.08$, which implies small effects.\footnote{We define     $d = \frac{\tfrac{1}{R}\sum_{r=1}^{R} \left(\mu_{(1)}^{(r)} -  
  \mu_{(2)}^{(r)}\right)}         
               {\tfrac{1}{R}\sum_{r=1}^{R}                  
  \sqrt{\tfrac{1}{2}\!\left(\hat{s}_{(1)}^{(r)\,2} +                            \hat{s}_{(2)}^{(r)\,2}\right)}}$
  the MC-averaged top-two gap divided by the MC-averaged pooled empirical SD of 
  those two arms.}
Figures~\ref{fig:binom_fixed_bias}--\ref{fig:binom_fixed_regret} report the results. 

The qualitative patterns broadly mirror those of the prior simulation. For bias (Figure~\ref{fig:binom_fixed_bias}), cross-fitting performs best for both the selective and global targets at small $K$ and small $d$; as the number of arms and Cohen's $d$ increase, the Bayes shrinkage estimator and bootstrap corrections become more competitive. For MSE (Figure~\ref{fig:binom_fixed_mse}), the Bayes estimator performs best at small $K$, while bootstrap corrections overtake it as $K$ grows. For coverage (Figure~\ref{fig:binom_fixed_cov}), empirical likelihood and \cite{andrews2024inference} methods provide almost nominal coverage across the board. 
For regret (Figure~\ref{fig:binom_fixed_regret}), the plug-in estimator achieves the lowest regret across all $K$.  
  
\section{Conclusion}

Addressing the winner's curse is important when the results of experiments are used to decide which treatment to deploy and to inform policy making. Our paper identifies that by itself, the term "winner's curse" does not imply a well-defined objective, and there are three different goals that an experimenter might care about: the global winner's curse, the selected winner's curse, and regret (to the optimal decision). While each of these three objectives have been explored in past literature, we provide a framework to evaluate all three and discuss when each case is relevant. Our Winner's Curse identity in \Cref{eq:Winners-Curse-Identity} links these three objectives together. 

Our discussion of the the problem from a statistical perspective showed that the standard plug-in estimator (the max of two mean estimator) exhibits a Winner's Curse, causing an upward bias over the true value. The estimator contains a kink, or non-differentiable point, when the group means are equal, which poses a challenge for estimation and inference.

By reviewing and categorizing the different solutions in the literature, we establish the statistical properties for each estimator: bias, mean-squared error (MSE), and coverage of the confidence interval around the true value. We then provide a simulation study in \Cref{sec:Results-Comparison} to contrast these solutions, where we focus on the most common scenarios an experimenter might encounter: null effects, small effects, moderate effects, and large effects.

Methodologically, we provide a novel adaptation of empirical likelihood (EL) techniques for conducting inference and constructing confidence intervals. EL offers a non-parametric approach for likelihood-based inference. We first account for the kink in the problem by adapting critical value to be a chi-bar-squared distribution instead of the typical $\chi^2$ distribution used for a likelihood ratio test. We then propose an adaptive procedure based on pre-tests and theoretically ensure this adaptive procedure provides pointwise asymptotic coverage of the confidence interval (\Cref{thm:pointwise-coverage-two-arms} and \Cref{thm:pointwise-coverage-many-arms}). In practice, since we cannot know whether we are at the kink, the adaptive procedure ensures that our constructed confidence intervals attains proper coverage both at or away from the kink.

We then provide a calibrated, simulation to a real-world A/B testing platform, with many arms and non-normally distributed arms.
Increasing the number of arms does have implications for computational complexity and sample size. If more arms imply smaller samples per arms, many of the methods that underperform in the near-tie regime suffer more. This is especially true for bootstrap methods. The computational burden of bootstrap methods with multiple arms also makes them less feasible, making the applicability of our EL method even more relevant.

\Cref{tab:wc_scenarios_mse} summarizes our main simulation results. The optimal correction method depends critically on the anticipated treatment effect regime. When group means are equal or close (at or near the kink), cross-fitting performs best for bias reduction, while bootstrap bias correction achieves the lowest MSE. The efficiency loss inherent in sample-splitting explains why cross-fitting underperforms on MSE despite its strong bias properties. As group means diverge, the plug-in estimator increasingly dominates: It performs best for large treatment effects across all objectives and consistently minimizes regret regardless of effect size. For inference, our empirical likelihood method achieves correct coverage across the simulation designs and proves less sensitive to tuning parameters than resampling alternatives.

There are direct implications from our results for applied work. First, the choice of the global winner's curse ($WC_{global}$), selected winner's curse ($WC_{select}$), and regret will depend on the setting the decision-maker is facing. If they aim to choose the best option and not re-optimize in the future, then the selected winner's curse will apply as they care about bias and uncertainty of their chosen option. If they aim to repeatedly experiment and choose the best arm over time, then our global winner's curse analysis will provide them guidance. 
However, if the decision-maker is concerned about the loss from not implementing the optimal treatment, then our analysis with regret will apply. 
While the three scenarios are linked by our identity, they are not identical and thus the optimal solutions for each will be different. As a result, decision-makers need to determine which solution best suits their setting.

\begin{singlespace}
\bibliographystyle{ecta}
\bibliography{references}
\end{singlespace}

\newpage

\appendix
\crefalias{section}{appendix}
\crefalias{subsection}{appendix}

\newcommand{\appendicesstart}{
  \clearpage
  \setcounter{table}{0}
  \setcounter{figure}{0}
  \setcounter{equation}{0}
  \renewcommand{\thetable}{A.\arabic{table}}
  \renewcommand{\thefigure}{A.\arabic{figure}}
  \renewcommand{\theequation}{A.\arabic{equation}}
}

\appendicesstart

\part*{Supplementary Materials}  
\addcontentsline{toc}{part}{Supplementary Materials}
\setcounter{page}{1}

\section{Properties of the plug-in estimator}

\subsection{Asymptotic distribution
\label{app:Plugin-Asymptotics}}

We first establish consistency of the plug-in estimator.
Let $\hat{\mu}_n = (\hat{\mu}_{1n}, \hat{\mu}_{2n})^\top$ be the vector of sample means from two independent groups and $\mu = (\mu_1, \mu_2)^\top$ be their true means. By the Multivariate Central Limit Theorem, we have,
\begin{equation}
    \sqrt{n}(\hat{\mu}_n - \mu) \xrightarrow{d} Z \sim N\left( \begin{pmatrix} 0 \\ 0 \end{pmatrix}, \Sigma \right)
\end{equation}
where $\Sigma = \text{diag}(\sigma_1^2, \sigma_2^2)$. Let $Z = (Z_1, Z_2)^\top$.

We are interested in the asymptotic distribution of $\hat{\theta}_n = \max(\hat{\mu}_{1n}, \hat{\mu}_{2n})$. The parameter of interest is $\theta = \max(\mu_1, \mu_2)$.
Consider the normalized difference:
\[
\sqrt{n}(\hat{\theta}_n - \theta) = \sqrt{n}(\max(\hat{\mu}_{1n}, \hat{\mu}_{2n}) - \max(\mu_1, \mu_2))
\]
We analyze the limit case by case.

\textbf{Case 1: $\mu_1 > \mu_2$.}
In this case, $\theta = \mu_1$. From the weak law of large numbers, we have $\hat{\mu}_{1n} \xrightarrow{p} \mu_1$ and $\hat{\mu}_{2n} \xrightarrow{p} \mu_2$. For sufficiently large $n$, $\hat{\mu}_{1n} > \hat{\mu}_{2n}$ with probability approaching 1. 
We define $A_n = \{\hat{\mu}_{1n} \ge \hat{\mu}_{2n}\}$ to be this event and $ \mathbf{1}\{A_n\}\xrightarrow{p} (1+ o_p(1))$. Similarly, we define $A_n^c = \{\hat{\mu}_{1n} < \hat{\mu}_{2n}\}$ and $\mathbf{1}\{A_n^c\}\xrightarrow{p} o_p(1)$.
\begin{align*}
    \sqrt{n}(\hat{\theta}_n - \theta) &= \sqrt{n}(\hat{\mu}_{1n} - \mu_1) \mathbf{1}\{A_n\} + \sqrt{n}(\hat{\mu}_{2n} - \mu_1) \mathbf{1}\{A_n^c\} \\
    &= \sqrt{n}(\hat{\mu}_{1n} - \mu_1) \cdot (1 + o_p(1)) + o_p(1) \\
    &= \sqrt{n}(\hat{\mu}_{1n} - \mu_1) + o_p(1) \\
    &\xrightarrow{d} Z_1
\end{align*}

\textbf{Case 2: $\mu_2 > \mu_1$.}
By symmetry, $\theta = \mu_2$ and we have $ \sqrt{n}(\hat{\theta}_n - \theta) \xrightarrow{d} Z_2$ .

\textbf{Case 3: $\mu_1 = \mu_2$. }
This case represents the kink in the $\max$ function. We let $\mu_1 = \mu_2 = \mu^0$, then $\theta = \max\{\mu_1,\mu_2\}= \mu^0$.
\begin{align*}
    \sqrt{n}(\hat{\theta}_n - \theta) &= \sqrt{n}(\max(\hat{\mu}_{1n}, \hat{\mu}_{2n}) - \mu^0) \\
    &= \max(\sqrt{n}(\hat{\mu}_{1n} - \mu^0), \sqrt{n}(\hat{\mu}_{2n} - \mu^0))
\end{align*}
Since the joint vector of normalized differences converges to $(Z_1, Z_2)$ by the Multivariate Central Limit Theorem, applying the Continuous Mapping Theorem yields
\[
    \sqrt{n}(\hat{\theta}_n - \theta) \xrightarrow{d} \max(Z_1, Z_2),
\]
which establishes the consistency of the plug-in estimator.
However, the $\max(Z_1, Z_2)$ is not necessarily Gaussian as it is the maximum of two Gaussian distributions. For example, if $Z_1$ and $Z_2$ are independent Gaussian distributions with variances $\sigma_1^2$ and $\sigma_2^2$, then for any $t \in \mathbb{R}$,
\[
\mathbb{P}\big(\max(Z_1,Z_2) \le t\big)
= \mathbb{P}(Z_1 \le t, Z_2 \le t)
= \Phi\!\left(\frac{t}{\sigma_1}\right) \Phi\!\left(\frac{t}{\sigma_2}\right),
\]
which is not the CDF of a Gaussian distribution. Because the $\max$ function itself is not differentiable at the kink, we cannot apply the Delta method to construct the limiting distribution for this statistic which thus poses a challenge for inference.

\subsection{Finite sample bias \label{app:Plugin-Bias}}

From the asymptotic distribution, we can show that bias of the max of two means estimator, or Winner's Curse (Global), will decay at $\sqrt{n}$ rate at the kink when the two means are equal. 
We first define $X_n = \sqrt{n}(\hat\theta_n - \theta) = \max \{\sqrt{n}(\hat\mu_1 - \mu_1),\sqrt{n}(\hat\mu_2 - \mu_2)\}$ and $X=\max\{Z_1,Z_2\}$.  We then define bias as $\sqrt{n}b_n = \sqrt{n}(E[\hat\theta_n] - \theta) = E[X_n]$. 

From our asymptotic results in the prior subsection, we have $\sqrt{n}b_n = E[X_n] \to E[\max\{Z_1,Z_2\}]$. Then the bias is 
\[b_n = \frac{E[\max(Z_1,Z_2)]}{\sqrt{n}} + o_p\left(\frac{1}{\sqrt{n}}\right)\]
which implies the bias decays at $\sqrt{n}$ rate. 

For $Z_1,Z_2$ as independent mean-zero normals with variances, $\sigma_1,\sigma_2$, we can compute the bias term explicitly. We first rewrite
\[
\max(Z_1, Z_2)
= \frac{Z_1 + Z_2 + |Z_1 - Z_2|}{2}
\] which implies
$\mathbb{E}[\max(Z_1, Z_2)] = \frac{1}{2}\,\mathbb{E}|Z_1 - Z_2|$.
For the difference of independent normals, $Z_1 - Z_2 \sim N(0,\sigma_1^2 + \sigma_2^2)$, we then attain
\[
\mathbb{E}|Z_1 - Z_2| = \sqrt{\frac{2}{\pi}}\sqrt{\sigma_1^2 + \sigma_2^2}.
\]
Plugging this result into the prior max equation, we find
\[
\mathbb{E}[\max(Z_1, Z_2)]
= \sqrt{\frac{\sigma_1^2 + \sigma_2^2}{2\pi}}
\]
which gives us the analytical bias term
\[
b_n
= \mathbb{E}[\hat{\theta}_n] - \theta
= \frac{1}{\sqrt{n}} \sqrt{\frac{\sigma_1^2 + \sigma_2^2}{2\pi}}
+ o_p\left(\frac{1}{\sqrt{n}}\right)
\]
that decays at rate $1/\sqrt{n}$ at the kink.

\subsection{Hadamard differentiability \label{app:Plugin-Hadamard}}

\citet{fang2019inference} formalize that a statistic needs to be fully Hadamard differentiable as a necessary and sufficient condition for the standard nonparametric bootstrap to be consistent. We derive the Hadamard derivative for our plug-in estimator and show that it is only directionally, and not fully, Hadamard differentiable.

We let $\mathbb{D}$ and $\mathbb{E}$ be Banach spaces and $\phi : \mathbb{D}_\phi \to \mathbb{E}$ be a map with $\mathbb{D}_\phi \subset \mathbb{D}$. 
We say that $\phi$ is directionally Hadamard differentiable at $\mu \in \mathbb{D}_\phi$ tangentially to a subset $\mathbb{D}_0 \subset \mathbb{D}$ if there exists a map $\phi'_\mu : \mathbb{D}_0 \to \mathbb{E}$
such that, for every $h \in \mathbb{D}_0$ and for every sequence $t_n \downarrow 0$ and 
$h_n \to h$ with $h_n \in \mathbb{D}_0$ and $\mu + t_n h_n \in \mathbb{D}_\phi$ for all $n$,
\[
    \lim_{n \to \infty} 
    \frac{\phi(\mu + t_n h_n) - \phi(\mu)}{t_n}
    = \phi'_\mu(h).
\]
If the map $h \mapsto \phi'_\mu(h)$ is linear and continuous on $\mathbb{D}_0$, 
then $\phi$ fully Hadamard differentiable. 
The important requirement for the (functional) Delta Method to apply in the usual linear form  (and for the standard bootstrap to be consistent in this setting) is precisely that 
$\phi'_\mu$ is linear and continuous.

We now derive the Hadamard derivative for our functional $\phi(x) = \max(x_1, x_2)$ defined on $\mathbb{R}^2$. Let $\mu = (\mu_1, \mu_2)$ and perturbation $h = (h_1, h_2)$. Let $\phi : \mathbb{R}^2 \to \mathbb{R}$ be given by
    $\phi(x_1,x_2) = \max(x_1,x_2), $
and fix $\mu = (\mu_1,\mu_2)^\top \in \mathbb{R}^2$.
We consider sequences $t_n \downarrow 0$ with $t_n > 0$ and $h_n = (h_{n1},h_{n2})^\top \to h = (h_1,h_2)^\top$,
and compute the limit
\[
    \lim_{n\to\infty}
    \frac{\phi(\mu + t_n h_n) - \phi(\mu)}{t_n}.
\]

\textbf{Case 1: $\mu_1 > \mu_2$}. We define $\Delta = \mu_1 - \mu_2 > 0$.
We have
    $\phi(\mu + t_n h_n)
    = \max(\mu_1 + t_n h_{n1}, \mu_2 + t_n h_{n2})$.
We first observe that
\[
    \mu_1 + t_n h_{n1} \ge \mu_2 + t_n h_{n2}
    \iff
    \Delta \ge t_n (h_{n2} - h_{n1}).
\]
Since $h_n \to h$, the sequence $(h_{n2} - h_{n1})$ is bounded, we know $|h_{n2} - h_{n1}| \le M$ for a constant $M$ and sufficiently large $n$.
Because $t_n \to 0$ for large enough $n$, so $t_n M \to 0$ and $t_n M < \Delta$. Together, we have 
\[
    |t_n (h_{n2} - h_{n1})| \le t_n M < \Delta,
\]
which implies
    $t_n (h_{n2} - h_{n1}) \le \Delta
    \quad\text{for sufficiently large } n$.
Then, we rearrange the equation
    $\mu_1 + t_n h_{n1} \ge \mu_2 + t_n h_{n2}$,
so that
\[
    \phi(\mu + t_n h_n) =  \max(\mu_1 + t_n h_{n1}, \mu_2 + t_n h_{n2}) = \mu_1 + t_n h_{n1}
    \quad\text{for all sufficiently large } n.
\]
Therefore,
\[
    \frac{\phi(\mu + t_n h_n) - \phi(\mu)}{t_n}
    = \frac{\mu_1 + t_n h_{n1} - \mu_1}{t_n}
    = h_{n1}
    \quad\text{for sufficiently enough } n,
\]
and since $h_{n1} \to h_1$,
\[
    \phi'_\mu(h) = \lim_{n\to\infty}
    \frac{\phi(\mu + t_n h_n) - \phi(\mu)}{t_n}
    = h_1.
\]
Thus the directional derivative at $\mu$ in direction $h$ exists and is $\phi'_\mu(h) = h_1$,
which is linear in $h$.

\textbf{Case 2: $\mu_2 > \mu_1$}. By symmetric and similar logic to the prior case, the Hadamard derivative is $h_2$ and is linear in $h$.

\textbf{Case 3: $\mu_1 = \mu_2 = \mu^0$} (the kink). In this case, we have
    $\phi(\mu + t_n h_n)
    = \max(\mu^0 + t_n h_{n1},\, \mu^0 + t_n h_{n2})$,
so
\begin{align*}
    \frac{\phi(\mu + t_n h_n) - \phi(\mu)}{t_n}
    &= \frac{\max(\mu^0 + t_n h_{n1}, \mu^0 + t_n h_{n2}) - \mu^0}{t_n} \\
    &= \frac{\mu^0 + t_n \max(h_{n1},h_{n2}) - \mu^0}{t_n} \\
    &= \max(h_{n1},h_{n2}),
\end{align*}
where we used $t_n > 0$ to factor $t_n$ out of the max in the second line.

Since $h_n \to h$ and the function $(u,v) \mapsto \max(u,v)$ is continuous, we have $\max(h_{n1},h_{n2}) \to \max(h_1,h_2)$ for large enough $n$. 
Thus, we see
\[
    \phi'_\mu(h)
    = \lim_{n\to\infty}
    \frac{\phi(x + t_n h_n) - \phi(x)}{t_n}
    = \max(h_1,h_2).
\]
In this case, the directional derivative exists for every $h$ but the map
\[
    h \mapsto \phi'_\mu(h) = \max(h_1,h_2)
\]
is not linear, so $\phi$ is only \textit{directionally Hadamard differentiable}
at $x$ and not \textit{fully} Hadamard differentiable.

Thus, when the means are equal (at the kink) across arms, the standard nonparametric bootstrap will not apply and adjustments need to be made when using the bootstrap to perform bias correction or for conducting inference. If the means are different (away from the kink), then the bootstrap will apply and no adjustment needs to be made on the algorithm. Managerially, its not possible to know ex-ante which case we are in, which creates a conundrum.

\section{Empirical likelihood details \label{app:EL-details}}

This section provides the technical details for the empirical likelihood (EL) construction of confidence intervals used for
$$\theta = \max(\mu_1,\mu_2) = \phi(\mu_1,\mu_2),$$
where $\mu_1$ and $\mu_2$ are the two population means defined in the main text. We proceed in five steps: 
First, we briefly review EL for the mean vector $\mu=(\mu_1,\mu_2)^\top$.
Second, we derive the profile EL ratio for the non-smooth parameter $\theta=\max(\mu_1,\mu_2)$ and discuss the naive calibration that uses a standard $\chi^2$ critical value.
Third, we correct for the kink using a chi-bar-squared critical value and show how this value is computed analytically.
Fourth, we propose an adaptive EL procedure that uses a pre-test to decide which critical value ($\chi^2$) or chi-bar-squared) to use. We show that the procedure produces pointwise asymptotically consistent confidence intervals. Lastly, we extend the analysis to finding the maximum across more than two groups.

\subsection{Empirical likelihood for a vector of means} 

We have two independent samples $\{X_{1i}\}_{i=1}^{n_1}$ and $\{X_{2j}\}_{j=1}^{n_2}$ with means $\mu_1,\mu_2$ and write the observed data as $x_{1i}$ ($i=1,\dots,n_1$) and $x_{2j}$ ($j=1,\dots,n_2$).
Empirical likelihood treats the (unknown) distribution in each arm non-parametrically by assigning probability weights
\[
p_{1i}\ge 0,\quad \sum_{i=1}^{n_1}p_{1i}=1,
\qquad\text{and}\qquad
p_{2j}\ge 0,\quad \sum_{j=1}^{n_2}p_{2j}=1
\]
to the observed points in each arm \citep{owen2001empirical}.

The non-parametric likelihood of the empirical distribution is
\begin{equation}
\label{eq:EL-npl}
L(\{p_{1i}\},\{p_{2j}\}) =
\prod_{i=1}^{n_1} p_{1i}\;
\prod_{j=1}^{n_2} p_{2j}
\end{equation}
from the independence across group (or arm).
To evaluate the likelihood at a candidate mean vector $\mu=(\mu_1,\mu_2)^\top$, EL imposes the moment restrictions
\begin{equation}
\label{eq:EL-moment}
\sum_{i=1}^{n_1} p_{1i}\,(x_{1i}-\mu_1)=0,
\qquad
\sum_{j=1}^{n_2} p_{2j}\,(x_{2j}-\mu_2)=0,
\end{equation}
which enforce that the weighted means equal $(\mu_1,\mu_2)$.
The \emph{profile empirical likelihood function} for $\mu$ is obtained by maximizing the non-parametric likelihood  function (\Cref{eq:EL-npl}) subject to the simplex constraints for the probability weights and the moment restrictions (\Cref{eq:EL-moment}):
\[
L(\mu) =
\sup_{\{p_{1i}\},\{p_{2j}\}}
\Big\{
L(\{p_{1i}\},\{p_{2j}\})
:\; \sum_{i=1}^{n_1} p_{1i}\,(x_{1i}-\mu_1)=0; \sum_{j=1}^{n_2} p_{2j}\,(x_{2j}-\mu_2)=0
\Big\}.
\]
Typically, we can solve this constrained optimization problem using Lagrange multipliers. Because the two groups (or arms) are independent, the maximization is separate across groups. For $g\in\{1,2\}$, introduce a Lagrange multiplier $\lambda_g$ for the mean restriction and a multiplier $\nu_g$ for the probability constraint.  The maximizing weights have the closed form
\begin{equation}
\label{eq:EL-weights}
p_{gi}(\mu_g) =
\frac{1}{n_g}\,\frac{1}{1+\lambda_g(\mu_g)\,(x_{gi}-\mu_g)}
\end{equation}
for $i=1,\dots,n_g$ and where $\lambda_g(\mu_g)$ is the unique solution to
\begin{equation}
\label{eq:EL-lambda-eq}
\sum_{i=1}^{n_g}\frac{x_{gi}-\mu_g}{1+\lambda_g(\mu_g)\,(x_{gi}-\mu_g)}=0
\end{equation}
under the feasibility condition $1+\lambda_g(\mu_g)(x_{gi}-\mu_g)>0$ for all $i$. Solving the optimization routine, we find that the empirical likelihood is maximized at $\mu_g=\bar X_g$ (the sample mean) where $\lambda_g(\bar X_g)=0$ and $p_{gi}=1/n_g$.

We the define $\hat\mu=(\bar X_1,\bar X_2)^\top$ and denote the unconstrained maximizer and the profile empirical likelihood ratio as
\begin{equation}
\label{eq:EL-ratio-mu}
R_n(\mu)=
\frac{L(\mu)}{L(\hat\mu)} = 
\prod_{g=1}^2
\prod_{i=1}^{n_g}\bigl(1+\lambda_g(\mu_g)(x_{gi}-\mu_g)\bigr)^{-1},
\end{equation}
and the associated EL likelihood ratio statistic (or deviance) as
\begin{equation}
\label{eq:EL-deviance-mu}
W_n(\mu)
\;=\;
-2\log R_n(\mu)
\;=\;
2\sum_{g=1}^2\sum_{i=1}^{n_g}\log\!\bigl(1+\lambda_g(\mu_g)(x_{gi}-\mu_g)\bigr).
\end{equation}
Under standard regularity conditions, EL satisfies a Wilks-type theorem \citep{owen2001empirical}. For example, if the true mean vector is $\mu^0$, then
\begin{equation}
\label{eq:EL-wilks}
W_n(\mu^0)\ \xrightarrow{d}\ \chi^2_2,
\end{equation}
and $W_n(\mu)$ takes the place of a classical Wald statistic.

An asymptotic $1-\alpha$ confidence region for $\mu$ is therefore
\begin{equation}
\label{eq:EL-CR-mu}
\mathcal{C}^{EL}_{1-\alpha}(\mu)
\;=\;
\Big\{\mu\in\mathbb{R}^2:\; W_n(\mu)\le \chi^2_{2,\,1-\alpha}\Big\},
\end{equation}
where $\chi^2_{k,\,1-\alpha}$ denotes the $(1-\alpha)$ quantile of a $\chi^2_k$ distribution. 

\subsection{Profile empirical likelihood for the max of two means}

We now construct EL confidence intervals for our functional $\theta=\phi(\mu)=\max(\mu_1,\mu_2)$.
For a fixed value $\theta_0$, the set of mean vectors $\mu=(\mu_1,\mu_2)$ consistent with the null hypothesis
$H_0:\theta=\theta_0$ is
\begin{equation}
\label{eq:nullset-theta}
\mathcal{M}(\theta_0) =
\Big\{(\mu_1,\mu_2):\; \max(\mu_1,\mu_2)=\theta_0\Big\} =
\underbrace{\{(\theta_0,\mu_2):\mu_2\le \theta_0\}}_{\mathcal{A}(\theta_0)}
\ \cup\
\underbrace{\{(\mu_1,\theta_0):\mu_1\le \theta_0\}}_{\mathcal{B}(\theta_0)}.
\end{equation}
Thus $H_0:\theta=\theta_0$ is equivalent to either (i) $\mu_1=\theta_0$ and $\mu_2\le\theta_0$, or (ii) $\mu_2=\theta_0$
and $\mu_1\le\theta_0$.
The profile EL ratio for $\theta$ is obtained by maximizing the bivariate EL ratio over the set (\Cref{eq:nullset-theta}),
\begin{equation}
\label{eq:EL-profile-theta}
R_n(\theta_0)
\;=\;
\sup_{\mu\in \mathcal{M}(\theta_0)} R_n(\mu)
\;=\;
\max\Big\{\ \sup_{\mu_2\le \theta_0} R_n(\theta_0,\mu_2),\ \sup_{\mu_1\le \theta_0} R_n(\mu_1,\theta_0)\ \Big\}.
\end{equation}
Because the groups are independent, $R_n(\mu)=R_{1,n}(\mu_1)\,R_{2,n}(\mu_2)$ which is the product of the individual profile likelihood for each group $g$ and $R_{g,n}(\mu_g)$ is the individual profile likelihood for group $g$. Then, for group $g$, we write the one-sample EL deviance at value $m$ as
\begin{equation}
\label{eq:one-sample-deviance}
\ell_{g,n}(m) =
-2\log R_{g,n}(m),
\end{equation}
so that the total deviance for both groups is $W_n(\mu) = W_n(\mu_1,\mu_2)=\ell_{1,n}(\mu_1)+\ell_{2,n}(\mu_2)$.

Under $\mathcal{A}(\theta_0)$, we fix $\mu_1=\theta_0$ and maximize over $\mu_2\le \theta_0$.
Since $R_{2,n}(\mu_2)$ is maximized at $\mu_2=\bar X_2$, the optimizer is
\[
\mu_2^\star(\theta_0)=\min(\bar X_2,\theta_0),
\]
which yields deviance $0$ if $\bar X_2\le\theta_0$ and deviance $\ell_{2,n}(\theta_0)$ if $\bar X_2>\theta_0$  and the constraint binds at the boundary.
Therefore, the profile deviance under $\mathcal{A}(\theta_0)$ is
\[
T_{A,n}(\theta_0) =
\ell_{1,n}(\theta_0) + \mathbf{1}\{\bar X_2>\theta_0\}\,\ell_{2,n}(\theta_0).
\]

Under $\mathcal{B}(\theta_0)$, by symmetry, the profile deviance is
\[
T_{B,n}(\theta_0) =
\ell_{2,n}(\theta_0) + \mathbf{1}\{\bar X_1>\theta_0\}\,\ell_{1,n}(\theta_0).
\]
Combining the two cases and using \Cref{eq:EL-profile-theta}, the overall profile EL likelihood ratio statistic for 
$H_0:\theta=\theta_0$ is
\begin{equation}
\label{eq:profile-deviance-theta}
T_n(\theta_0) = -2\log R_n(\theta_0) =
\min\{T_{A,n}(\theta_0),\,T_{B,n}(\theta_0)\}.
\end{equation}
Equivalently, this procedure takes the maximum of the profile EL ratio over the two null branches $\mathcal A(\theta_0)$ and $\mathcal B(\theta_0)$ (or equivalently, in deviance form, the minimum across the two profiled deviances).

To construct the confidence interval, we see that given a critical value $c_{1-\alpha}$, the EL confidence set for $\theta$ is
\begin{equation}
\label{eq:CI-theta-general}
\mathcal{C}^{EL}_{1-\alpha}(\theta) =
\Big\{\theta\in\mathbb{R}:\; T_n(\theta)\le c_{1-\alpha}\Big\}.
\end{equation}
Numerically, we find its endpoints by solving $T_n(\theta)=c_{1-\alpha}$ on either side of the plug-in estimator
$\hat\theta_n=\max(\bar X_1,\bar X_2) = \max(\hat\mu_1,\hat\mu_2)$.

\subsection{Concerns at the kink}

A natural calibration is to use the standard $\chi^2_1$ critical value
$c^{\chi^2}_{1-\alpha}=\chi^2_{1,\,1-\alpha},$
which leads to the "naive" interval $\{\theta:T_n(\theta)\le \chi^2_{1,\,1-\alpha}\}$.

If the data generating process is away from the kink (e.g., $\mu_1>\mu_2$ so that $\theta=\mu_1$), then $\theta=\max(\mu_1,\mu_2)$ is locally smooth and equals $\mu_1$ in a neighborhood of the true parameter. In terms of the null set \Cref{eq:nullset-theta}, only the branch $\mathcal{A}(\theta)$ is asymptotically relevant: Under $H_0:\theta=\mu_1$, we have $\Pr(\bar X_2>\theta)\to 0$, so the inequality $\mu_2\le \theta$ is slack with probability approaching one. Consequently,
\begin{equation}
\label{eq:limit-smooth-T}
T_n(\theta)
=
\ell_{1,n}(\theta) + o_p(1)
\ \xrightarrow{d}\ \chi^2_1,
\end{equation}
and calibrating with the $\chi^2_{1,\,1-\alpha}$ critical value yields correct pointwise asymptotic coverage. By symmetry, the same logic applies when $\mu_2>\mu_1$.

At the kink, $\mu_1=\mu_2=\theta_0$, both branches $\mathcal{A}(\theta_0)$ and $\mathcal{B}(\theta_0)$ are locally relevant and the effective number of binding constraints is random.
To see this, we first note that under $H_0:\mu_g=\theta_0$ the one-sample EL deviances for group $g$ will satisfy
\[
\ell_{g,n}(\theta_0)
=
Z_g^2 + o_p(1),
\qquad
Z_g:=\frac{\sqrt{n_g}(\bar X_g-\theta_0)}{\sigma_g}\ \xrightarrow{d}\ N(0,1),
\]
and $Z_1\perp Z_2$ by independence of groups.
Plugging into likelihood ratio statistic (\Cref{eq:profile-deviance-theta}) gives the nonstandard limit
\begin{equation}
\label{eq:limit-kink-T}
T_n(\theta_0)\ \xrightarrow{d}\ 
T :=
\min\Big\{Z_1^2+\mathbf{1}\{Z_2>0\}Z_2^2,\;\; Z_2^2+\mathbf{1}\{Z_1>0\}Z_1^2\Big\},
\end{equation}
which has the piecewise representation,
\begin{equation}
\label{eq:T-piecewise}
T=
\begin{cases}
Z_1^2+Z_2^2, & Z_1>0,\ Z_2>0,\\
Z_1^2, & Z_1>0,\ Z_2\le 0,\\
Z_2^2, & Z_1\le 0,\ Z_2>0,\\
\min(Z_1^2,Z_2^2), & Z_1\le 0,\ Z_2\le 0,
\end{cases}
\end{equation} with four cases.
Due to the kink, the distribution in \Cref{eq:limit-kink-T}--\Cref{eq:T-piecewise} is not $\chi^2_1$ and yields larger critical values; consequently, using the smaller cutoff $\chi^2_{1,\,1-\alpha}$ yields a confidence interval that is too short and will undercover at the kink.

\subsection{Chi-bar-squared calibration at the kink}\label{app:Chi-bar-square-adaption}

We now derive the correct critical value for \Cref{eq:CI-theta-general} at the kink. While nonstandard, the limiting distribution $T$ in \Cref{eq:limit-kink-T} admits a closed-form CDF, which allows us to analytically derive the correct critical value.

At the kink ($\mu_1=\mu_2=\theta_0$), local deviations $h=(h_1,h_2)$ satisfying $\max(\theta_0+h_1,\theta_0+h_2)=\theta_0$
are exactly those with $\max(h_1,h_2)=0$, i.e.\ $h_1\le 0$, $h_2\le 0$, and at least one coordinate equals $0$.
This set is the union of the two negative coordinate axes,
$\mathcal{K} = 
\{(0,t):t\le 0\}\ \cup\ \{(t,0):t\le 0\}$.

In likelihood ratio problems under boundary constraints, the limiting statistic can be expressed as the squared Euclidean distance from a Gaussian score vector to the relevant tangent (cone) set \citep{Self1987}; empirical likelihood admits the same local quadratic form after studentization \citep{owen2001empirical}. In our setting, the (whitened) score vector is $Z=(Z_1,Z_2)^\top\sim N(0,I_2)$, and the squared distance to $\mathcal{K}$ is exactly the random variable $T$ in \Cref{eq:T-piecewise}.
The four cases in \Cref{eq:T-piecewise} correspond to which region of $\mathbb{R}^2$ the random vector $Z$ falls in:
(1) when both components are positive,  the closest point in $\mathcal{K}$ is the origin (both constraints bind), yielding $Z_1^2+Z_2^2$; 
(2--3) when exactly one component is positive, the projection lands on the corresponding axis (one constraint binds), yielding a $\chi^2_1$ term; 
(4) when both components are negative, the projection lands on the nearer axis, yielding $\min(Z_1^2,Z_2^2)$.

Let $F_k(t)=\Pr(\chi^2_k\le t)$ denote the $\chi^2_k$ CDF where $k$ is the degrees of freedom of the distribution. Because the signs of $Z_1$ and $Z_2$ are independent of their magnitudes and each quadrant has probability $1/4$, we can compute the CDF of $T$ by conditioning on quadrants:
\begin{align}
\label{eq:chibar-cdf}
\Pr(T\le t)
&=
\frac{1}{4}F_2(t)
+\frac{1}{4}F_1(t)
+\frac{1}{4}F_1(t)
+\frac{1}{4}\Pr\big(\min(\chi^2_1,\chi^2_1)\le t\big) \notag\\
&=
\frac{1}{4}F_2(t) + F_1(t) - \frac{1}{4}F_1(t)^2.
\end{align}
The last step uses $\Pr(\min(U,V)\le t)=1-\Pr(U>t,V>t)=1-(1-F_1(t))^2=2F_1(t)-F_1(t)^2$ for i.i.d.\ $U,V\sim\chi^2_1$. We refer to this distribution as the \emph{chi-bar-squared} calibration for the kink \citep{Dykstra1991}.

We then define the $(1-\alpha)$ critical value $c^{\bar\chi^2}_{1-\alpha}$ as the solution to
\begin{equation}
\label{eq:chibar-crit}
\frac{1}{4}F_2(c^{\bar\chi^2}_{1-\alpha}) + F_1(c^{\bar\chi^2}_{1-\alpha}) - \frac{1}{4}F_1(c^{\bar\chi^2}_{1-\alpha})^2 = 1-\alpha.
\end{equation}
Note that this equation is one-dimensional and can be solved by a standard root finder. For $\alpha=0.05$, we see that 
$c^{\bar\chi^2}_{0.95}\approx 4.245$ while $\chi^2_{1,\,0.95}\approx 3.842$. The chi-bar-squared adjusted EL confidence interval for $\theta$ is obtained by using $c^{\bar\chi^2}_{1-\alpha}$ in
\Cref{eq:CI-theta-general}:
\begin{equation}
\label{eq:CI-theta-chibar}
\mathcal{C}^{EL,\bar\chi^2}_{1-\alpha}(\theta) =
\Big\{\theta\in\mathbb{R}:\; T_n(\theta)\le c^{\bar\chi^2}_{1-\alpha}\Big\}.
\end{equation}
By construction, this calibration yields the correct asymptotic coverage at the kink $\mu_1=\mu_2$, but will overcover when not at the kink.

\subsubsection{Many arms chi-bar-squared critical value}\label{sec:Many-arms-chi-bar-squared-adaption}

After deriving the critical value for two arms, we now generalize these results to finite $J$ arms in an experiment and derive the generalize chi-bar-square critical value at a kink with $J$ ties. More specifically, we evaluate the null hypothesis $H_0: \max(\mu_1, \dots, \mu_J) = \theta_0$ at the kink (where $\mu_1 = \dots = \mu_J = \theta_0$) which requires calibrating against a generalized chi-bar-squared distribution. 

At the exact $J$-arm kink, the local deviations $h = (h_1, \dots, h_J)$ are constrained by $\max_j h_j = 0$. The tangent cone $\mathcal{K}$ for this constraint is the boundary of the negative orthant in $\mathbb{R}^J$, defined by $h_j \le 0$ for all $j$, with at least one $h_j = 0$. The limiting profile empirical likelihood deviance $T$ is the squared Euclidean distance from the studentized score vector $Z = (Z_1, \dots, Z_J)^\top \sim N(0, I_J)$ to the tangent cone $\mathcal{K}$.

Since the components of $Z$ are independent standard normals, $Z$ falls into any of the $2^J$ orthants with equal probability $(1/2)^J$. We use this result to derive the cumulative distribution function (CDF) of $T$ by conditioning on the signs of the components of $Z$:

\begin{enumerate}[noitemsep]
    \item \textbf{At least one positive component:} If exactly $k$ components are positive ($1 \le k \le J$), the projection of $Z$ onto $\mathcal{K}$ sets those $k$ components to $0$ and leaves the negative components unchanged. The squared distance is the sum of the $k$ positive squared components, which follows a $\chi_k^2$ distribution. There are $\binom{J}{k}$ such orthants.
    \item \textbf{All negative components:} If $Z_j \le 0$ for all $j$, the vector lies strictly inside the negative orthant. To project $Z$ onto the boundary $\mathcal{K}$, the component closest to zero must be shifted to zero. The squared distance is $\min(Z_1^2, \dots, Z_J^2)$. The CDF of the minimum of $J$ independent $\chi_1^2$ variables is $1 - (1 - F_1(t))^J$, where $F_1(t) = \Pr(\chi_1^2 \le t)$. There is exactly $1$ such orthant.
\end{enumerate}
By the Law of Total Probability, the exact mixture CDF for $T$ is:
$$ \Pr(T \le t) = \sum_{k=1}^J \binom{J}{k} \left(\frac{1}{2}\right)^J F_k(t) + \left(\frac{1}{2}\right)^J \left[ 1 - (1 - F_1(t))^J \right] $$
where $F_k(t)$ is the CDF of a $\chi_k^2$ random variable. For a given significance level $\alpha$, the critical value $c^{(J)}_{1-\alpha}$ is the unique root solving $\Pr(T \le c^{(J)}_{1-\alpha}) = 1 - \alpha$. Because the CDF is strictly monotonic, this root can be computed via standard 1D numerical root-finding for a given number of arms $J$ at the kink.

\subsection{Adaptive EL confidence intervals via a pre-test for the kink}
\label{app:EL-adaptive}

We adopt the test-based "pre-test" idea of \citet{andrews2000inconsistency} to construct an adaptive empirical likelihood procedure that switches between the standard $\chi^2_1$ calibration (valid \textit{away} from the kink) and the chi-bar-squared calibration (valid \textit{at} the kink). The goal is to achieve correct pointwise asymptotic coverage both at and away from the kink. The procedure has two steps. We first work this out for two arms and then establish the procedure for $J$ arms.

\textbf{Step 1: A pre-test for whether the DGP is at the kink.} Recall, we defined $\Delta = \mu_1-\mu_2$, so that the kink corresponds to $\Delta=0$. We further define the sample difference $\hat\Delta = \bar X_1-\bar X_2$, and let $$\hat s_\Delta^2 = \frac{s_1^2}{n_1}+\frac{s_2^2}{n_2},$$
where $s_g^2$ is the sample variance in group $g$. The studentized difference statistic is
$$S_n = \frac{|\hat\Delta|}{\hat s_\Delta}.$$
We then choose a set of pre-test critical values $\kappa_n>0$ that satisfy
\begin{equation}
\label{eq:kappa-conditions}
\kappa_n \to \infty
\qquad\text{and}\qquad
\kappa_n = o(\sqrt{n}),
\end{equation}
where $n$ denotes the overall sample size (or, equivalently, we can say $\kappa_n=o(\min\{\sqrt{n_1},\sqrt{n_2}\})$ for each groups's size).

The procedure classifies the setting as at the kink when the equality of means is not strongly rejected,
\begin{equation}\label{eq:K_n_definition}
\hat{K}_n = \mathbb{I}\{S_n \le \kappa_n\}.
\end{equation}
Then, $\hat{K}_n=1$ indicates "at (or close to) the kink" and $\hat{K}_n=0$ indicates "away from the kink".

\textbf{Step 2: Adaptive critical value for the EL likelihood ratio inversion.} We let $T_n(\theta)$ denote the profile EL deviance for $\theta=\max(\mu_1,\mu_2)$ defined in \Cref{eq:profile-deviance-theta}. We let $c^{\bar\chi^2}_{1-\alpha}$ be the chi-bar critical value defined by \Cref{eq:chibar-crit}, and let $\chi^2_{1,1-\alpha}$ be the $(1-\alpha)$ quantile of $\chi^2_1$.
The adaptive critical value is defined as
\begin{equation}
\label{eq:adaptive-crit}
c^{ad}_{1-\alpha} = \hat{K}_n\,c^{\bar\chi^2}_{1-\alpha} + (1-\hat{K}_n)\,\chi^2_{1,1-\alpha},
\end{equation}
which provides the adaptive EL confidence set,
\begin{equation}
\label{eq:adaptive-CI}
\mathcal C^{EL,ad}_{1-\alpha}(\theta) = \{\theta\in\mathbb{R}: T_n(\theta)\le c^{ad}_{1-\alpha}\}.
\end{equation}

\textbf{Asymptotic Coverage.}  We now show that the pre-test decision is consistent in \Cref{lem:Consistency-2-arm}, and the adaptive EL approach yields correct pointwise asymptotic coverage both at the kink ($\Delta=0$) and away from the kink ($\Delta\neq 0$) in \Cref{thm:pointwise-coverage-two-arms}.\\

\begin{lemma}[Consistency of the pre-test decision]\label{lem:Consistency-2-arm}
Assume $n_1,n_2\to\infty$, $n_1/(n_1+n_2)\to\pi\in(0,1)$, $s_g^2\to_p \sigma_g^2\in(0,\infty)$.
If $\Delta=0$, then $\Pr(\hat{K}_n=1)\to 1$. 
If $\Delta\neq 0$, then $\Pr(\hat{K}_n=1)\to 0$.\\    
\end{lemma}

\begin{proof}
First, we consider the case when $\Delta=0$. By the Central Limit Theorem, $S_n \Rightarrow |N(0,1)|$. Since $\kappa_n\to\infty$,
$\Pr(\hat{K}_n=0)=\Pr(S_n>\kappa_n)\to 0$, so $\Pr(\hat{K}_n=1)\to 1$.

Second, we consider the case when $\Delta\neq 0$. Then $\hat\Delta=\Delta+O_p(r_n^{-1})$, where
$$
r_n =\left(\frac{1}{n_1}+\frac{1}{n_2}\right)^{-1/2}.
$$
Moreover, under $s_g^2\to_p \sigma_g^2\in(0,\infty)$ we have 
$$
r_n\,\hat s_\Delta \xrightarrow{p} c = \left(\frac{\sigma_1^2}{n_1/(n_1+n_2)}+\frac{\sigma_2^2}{n_2/(n_1+n_2)}\right)^{1/2}\Big/\sqrt{n_1+n_2}.
$$
In particular, we see
$$
r_n^2\,\hat s_\Delta^2
=\frac{s_1^2 n_2+s_2^2 n_1}{n_1+n_2}
= s_1^2 \frac{n_2}{n_1+n_2} +s_2^2 \frac{n_1}{n_1+n_2}
\xrightarrow{p}\sigma_1^2  (1-\pi) + \sigma_2^2 \pi
=c^2,
$$
so $r_n \hat s_\Delta \xrightarrow{p} c\in(0,\infty)$.
Therefore,
$$
S_n=\frac{|\hat\Delta|}{\hat s_\Delta}
=\frac{|\Delta|+O_p(r_n^{-1})}{(c+o_p(1))/r_n}
=\left(\frac{|\Delta|}{c}\right) r_n + o_p(r_n)
\xrightarrow{p}\infty,
$$
since $r_n\to\infty$ when $n_1,n_2\to\infty$. Because $\kappa_n=o(r_n)$, it follows that
$\Pr(S_n\le \kappa_n)\to 0$ and thus $\Pr(\hat K_n=1)\to 0$.
\end{proof}

\begin{theorem}[Pointwise asymptotic coverage of the adaptive EL CI]\label{thm:pointwise-coverage-two-arms}
Let $\theta_0=\max(\mu_1,\mu_2)$ denote the true value. Under the same conditions as \Cref{lem:Consistency-2-arm}, for any fixed $\alpha\in(0,1)$, $ \Pr\big(\theta_0\in \mathcal C^{EL,ad}_{1-\alpha}(\theta)\big)\to 1-\alpha$.
\end{theorem}

\begin{proof}
We see by definition of the confidence interval (\Cref{eq:adaptive-CI}),
$$
\Pr(\theta_0\in \mathcal C^{EL,ad}_{1-\alpha}) = \Pr(T_n(\theta_0)\le c^{ad}_{1-\alpha}).
$$
We consider the kink and non-kink cases separately.

\emph{Case 1: $\Delta=0$ (at the kink).}
From \Cref{eq:limit-kink-T} and \Cref{eq:CI-theta-chibar}, we have
$T_n(\theta_0)\Rightarrow T$ and $\Pr(T\le c^{\bar\chi^2}_{1-\alpha})=1-\alpha$ by definition of $c^{\bar\chi^2}_{1-\alpha}$ (\Cref{eq:adaptive-crit}).
We decompose the probability based on the pre-test value,
$$
\Pr(T_n(\theta_0)\le c^{ad}_{1-\alpha}) =
\Pr(T_n(\theta_0)\le c^{\bar\chi^2}_{1-\alpha},\ \hat{K}_n=1) +
\Pr(T_n(\theta_0)\le \chi^2_{1,1-\alpha},\ \hat{K}_n=0).
$$
The second term is bounded by $\Pr(\hat{K}_n=0)\to 0$ by \Cref{lem:Consistency-2-arm}. We write out the first term to get
$$
\Pr(T_n(\theta_0)\le c^{\bar\chi^2}_{1-\alpha},\ \hat{K}_n=1) =
\Pr(T_n(\theta_0)\le c^{\bar\chi^2}_{1-\alpha}) - \Pr(T_n(\theta_0)\le c^{\bar\chi^2}_{1-\alpha},\ \hat{K}_n=0),
$$
and the last probability is again bounded by $\Pr(\hat{K}_n=0)\to 0$. Hence, we see
$$
\left| \Pr(T_n(\theta_0)\le c^{ad}_{1-\alpha}) -  \Pr(T_n(\theta_0)\le c^{\bar\chi^2}_{1-\alpha}) \right| \leq  2 \Pr(\hat{K}_n=0) \to 0.
$$
Since  $T_n(\theta_0) \Rightarrow T$, we have
$$
\Pr(T_n(\theta_0)\le c^{ad}_{1-\alpha}) \to \Pr(T\le c^{\bar\chi^2}_{1-\alpha}) = 1-\alpha.
$$

\emph{Case 2: $\Delta\neq 0$ (away from the kink).}
Without loss of generality suppose $\mu_1>\mu_2$, so $\theta_0=\mu_1$.
From the smooth (away from the kink) argument in \Cref{eq:limit-smooth-T}, we have $T_n(\theta_0)\Rightarrow \chi^2_1$.
By \Cref{lem:Consistency-2-arm}, $\Pr(\hat{K}_n=1)\to 0$ and $\Pr(\hat{K}_n=0)\to 1$.
Using the definition of the adaptive critical value \Cref{eq:adaptive-crit}, we can decompose
$$
\Pr\big(T_n(\theta_0)\le c^{ad}_{1-\alpha}\big) =
\Pr\big(T_n(\theta_0)\le \chi^2_{1,1-\alpha},\ \hat{K}_n=0\big) + 
\Pr\big(T_n(\theta_0)\le c^{\bar\chi^2}_{1-\alpha},\ \hat{K}_n=1\big).
$$ The second term is bounded by $\Pr(\hat{K}_n=1)\to 0$.  For the first term, note that
$$
\Pr\big(T_n(\theta_0)\le \chi^2_{1,1-\alpha},\ \hat{K}_n=0\big) =
\Pr\big(T_n(\theta_0)\le \chi^2_{1,1-\alpha}\big) -
\Pr\big(T_n(\theta_0)\le \chi^2_{1,1-\alpha},\ \hat{K}_n=1\big),
$$
and the last probability is bounded by $\Pr(\hat{K}_n=1)\to 0$.
Therefore,
$$
\Big| \Pr\big(T_n(\theta_0)\le c^{ad}_{1-\alpha}\big) -
\Pr\big(T_n(\theta_0)\le \chi^2_{1,1-\alpha}\big) \Big| \le 
2 \Pr(\hat{K}_n=1) \to 0,
$$
which implies
$$
\Pr\big(T_n(\theta_0)\le c^{ad}_{1-\alpha}\big) =
\Pr\big(T_n(\theta_0)\le \chi^2_{1,1-\alpha}\big) +o(1). $$
Finally, since $T_n(\theta_0)\Rightarrow \chi^2_1$,
$$ \Pr\big(T_n(\theta_0)\le \chi^2_{1,1-\alpha}\big)\to 1-\alpha, $$
and hence
$$ \Pr\big(T_n(\theta_0)\le c^{ad}_{1-\alpha}\big)\to 1-\alpha. $$
\end{proof}

As discussed in \citet{andrews2000inconsistency}, the use of a pre-test yields pointwise asymptotic validity at $\Delta=0$ and at any fixed $\Delta\neq 0$, but it is not uniformly valid over sequences $\Delta_n$ that approach zero at $1/\sqrt{n}$ rates.

\textbf{Choice of the pre-test critical value.} \Cref{eq:kappa-conditions} highlights the conditions for $\kappa_n$, and it requires $\kappa_n\to\infty$ (so the pre-test size tends to zero under $\Delta=0$), but not too quickly (so the pre-test still rejects with probability approaching one when $\Delta\neq 0$).
Equivalently, one may define $\kappa_n=z_{1-\gamma_n/2}$ for a sequence of pre-test significance levels $\gamma_n\downarrow 0$ such that $z_{1-\gamma_n/2}=o(\sqrt{n})$ (e.g., $\gamma_n=1/\log n$).
This is directly analogous to \citet{andrews2000inconsistency}, who emphasizes that the pre-test level should converge to zero to avoid incorrect boundary classification, while still allowing consistency away from the boundary. In our application, we adopt $\kappa_n = \sqrt{\log\log(\min\{n_1,n_2\})}$ which satisfies our required conditions on $\kappa_n$ (\Cref{eq:kappa-conditions}).

\subsection{Adaptive pre-test EL for many arms}\label{sec:Adaptive-many-arms}

In this section, we adapt our results for two arms to finite $J$ many arms or groups. Since pairwise pre-test comparisons combinatorially increases in the number of arms, we instead propose an "active-set" approach that first determines the best performing arm and then compares other arms to the maximum.
This procedure dynamically estimates $J_{eff}$, the number of arms tied at the maximum, and then selects the corresponding critical value. We detail the procedure in \Cref{alg:active_set} and then show asymptotically the procedure is both consistent and provides the correct pointwise coverage.

To set notation, we define the true active set of optimal arms as $\mathcal{A}^* = \{j : \mu_j = \theta_0\}$, where $\theta_0 = \max_j \mu_j$, and let $J^* = |\mathcal{A}^*|$. We define the empirical winner among the arm that has the highest sample mean: $\hat{k} = \arg\max_j \bar{X}_j$. The studentized differences of arm $j$ to the empirical winner (arm $\hat k$) is 
$S_{n,j} = \frac{\bar{X}_{\hat{k}} - \bar{X}_j}{\sqrt{s_{\hat{k}}^2/n_{\hat{k}} + s_j^2/n_j}}$
We choose the pre-test critical value as $\kappa_n = \sqrt{\log(\min_j n_j)}$.

\begin{algorithm}[H]
\SetAlgoLined
\caption{Adaptive EL Inference for $J$ Arms}
\label{alg:active_set}
\KwIn{Samples $\{X_{ji}\}_{i=1}^{n_j}$ for $j \in \{1, \dots, J\}$, significance level $\alpha$.}
\KwOut{Adaptive critical value $c_{1-\alpha}^{ad} \gets c^{(J_{eff})}_{1-\alpha}$.}
Compute sample means $\bar{X}_j$ and sample variances $s_j^2$ for all $j$\;
Identify the empirical winner: $\hat{k} \gets \arg\max_j \bar{X}_j$\;
Set the pre-test threshold: $\kappa_n \gets \sqrt{\log(\min_j n_j)}$\;
Initialize the active set: $\hat{\mathcal{A}} \gets \{\hat{k}\}$\;
\For{each arm $j \neq \hat{k}$}{
    Compute studentized difference: $S_{n,j} \gets \frac{\bar{X}_{\hat{k}} - \bar{X}_j}{\sqrt{s_{\hat{k}}^2/n_{\hat{k}} + s_j^2/n_j}}$\;
    \If{$S_{n,j} \le \kappa_n$}{
        $\hat{\mathcal{A}} \gets \hat{\mathcal{A}} \cup \{j\}$\;
    }
}
Count effective tied arms: $J_{eff} \gets |\hat{\mathcal{A}}|$\;
\end{algorithm}

\begin{lemma}[Consistency of the $J$-Arm Active Set]\label{lem:Consistency-J-arm}
For arm $j \in \{1,\ldots,J\}$, assume $n_j \to \infty$ such that $n_j/n \to \pi_j \in (0,1)$ for all $j$, and $s_j^2 \xrightarrow{p} \sigma_j^2 \in (0, \infty)$. Then, $\Pr(\hat{\mathcal{A}} = \mathcal{A}^*) \to 1$ as $n \to \infty$.
\end{lemma}

\begin{proof}
We evaluate the inclusion probability for optimal and suboptimal arms:

\textbf{Case 1 (Optimal Arms):} For any $j \in \mathcal{A}^*$, we have $\mu_j = \theta_0$ . We first evaluate the asymptotic properties of the studentized statistic $S_{n,j}$ by scaling the numerator and denominator by $\sqrt{n}$,
$$S_{n,j} = \frac{\sqrt{n}(\bar{X}_{\hat{k}} - \bar{X}_j)}{\sqrt{n}\sqrt{s_{\hat{k}}^2/n_{\hat{k}} + s_j^2/n_j}} = \frac{\sqrt{n}(\bar{X}_{\hat{k}} - \bar{X}_j)}{\sqrt{\frac{s_{\hat{k}}^2}{n_{\hat{k}}/n} + \frac{s_j^2}{n_j/n}}}$$

By the Central Limit Theorem, $\sqrt{n_m}(\bar{X}_m - \mu_m) = O_p(1)$ for all $m$. Since $\mu_{\hat{k}} = \mu_j = \theta_0$ and the fractional sample sizes converge ($n_m/n \to \pi_m$), the scaled difference in means is stochastically bounded: $\sqrt{n}(\bar{X}_{\hat{k}} - \bar{X}_j) = O_p(1)$.

For the denominator, we apply the assumptions that $s_m^2 \xrightarrow{p} \sigma_m^2 \in (0, \infty)$ and $n_m/n \to \pi_m \in (0,1)$. By the Continuous Mapping Theorem,
$$\sqrt{\frac{s_{\hat{k}}^2}{n_{\hat{k}}/n} + \frac{s_j^2}{n_j/n}} \xrightarrow{p} \sqrt{\frac{\sigma_{\hat{k}}^2}{\pi_{\hat{k}}} + \frac{\sigma_j^2}{\pi_j}} = C$$ where $ C\in (0, \infty)$
The limit of the scaled denominator ($C$) is bounded away from zero and infinity. Taking the ratio of these rates yields $S_{n,j} = O_p(1)$.

Since $S_{n,j} = O_p(1)$, for any $\epsilon > 0$, there exists a constant $M > 0$ such that $\Pr(S_{n,j} > M) < \epsilon$ for all sufficiently large $n$ (and equivalently $\Pr(S_{n,j} \leq M)  = 1- \Pr(S_{n,j} > M)  > 1 - \epsilon$). The procedure requires a tuning parameter that diverges, $\kappa_n \to \infty$, there exists an $N$ such that for all $n > N$, $\kappa_n > M$. Then we see that
$$\Pr(S_{n,j} \le \kappa_n) \ge \Pr(S_{n,j} \le M) > 1 - \epsilon$$
Since this holds for any $\epsilon > 0$, we conclude $\lim_{n \to \infty} \Pr(S_{n,j} \le \kappa_n) = 1$. Thus, $\Pr(j \in \hat{\mathcal{A}}) \to 1$.

\textbf{Case 2 (Suboptimal Arms):} For any $j \notin \mathcal{A}^*$, we have $\mu_j < \theta_0$. Let true mean gap be $\Delta_j = \theta_0 - \mu_j > 0$. 

By the Weak Law of Large Numbers, $\bar{X}_m \xrightarrow{p} \mu_m$ for all $m$. Because the optimal arms converge to $\theta_0$ and suboptimal arms converge to strictly lesser values, the empirical winner $\hat{k}$ will belong to $\mathcal{A}^*$ with probability approaching 1. Conditioning on the high-probability event that the winner is in the active set ($\{\hat{k} \in \mathcal{A}^*\}$), we have $\bar{X}_{\hat{k}} \xrightarrow{p} \theta_0$. 

The unscaled numerator of the test statistic converges in probability to the true mean gap, which is a strictly positive constant.
$$\bar{X}_{\hat{k}} - \bar{X}_j \xrightarrow{p} \theta_0 - \mu_j = \Delta_j > 0$$

From Case 1, for the denominator of the test statistics, the scaled portion of the standard error converges to a strictly positive constant $C$,
$$\sqrt{\frac{s_{\hat{k}}^2}{n_{\hat{k}}/n} + \frac{s_j^2}{n_j/n}} \xrightarrow{p} \sqrt{\frac{\sigma_{\hat{k}}^2}{\pi_{\hat{k}}} + \frac{\sigma_j^2}{\pi_j}} = C$$
where $C \in (0, \infty)$.
Then, we see that for $S_{n,j}$,
$$S_{n,j} = \frac{\bar{X}_{\hat{k}} - \bar{X}_j}{\frac{1}{\sqrt{n}} \sqrt{\frac{s_{\hat{k}}^2}{n_{\hat{k}}/n} + \frac{s_j^2}{n_j/n}}} = \sqrt{n} \left( \frac{\bar{X}_{\hat{k}} - \bar{X}_j}{\sqrt{\frac{s_{\hat{k}}^2}{n_{\hat{k}}/n} + \frac{s_j^2}{n_j/n}}} \right)$$

By the Continuous Mapping Theorem, the term inside the parentheses converges in probability to $\Delta_j / C > 0$. This strictly positive constant is multiplied by $\sqrt{n}$ and the test statistic diverges to positive infinity exactly at a $\sqrt{n}$ rate.

The pre-test procedure chooses a threshold $\kappa_n$ such that it grows strictly slower than $\sqrt{n}$, meaning $\kappa_n = o(\sqrt{n})$. To evaluate the probability of the arm being included in the active set, we divide the inequality by $\sqrt{n}$,
$$\Pr(S_{n,j} \le \kappa_n) = \Pr\left( \frac{S_{n,j}}{\sqrt{n}} \le \frac{\kappa_n}{\sqrt{n}} \right)$$
As $n \to \infty$, the left side of the inner inequality converges in probability to $\Delta_j / C > 0$, while the right side ($\kappa_n / \sqrt{n}$) converges deterministically to $0$. The probability that a strictly positive sequence is bounded above by a sequence converging to zero is further asymptotically zero
$$\lim_{n \to \infty} \Pr\left( \frac{S_{n,j}}{\sqrt{n}} \le \frac{\kappa_n}{\sqrt{n}} \right) = 0$$
Therefore, $\Pr(j \in \hat{\mathcal{A}}) \to 0$ and suboptimal arms are consistently excluded from the active set.

Combining both cases over the finite $J$ arms and applying the union bound yields $\Pr(\hat{\mathcal{A}} = \mathcal{A}^*) \to 1$.
\end{proof}

\begin{theorem}[Pointwise Asymptotic Coverage]\label{thm:pointwise-coverage-many-arms}
Let $\theta_0 = \max_j \mu_j$. Under the conditions of Lemma \Cref{lem:Consistency-J-arm}, for any fixed $\alpha \in (0,1)$, $\Pr(T_n(\theta_0) \le c_{1-\alpha}^{ad}) \to 1 - \alpha$, where $T_n(\theta_0)$ is the profile empirical likelihood deviance.
\end{theorem}

\begin{proof}
We split the coverage probability based on the pre-test outcome into two terms. Because $c_{1-\alpha}^{ad} = c_{1-\alpha}^{(J^*)}$ under the event $\hat{\mathcal{A}} = \mathcal{A}^*$, we have,

\begin{equation}\label{eqn:Coverage-J-Arms-Decomposition}
    \Pr(T_n(\theta_0) \le c_{1-\alpha}^{ad}) = \Pr\left(T_n(\theta_0) \le c_{1-\alpha}^{(J^*)} \cap \hat{\mathcal{A}} = \mathcal{A}^*\right) + \Pr\left(T_n(\theta_0) \le c_{1-\alpha}^{(|\hat{\mathcal{A}}|)} \cap \hat{\mathcal{A}} \neq \mathcal{A}^*\right).
\end{equation}
The second term is bounded above by $\Pr(\hat{\mathcal{A}} \neq \mathcal{A}^*)$, which vanishes asymptotically by \Cref{lem:Consistency-J-arm}, so it is $o(1)$. To simplify the first term, we apply the Law of Total Probability to the marginal event $\{T_n(\theta_0) \le c_{1-\alpha}^{(J^*)}\}$ and rearrange the terms,
   $\Pr\left(T_n(\theta_0) \le c_{1-\alpha}^{(J^*)} \cap \hat{\mathcal{A}} = \mathcal{A}^*\right) = \Pr\left(T_n(\theta_0) \le c_{1-\alpha}^{(J^*)}\right) - \Pr\left(T_n(\theta_0) \le c_{1-\alpha}^{(J^*)} \cap \hat{\mathcal{A}} \neq \mathcal{A}^*\right)$. 
The second term on the right-hand side is bounded above by $\Pr(\hat{\mathcal{A}} \neq \mathcal{A}^*) \to 0$, rendering it $o(1)$. We substitute this back into \Cref{eqn:Coverage-J-Arms-Decomposition} to get the coverage probability,
$$ \Pr(T_n(\theta_0) \le c_{1-\alpha}^{ad}) = \Pr\left(T_n(\theta_0) \le c_{1-\alpha}^{(J^*)}\right) + o(1) $$

To establish the limiting distribution of $T_n(\theta_0)$, we first note that null hypothesis imposes the constraint $\max_j \mu_j = \theta_0$. For any suboptimal arm $m \not\in \mathcal{A}^*$, $\bar{X}_m \xrightarrow{p} \mu_m < \theta_0$. Consequently, the inequality constraint $\mu_m \le \theta_0$ is strictly slack with probability approaching 1, and the projection penalty for suboptimal arms converges to zero. 

The asymptotic distribution is entirely dictated by the $J^*$ active constraints from $\mathcal{A}^*$. Geometrically, this restricts the $J^*$-dimensional score vector to the boundary of the negative orthant. 
Therefore, the deviance converges in distribution to the $J^*$-arm limit, $T_n(\theta_0) \Rightarrow \bar{\chi}^2_{J^*}$.
From the definition of the critical value, $\Pr(\bar{\chi}^2_{J^*} \le c^{(J^*)}_{1-\alpha}) = 1 - \alpha$, which completes the proof.
\end{proof}

\section{Why standard cross-fitting underestimates the variance at the kink ($\mu_1=\mu_2$)} \label{app:Cross-Fit-Coverage}

This section discusses why the standard cross-fitting procedure will lead to undercoverage at the kink when the group means are equal ($\mu_1=\mu_2$). 
In short, at the kink, the cross-fitting induces a correlation across folds that leads the variance of the estimator to be underestimated.  

Consider the case where we split the data sample to two folds, compute two sample-splitting estimators, and average them for the cross-fitting estimator.
Intuitvely, at the kink, imagine Fold 1 has a large positive spike in arm 1's estimate. This single spike does two things simultaneously to the two sample-split estimators.
First, it forces sample-split estimator 1's decision because arm 1 looks huge in fold 1, estimator 1 decides to select arm 1.
Second, it inflates sample-split estimator 2's value because estimator 2 might independently decide to pick arm 1 from fold 2. 
If it does, then the estimator reads the value from fold 1 and is impacted by that same huge spike. 
Together this induces a correlation across arms from the random noise. 
If we are away from the kink, the selection will be driven by the true difference of the arms instead of the noise and this issue is alleviated.

We now work this out more formally and provide a correction term for the variance of the cross-fitting estimator at the kink. 
Consider two-fold cross-fitting with two groups $j\in\{1,2\}$ and folds $\ell\in\{1,2\}$.
Let $\hat\mu^{(\ell)}_j$ be the sample mean of group $j$ in fold $\ell$, that has $m$ observations, so
$Var(\hat\mu^{(\ell)}_j)=1/m$. Define the fold-$\ell$ selection indicator
$$
I^{(\ell)} = \mathbf{1}\{\hat\mu^{(\ell)}_1 > \hat\mu^{(\ell)}_2\}.
$$
From this notation, the two sample-splitting estimators are
$$
\hat\theta^{(1)} = I^{(1)}\hat\mu^{(2)}_1 + (1-I^{(1)})\hat\mu^{(2)}_2,\qquad
\hat\theta^{(2)} = I^{(2)}\hat\mu^{(1)}_1 + (1-I^{(2)})\hat\mu^{(1)}_2 .
$$
At the kink, $\mu_1=\mu_2$. Without loss of generality, we assume $\mu_1=\mu_2=0$, so $E[\hat\theta^{(1)}]=E[\hat\theta^{(2)}]=0$ and 
$$Cov(\hat\theta^{(1)},\hat\theta^{(2)}) = E[\hat\theta^{(1)}\hat\theta^{(2)}]. $$
Expanding the product inside the expectation term yields four terms:
\begin{align*}
\hat\theta^{(1)}\hat\theta^{(2)}
&= I^{(1)}I^{(2)}\,\hat\mu^{(2)}_1\hat\mu^{(1)}_1
+ I^{(1)}(1-I^{(2)})\,\hat\mu^{(2)}_1\hat\mu^{(1)}_2 \\
&\quad + (1-I^{(1)})I^{(2)}\,\hat\mu^{(2)}_2\hat\mu^{(1)}_1
+ (1-I^{(1)})(1-I^{(2)})\,\hat\mu^{(2)}_2\hat\mu^{(1)}_2 .
\end{align*}
To apply fold independence cleanly, we regroup each term into a fold-1 factor times a fold-2 factor. For the first term, we see
$$ I^{(1)}I^{(2)}\,\hat\mu^{(2)}_1\hat\mu^{(1)}_1 = (I^{(1)}\hat\mu^{(1)}_1)\,(I^{(2)}\hat\mu^{(2)}_1),$$
and similarly we rewrite for the other terms. Since all fold-1 quantities are independent of all fold-2 quantities, each expectation factors into a product of two within-fold expectations. For the first term, this implies 
$$ E[I^{(1)}I^{(2)}\,\hat\mu^{(2)}_1\hat\mu^{(1)}_1] = E[I^{(1)}\hat\mu^{(1)}_1]E[I^{(2)}\hat\mu^{(2)}_1].$$

We now turn our attention to computing the within fold values. We let $Z_1,Z_2\stackrel{iid}{\sim}N(0,1)$ and define $I=\mathbf{1}\{Z_1>Z_2\}$.
We can show,\footnote{To see this, we have $Z_1=(\Delta+S)/2$ with $\Delta=Z_1-Z_2\sim N(0,2)$, $S=Z_1+Z_2\sim N(0,2)$ independent,
so $E[Z_1\mathbf{1}\{\Delta>0\}]=\tfrac12E[\Delta\mathbf{1}\{\Delta>0\}]=\tfrac12\cdot \tfrac{\sqrt{2}}{\sqrt{2\pi}}=1/(2\sqrt{\pi})$.)}
$$ E[Z_1 I] = E[Z_1\mathbf{1}\{Z_1>Z_2\}] = \frac{1}{2\sqrt{\pi}}. $$

Scaling from $Z_j$ to $\hat\mu^{(\ell)}_j$ using the standard error of the group mean ($1/\sqrt m$) yields the constant
$$ K  = E\left[I^{(\ell)}\hat\mu^{(\ell)}_1\right] = \frac{1}{\sqrt m} \left(\frac{1}{2\sqrt{\pi}}\right). $$
By symmetry at the kink, the remaining within-fold expectations are
\[
E\!\left[I^{(\ell)}\hat\mu^{(\ell)}_2\right] = -K,\qquad
E\!\left[(1-I^{(\ell)})\hat\mu^{(\ell)}_2\right] = K,\qquad
E\!\left[(1-I^{(\ell)})\hat\mu^{(\ell)}_1\right] = -K.
\]

The expectations of each of the four terms are,
\begin{align*}
E\left[I^{(1)}I^{(2)}\hat\mu^{(2)}_1\hat\mu^{(1)}_1\right]
&= E[I^{(1)}\hat\mu^{(1)}_1]E[I^{(2)}\hat\mu^{(2)}_1] = K^2,\\
E\left[I^{(1)}(1-I^{(2)})\hat\mu^{(2)}_1\hat\mu^{(1)}_2\right]
&= E[I^{(1)}\hat\mu^{(1)}_2]E[(1-I^{(2)})\hat\mu^{(2)}_1] = (-K)(-K)=K^2,\\
E\left[(1-I^{(1)})I^{(2)}\hat\mu^{(2)}_2\hat\mu^{(1)}_1\right]
&= E[(1-I^{(1)})\hat\mu^{(1)}_1]E[I^{(2)}\hat\mu^{(2)}_2] = (-K)(-K)=K^2,\\
E\left[(1-I^{(1)})(1-I^{(2)})\hat\mu^{(2)}_2\hat\mu^{(1)}_2\right]
&= E[(1-I^{(1)})\hat\mu^{(1)}_2]E[(1-I^{(2)})\hat\mu^{(2)}_2] = K^2.
\end{align*}
Combining the terms, we attain,
$$ Cov(\hat\theta^{(1)},\hat\theta^{(2)}) = E[\hat\theta^{(1)}\hat\theta^{(2)}] = 4K^2
= 4\left(\frac{1}{m}\cdot\frac{1}{4\pi}\right) = \frac{1}{\pi m}. $$
At the kink $Var(\hat\theta^{(1)})=Var(\hat\theta^{(2)})=Var(\hat\mu^{(\ell)}_j)=1/m$, so we see that correlation across the two folds' estimates is 
$$
\rho = Corr(\hat\theta^{(1)},\hat\theta^{(2)}) = \frac{Cov(\hat\theta^{(1)},\hat\theta^{(2)})}{\sqrt{Var(\hat\theta^{(1)})Var(\hat\theta^{(2)})}} = \frac{1/(\pi m)}{1/m} = \frac{1}{\pi}.
$$
Finally, for the cross-fit average of the two folds, $\hat\theta_{\mathrm{CF}}=\frac{1}{2}(\hat\theta^{(1)}+\hat\theta^{(2)})$, we see
$$
Var(\hat\theta_{\mathrm{CF}}) = \frac14\Big(Var(\hat\theta^{(1)})+Var(\hat\theta^{(2)})+2Cov(\hat\theta^{(1)},\hat\theta^{(2)})\Big) = \frac{1}{2m}\left(1+\frac{1}{\pi}\right),
$$
so relative to the usual variance $\frac{1}{2m}$ implied from the the independent splits, the correct kink variance is inflated by a factor of $\left(1+\frac{1}{\pi}\right)$.

\section{Extended simulation details} \label{app:Simulation-Details}

\subsection{Additional methods}
Table~\ref{tab:appendix_additional_methods} summarizes all additional procedures used in our simulations.  Figures~\ref{fig:wc_global_app}--\ref{fig:cov_select_app} replicate the main simulation plots on an expanded set of methods; each panel varies $d$ (columns) and $\sigma$ (rows), and the legend matches the main figures.

\begin{table*}[h!]
\centering
\caption{Additional methods}
\label{tab:appendix_additional_methods}
\begingroup
\setlength{\tabcolsep}{3pt}
\renewcommand{\arraystretch}{1.02}
\setlength{\aboverulesep}{0pt}
\setlength{\belowrulesep}{0pt}
\setlength{\abovecaptionskip}{2pt}
\setlength{\belowcaptionskip}{0pt}

\resizebox{\textwidth}{!}{
\begin{tabularx}{\textwidth}{@{}>{\raggedright\arraybackslash}p{\methodcolw}Y@{}}
\toprule
\textbf{Method} & \textbf{Procedure} \\
\midrule

\textbf{Hong--Li 18 (numerical derivative bootstrap)} &
\algcell{
\textbf{Input:} $D=\{(T_i,Y_i)\}_{i=1}^N$, $B=200$.
Let $D_k=\{Y_i:T_i=k\}$, $n_k=|D_k|$, and $n=n_1+n_2$.
Compute $\bar{Y}_k(D)$ and $\hat{\phi}(D)=\max_{k}\bar{Y}_k(D)$.\\
Set $\varepsilon = N^{-1/4}$ (fallback: $n^{-1/4}$ if needed).\\[2pt]

\textbf{For} $b=1,\dots,B$ \textbf{do}\\
\hspace*{1em} Sample $D_1^{*(b)}$ by drawing $n_1$ elements from $D_1$ with replacement;
sample $D_2^{*(b)}$ by drawing $n_2$ elements from $D_2$ with replacement.\\
\hspace*{1em}Compute bootstrap means $\bar{Y}_k^{*(b)}=\bar{Y}_k(D_k^{*(b)})$ for $k\in\{1,2\}$.\\
\hspace*{1em}Form perturbations:
$h_k^{*(b)}=\sqrt{n}\,\big(\bar{Y}_k^{*(b)}-\bar{Y}_k(D)\big)$, $k\in\{1,2\}$.\\
\hspace*{1em}Compute shifted functional
$\hat{\phi}^{*(b)}_{\varepsilon}=\max_{k}\big\{\bar{Y}_k(D)+\varepsilon\,h_k^{*(b)}\big\}$.\\
\hspace*{1em}Compute numerical-derivative draw
$d^{*(b)}=\big(\hat{\phi}^{*(b)}_{\varepsilon}-\hat{\phi}(D)\big)/\varepsilon$.\\
\textbf{End for.}\\[4pt]

\textbf{Bias-corrected point estimate:}\;
$\displaystyle
\hat{\phi}_{\mathrm{HL}}
=\frac{1}{B}\sum_{b=1}^{B}\left(\hat{\phi}(D)-\frac{d^{*(b)}}{\sqrt{n}}\right)
=\hat{\phi}(D)-\frac{\bar d}{\sqrt{n}},\quad
\bar d=\frac{1}{B}\sum_{b=1}^{B} d^{*(b)}.$
} \\ \hline

\textbf{AKM24 Conditional (AKM--Cond)} &
\algcell{
\textbf{Input:} $x_W=\hat{\phi}(D)$, $x_L$, $sd_W$, $\alpha=0.05$.\\
(Here $x_L$ is the nonwinning sample mean and $sd_W=\hat{se}(\bar{Y}_{\hat{k}})$ is the winning-arm standard error.)\\
\textbf{Point estimate:} solve for $\hat{\mu}_{\mathrm{cond}}$ such that
$F_{\mathrm{TN}}(x_W;\mu,sd_W,L=x_L,U=\infty)=\tfrac12$.\\
\textbf{95\% CI:} solve for $(\mu^L,\mu^U)$ such that
$F_{\mathrm{TN}}(x_W;\mu^L,\cdot)=1-\alpha/2$ and $F_{\mathrm{TN}}(x_W;\mu^U,\cdot)=\alpha/2$,
with the same truncation $[x_L,\infty)$.
} \\ \hline

\textbf{AKM24 Projection (AKM--Proj)} &
\algcell{
\textbf{Input:} $x_W=\hat{\phi}(D)$, $sd_W$, $\alpha=0.05$, $K=2$.\\
Let $c_{\alpha}=\Phi^{-1}\!\left(\frac{1+(1-\alpha)^{1/K}}{2}\right)$.\\
\textbf{95\% CI:} $\big[x_W-c_{\alpha}sd_W,\;x_W+c_{\alpha}sd_W\big]$.
} \\ \hline

\textbf{Empirical likelihood (EL, chi-square)} &
\algcell{
\textbf{Input:} arm samples $y_1=\{Y_i:T_i=1\}$ and $y_2=\{Y_i:T_i=2\}$, $\alpha=0.05$.\\
\textbf{Critical value (standard):}\;
$q_{\chi^2}(1-\alpha)=\chi^2_{1,\,1-\alpha}$, i.e.,
$\Pr\!\big(\chi^2_1 \le q_{\chi^2}(1-\alpha)\big)=1-\alpha$.\\
\textbf{95\% CI:}\; $\{\theta:\ \mathrm{dev}(\theta)\le \chi^2_{1,\,1-\alpha}\}$.

} \\ \hline

\textbf{Empirical likelihood (EL, chi-bar)} &
\algcell{
\textbf{Input:} arm samples $y_1=\{Y_i:T_i=1\}$ and $y_2=\{Y_i:T_i=2\}$, $\alpha=0.05$.\\
\textbf{Critical value (chi-bar):}\; $q_{\bar{\chi}^2}(1-\alpha)$ solves
$\Pr\!\big(\bar{\chi}^2 \le q_{\bar{\chi}^2}(1-\alpha)\big)=1-\alpha$, where
$\bar{\chi}^2$ is the chi-bar-squared null distribution for the max-mean constraint.\\
\textbf{95\% CI:}\; $\{\theta:\ \mathrm{dev}(\theta)\le q_{\bar{\chi}^2}(1-\alpha)\}$.

  } \\\hline
  \bottomrule
\end{tabularx}
}
\endgroup
\end{table*}

\begin{figure}[!ht]
    \centering
    \caption{Winner's Curse Bias (Global) }
    \includegraphics[width=0.9 \textwidth]{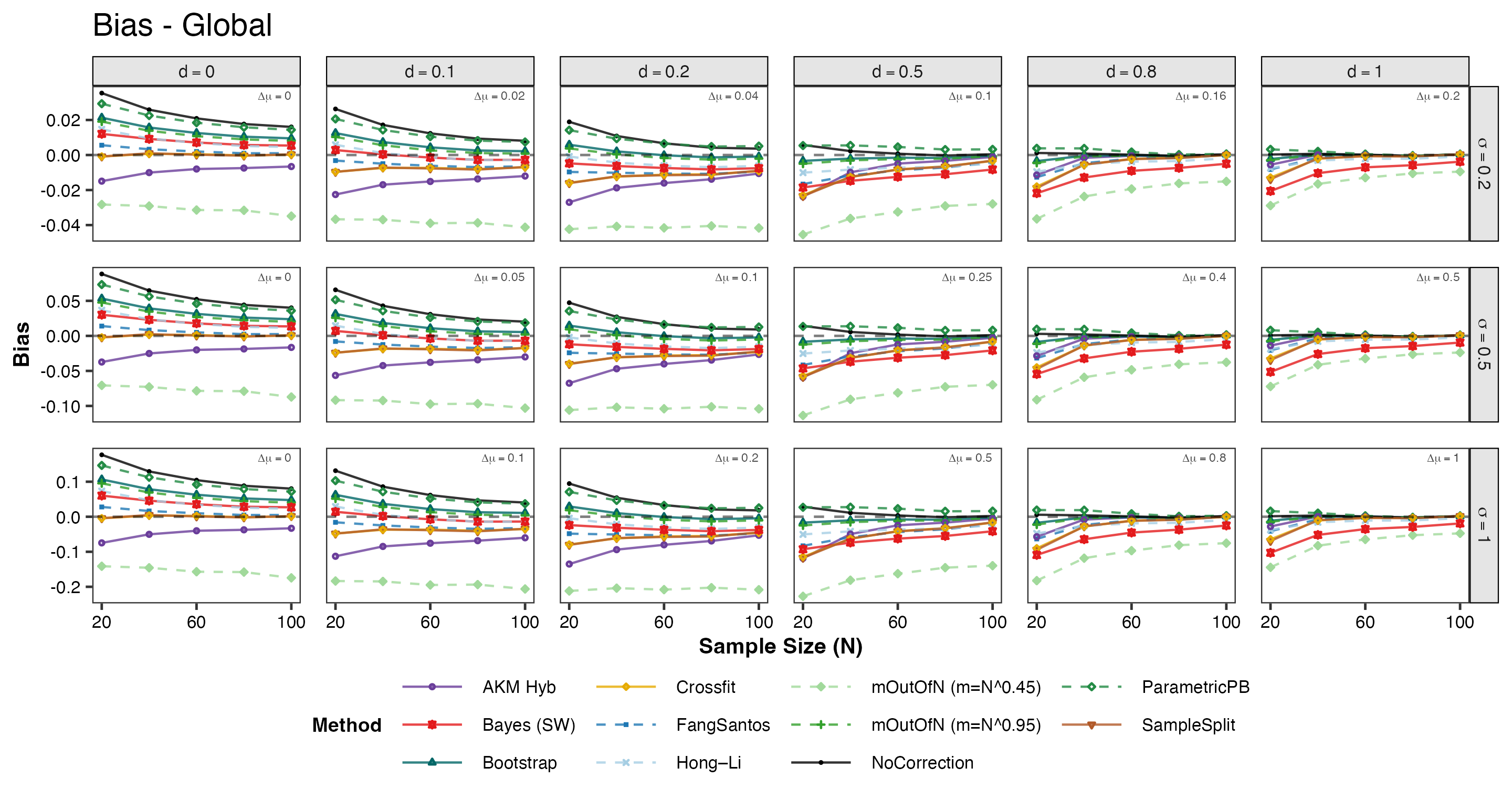}
    \label{fig:wc_global_app}

    \footnotesize
    Notes: WC Global bias as a function of sample size $N$. Panels vary by $\Delta\mu$ (columns) and $\sigma$ (rows). The legend matches Figure~\ref{fig:wc_global}, with the additional method Hong--Li.
\end{figure}

\normalsize

\begin{figure}[!ht]
    \centering
    \caption{Winner's Curse Bias (Select)}
    \includegraphics[width=0.9 \textwidth]{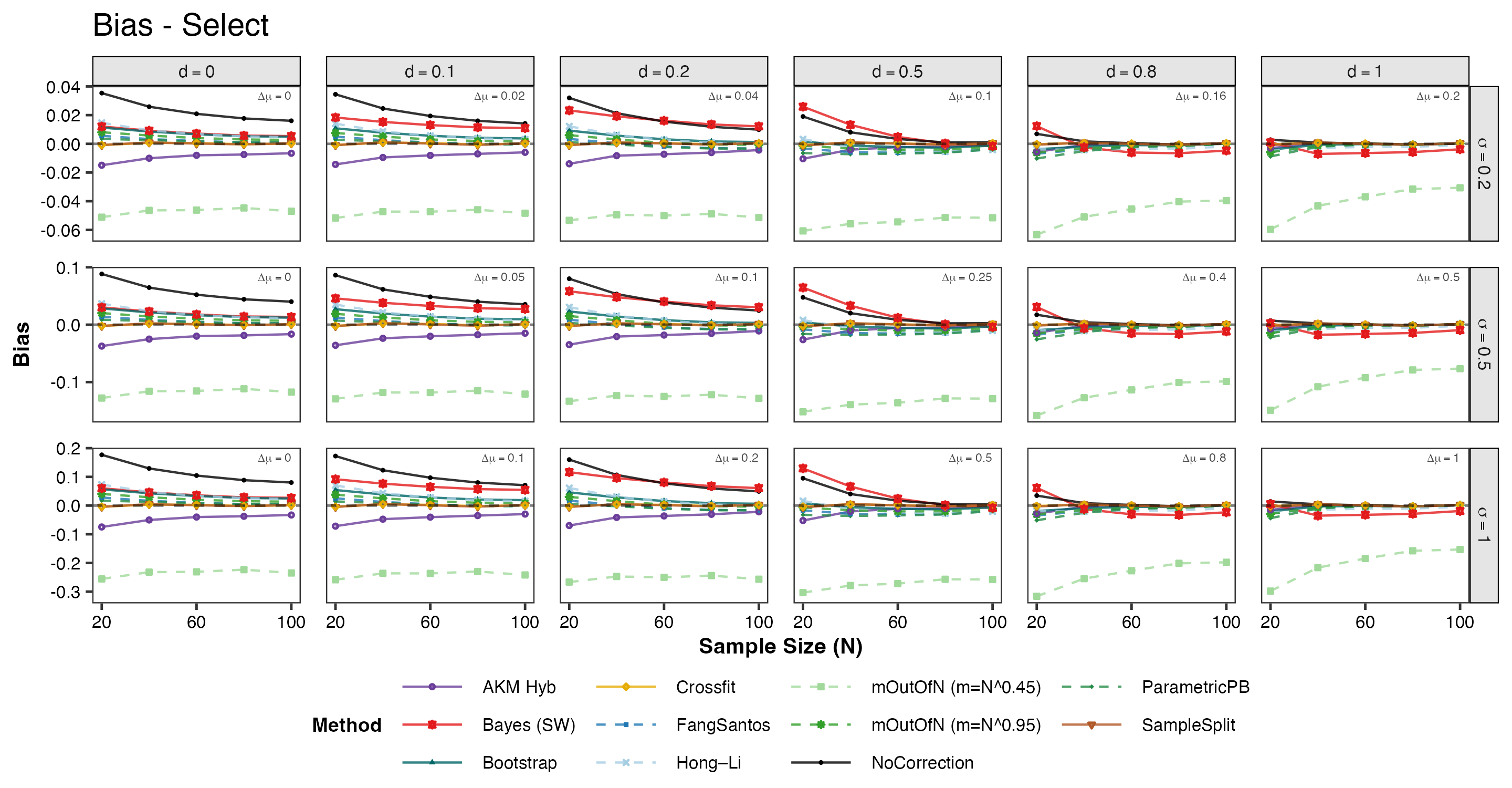}
    \label{fig:wc_select_app}
    \footnotesize
    
    Notes: WC Select bias as a function of sample size $N$. Panels vary by $\Delta\mu$ (columns) and $\sigma$ (rows). The legend matches Figure~\ref{fig:wc_global_app}.
\end{figure}

\normalsize

\begin{figure}[!ht]
    \centering
    \caption{Mean Squared Error (Global)}
    \includegraphics[width=0.9 \textwidth]{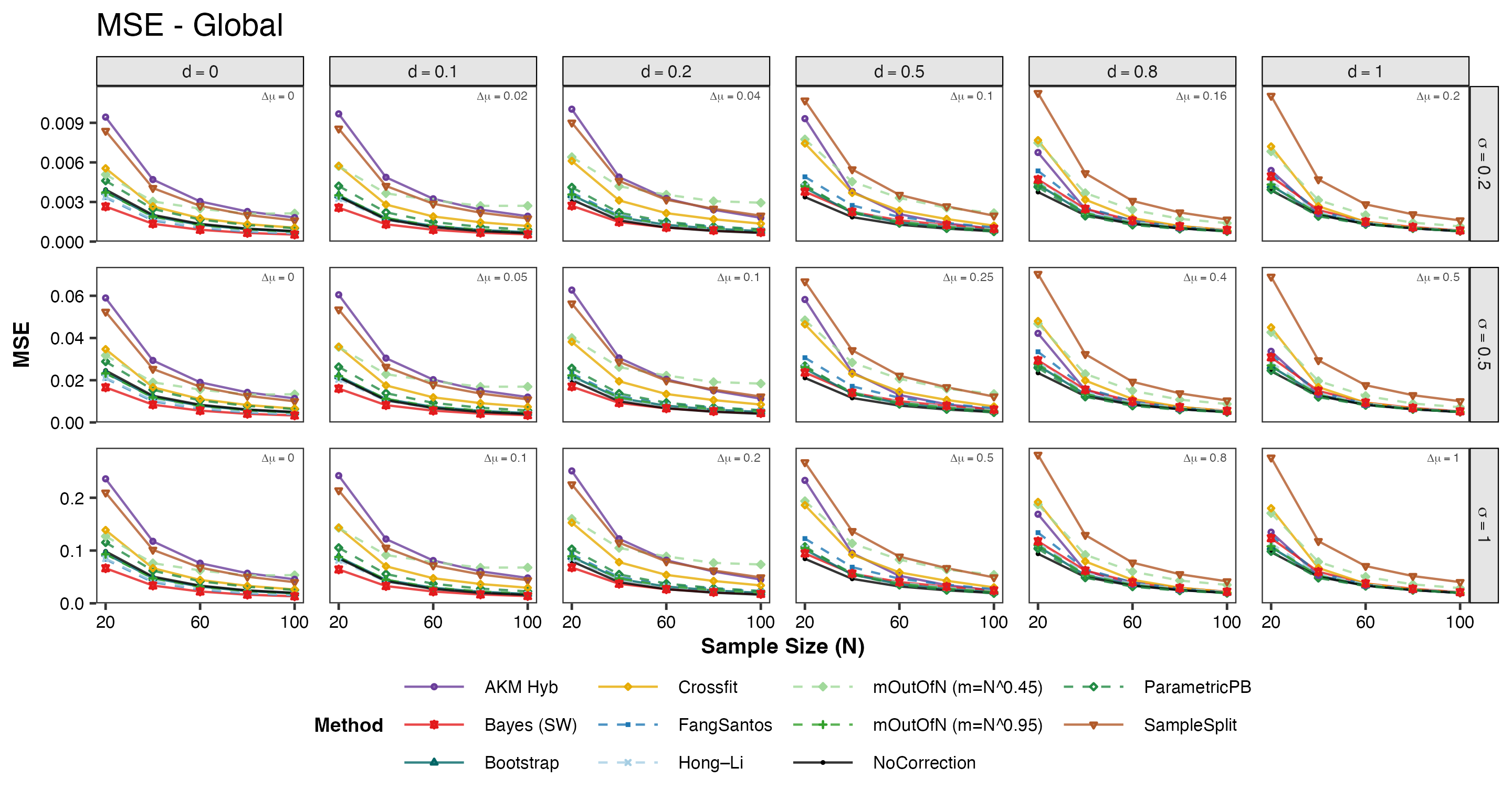}
    \label{fig:mse_global_app}
        \footnotesize

        Notes: Mean squared error (MSE) for WC Global as a function of sample size $N$. Panels vary by $\Delta\mu$ (columns) and $\sigma$ (rows). Legend as in Figure~\ref{fig:wc_global_app}.
\end{figure}

\normalsize

\begin{figure}[!ht]
    \centering
    \caption{Mean Squared Error (Select)}
    \includegraphics[width=0.9 \textwidth]{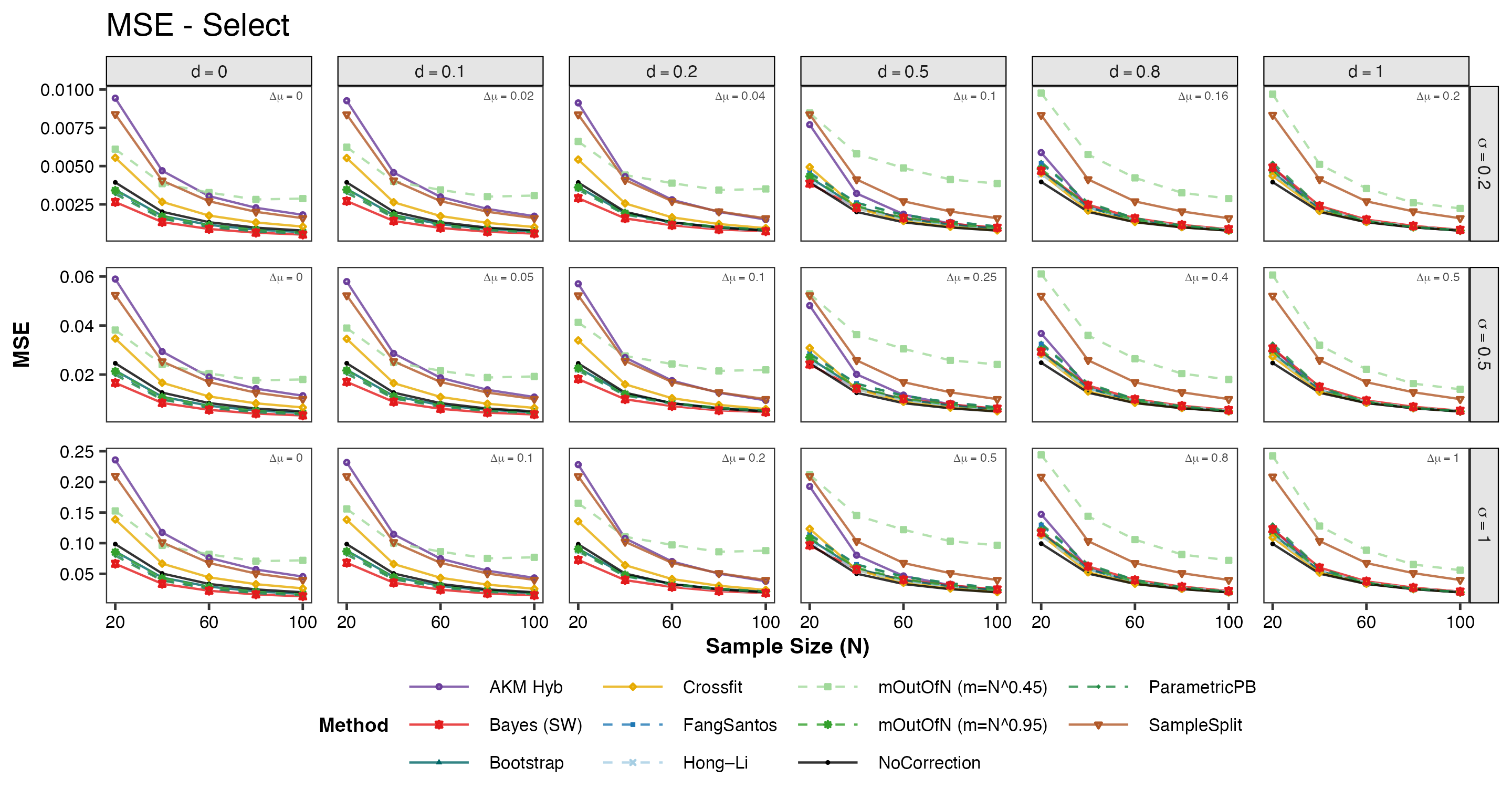}
    \label{fig:mse_select_app}
        \footnotesize

        Notes: Mean squared error (MSE) for WC Select as a function of sample size $N$. Panels vary by $\Delta\mu$ (columns) and $\sigma$ (rows). Legend as in Figure~\ref{fig:wc_global_app}.
\end{figure}

\normalsize

\begin{figure}[!ht]
    \centering
    \caption{Coverage (Global)}
    \includegraphics[width=0.9 \textwidth]{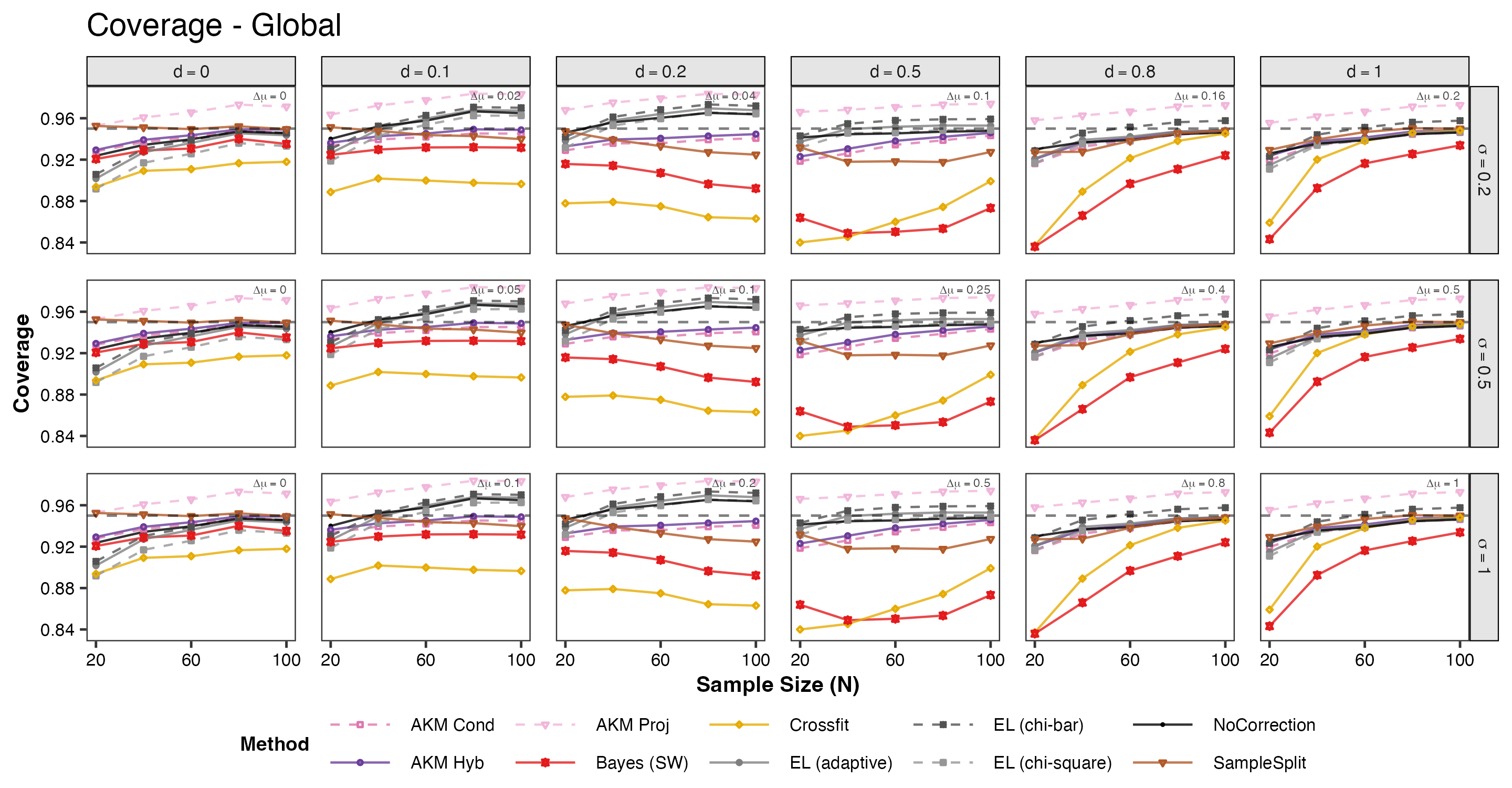}
    \label{fig:cov_global_app}

\footnotesize
    Empirical coverage of $\alpha$-level two-sided confidence intervals for $\mu_{k^*}$, i.e., $\Pr\left(\mu_{k^*}\in CI_\alpha(\hat\mu_{\hat k})\right)$, as a function of sample size $N$. Panels vary by $\Delta\mu$ (columns) and $\sigma$ (rows).   Legend as in Figure~\ref{fig:cov_global}, adding \cite{andrews2024inference} (conditional, projection) and EL with $\chi^2$ or chi-bar-squared critical values.
\end{figure}

\normalsize

\begin{figure}[!ht]
    \centering
    \caption{Coverage (Select)}
    \includegraphics[width=0.9 \textwidth]{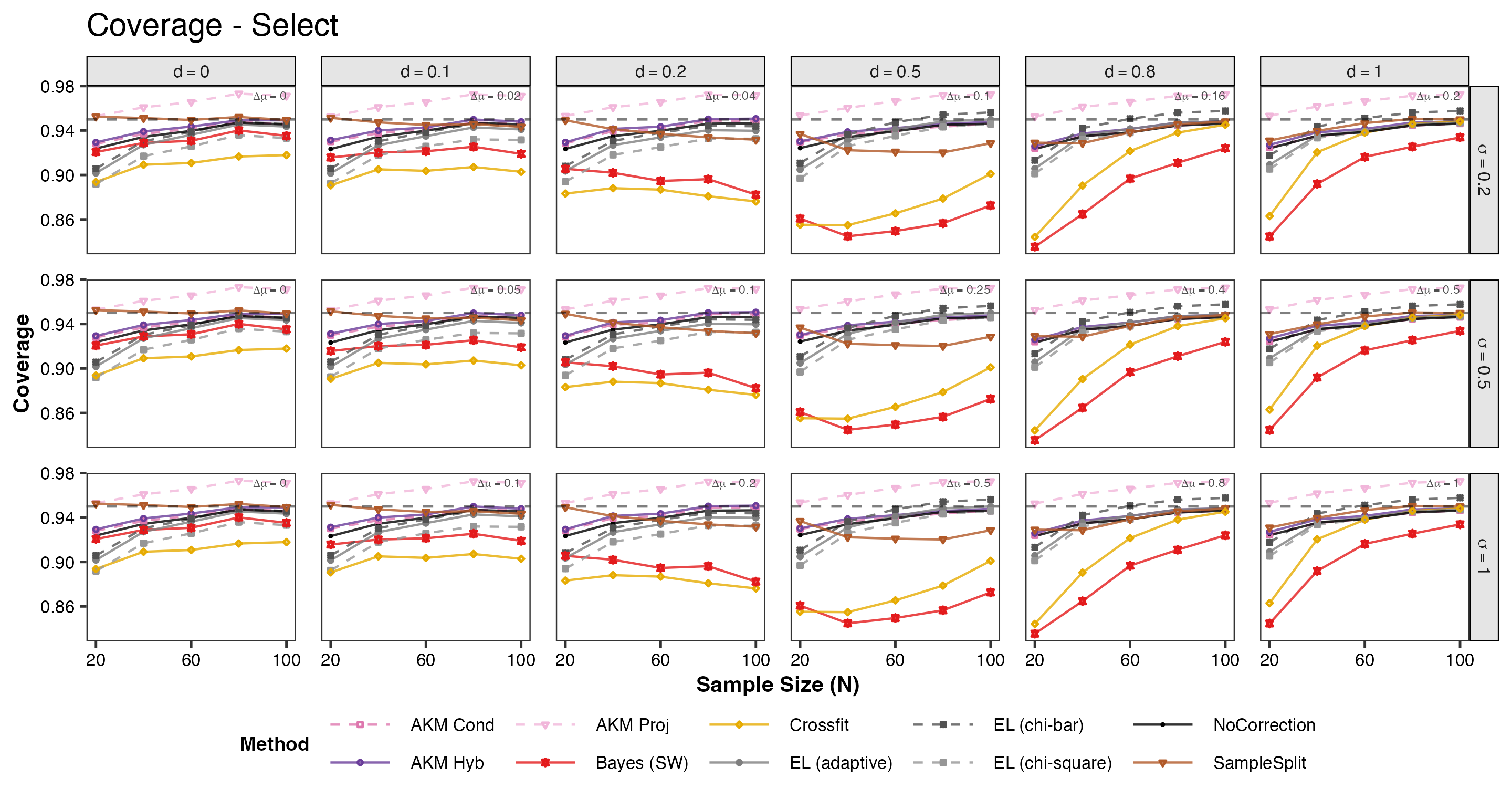}
    \label{fig:cov_select_app}

    \footnotesize
     Notes: Empirical coverage of $\alpha$-level two-sided confidence intervals for $\mu_{\hat k}$, i.e., $\Pr\left(\mu_{\hat k}\in CI_\alpha(\hat\mu_{\hat k})\right)$, as a function of sample size $N$. Panels vary by $\Delta\mu$ (columns) and $\sigma$ (rows). Legend as in Figure~\ref{fig:cov_global_app}.
\end{figure}

\normalsize

\clearpage
\subsection{Binomial distribution}

\begin{figure}[!ht]                                                         
    \centering                                                   \caption{Winner's Curse Bias (Global)}                                      
        
    \includegraphics[width=0.9 \textwidth]{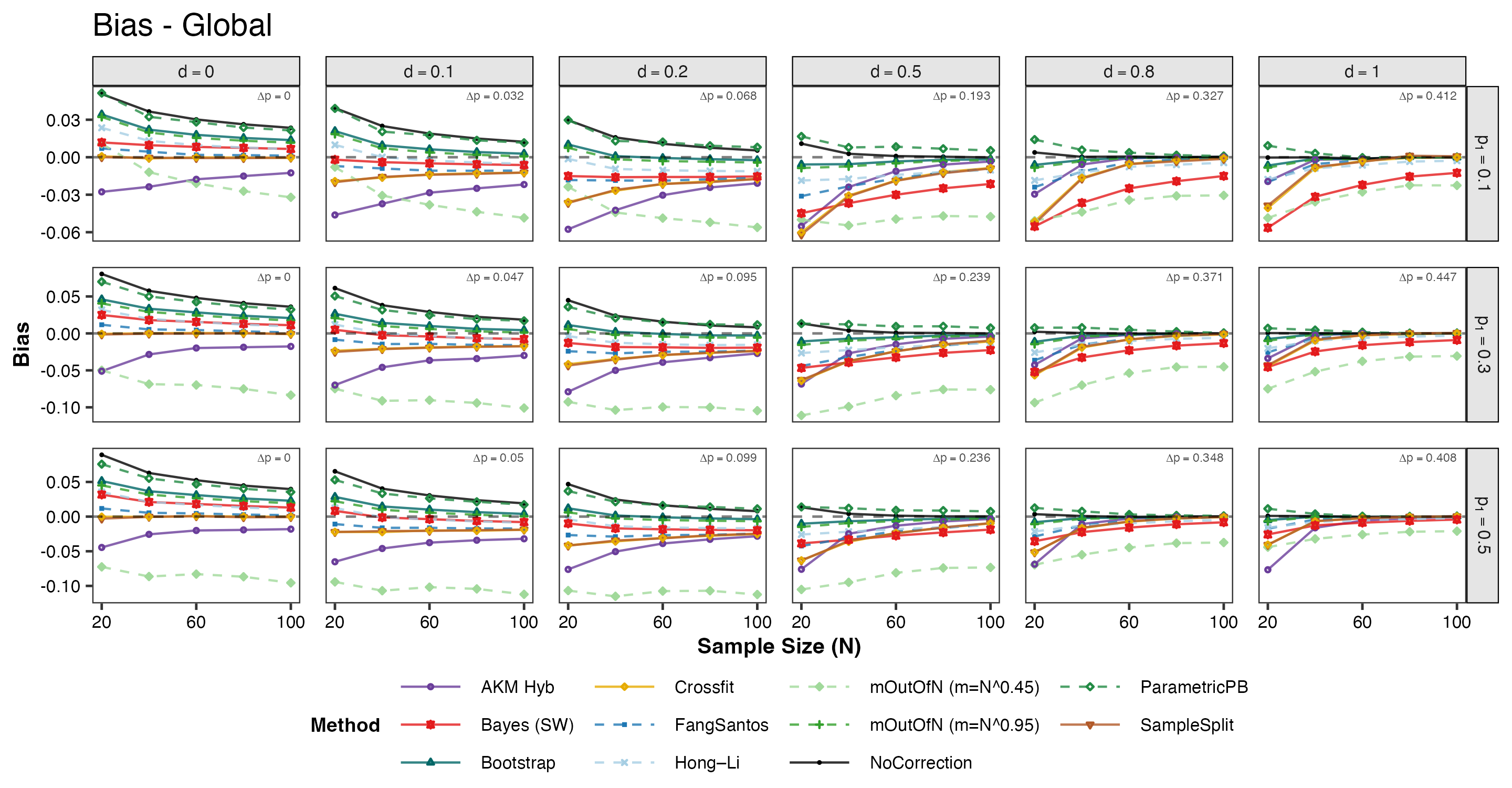}             
    \label{fig:dgrid-bias-global}                                               
    \footnotesize                                  
    
    \textit{Notes:} WC Global as a function of sample size. Panels vary by      
    Cohen's $d$ (columns, $d \equiv \Delta\mu/\sigma$) and $\sigma$ (rows);     
    $\Delta\mu$ shown as an inset annotation in each panel.         
  \end{figure}                                                         
  \begin{figure}[!htbp]                                                         
    \centering                                                 
        \caption{Winner's Curse Bias (Select)}                     
\includegraphics[width=0.9 \textwidth]{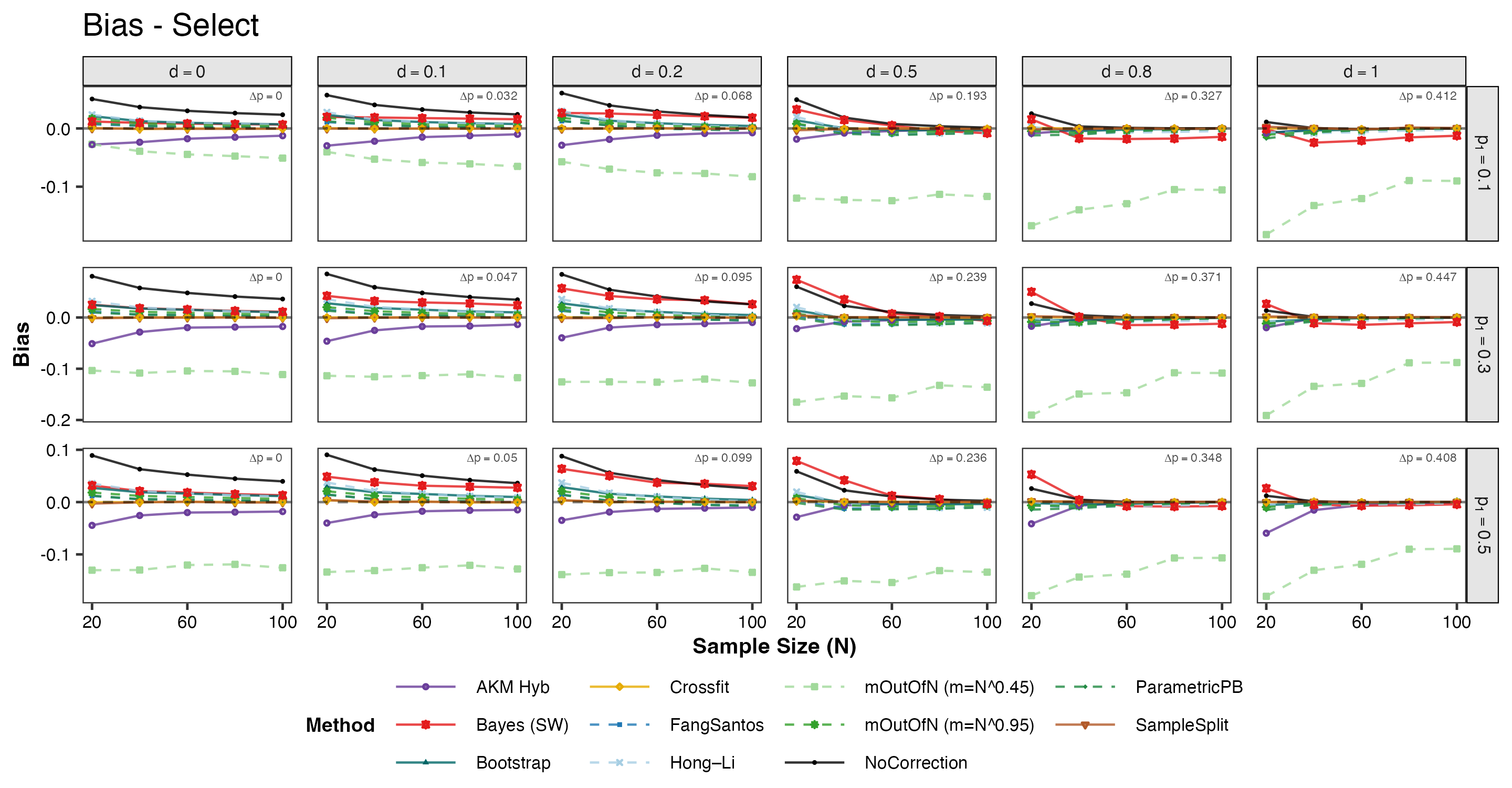}             
    \label{fig:dgrid-bias-select}
    \footnotesize                             
    
    \textit{Notes:} WC Select bias as a function of sample size $N$.
    Panels vary by Cohen's $d$ (columns) and $\sigma$ (rows).                   
  \end{figure}                                                 

  \begin{figure}[!htbp]
    \centering
        \caption{Mean Squared Error (Global)}
\includegraphics[width=0.9 \textwidth]{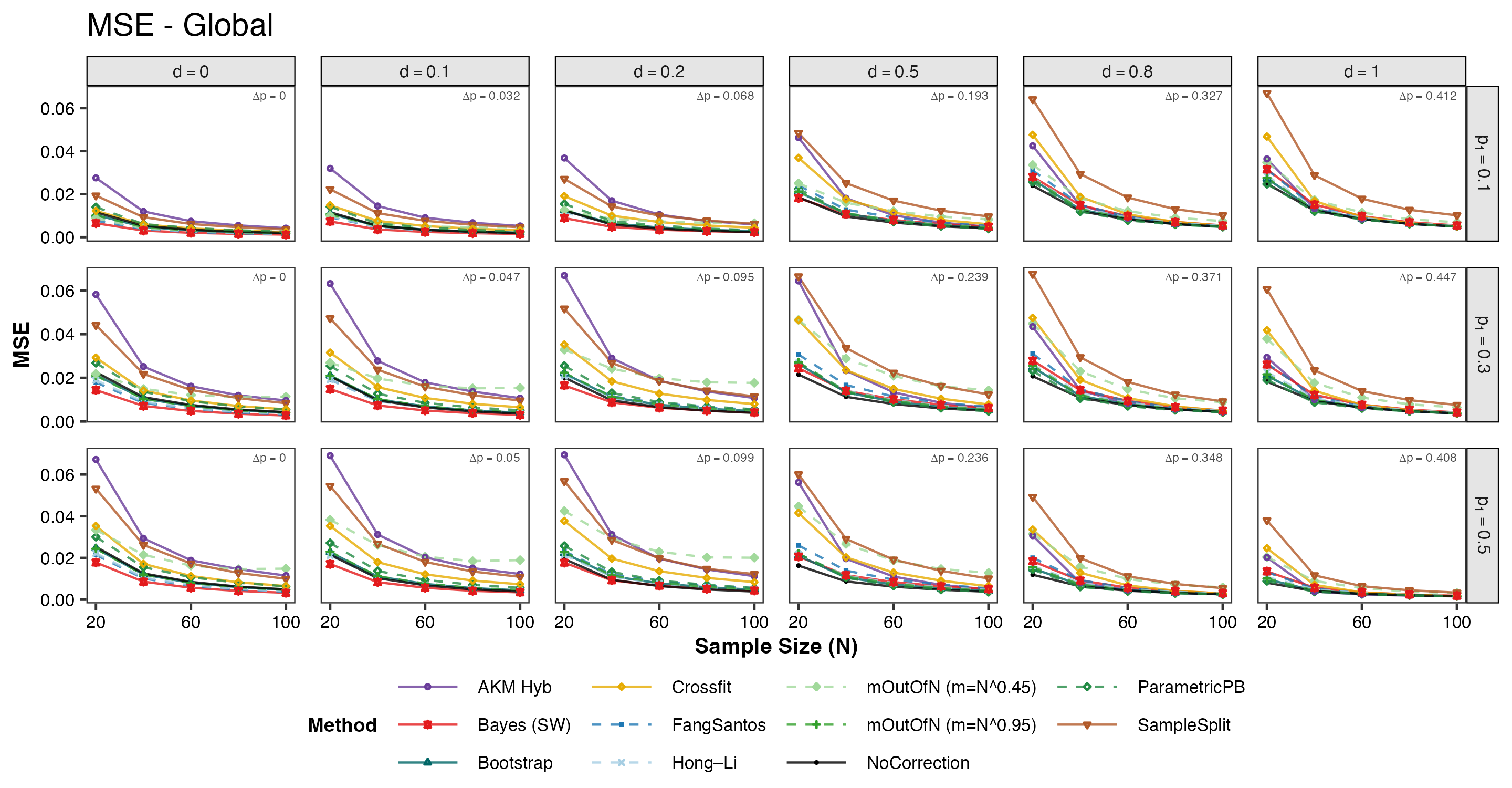}
    \label{fig:dgrid-mse-global}              
    \footnotesize                         
    
    \textit{Notes:} MSE for WC Global as a function of sample size $N$.
    Panels vary by Cohen's $d$ (columns) and $\sigma$ (rows).                   
  \end{figure}                            
                                                  
  \begin{figure}[!htbp]                                                         
    \centering
        \caption{Mean Squared Error (Select)}                      
\includegraphics[width=0.9 \textwidth]{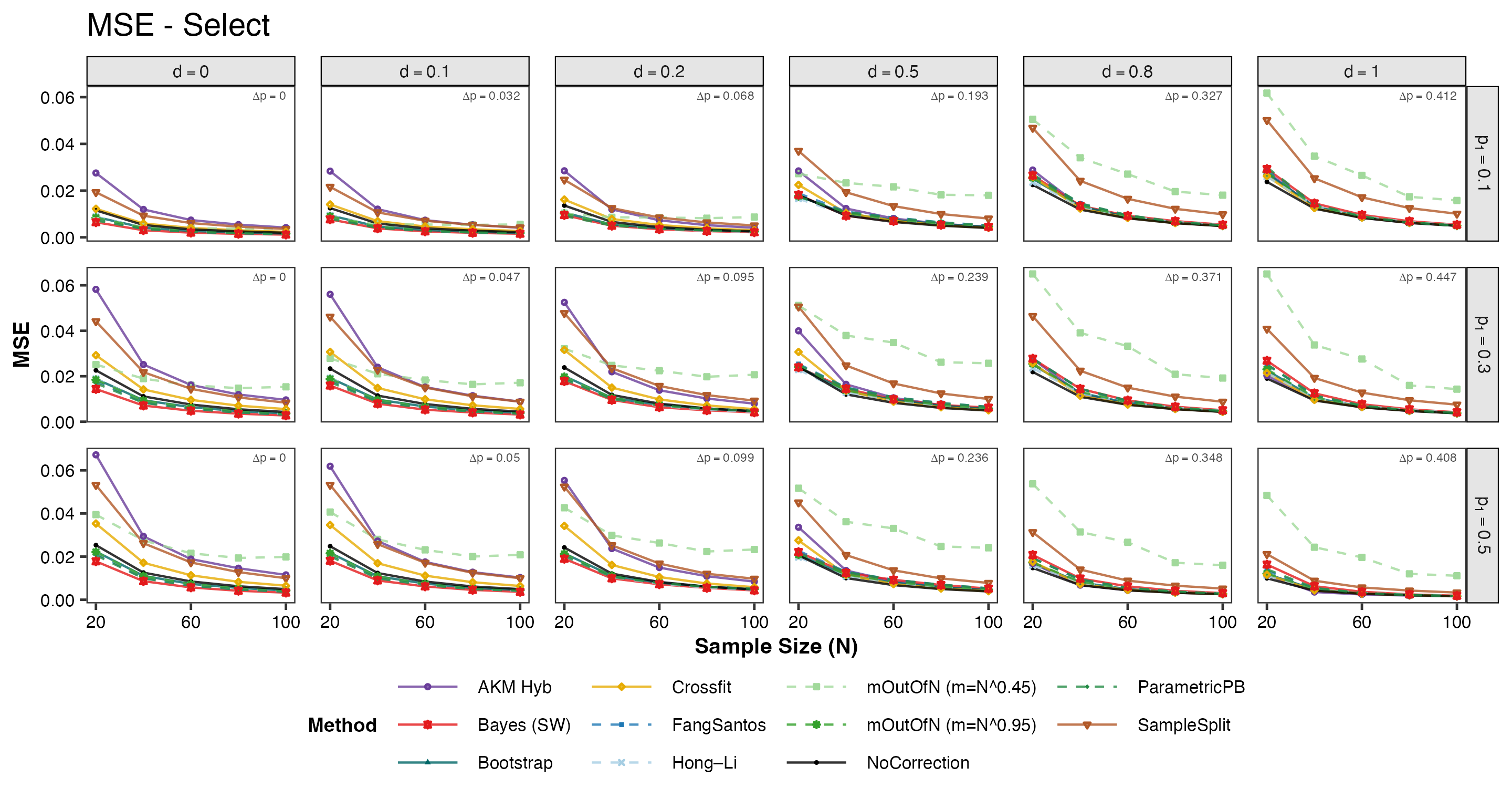}            
    \label{fig:dgrid-mse-select}          
    \footnotesize                             
    
    \textit{Notes:} MSE for WC Select as a function of sample size $N$.         
    Panels vary by Cohen's $d$ (columns) and $\sigma$ (rows).
  \end{figure}                                                                  
                                                 
  \begin{figure}[!htbp]                                                         
    \centering                             \caption{Coverage (Global)}                                                 
      \includegraphics[width=0.9 \textwidth]{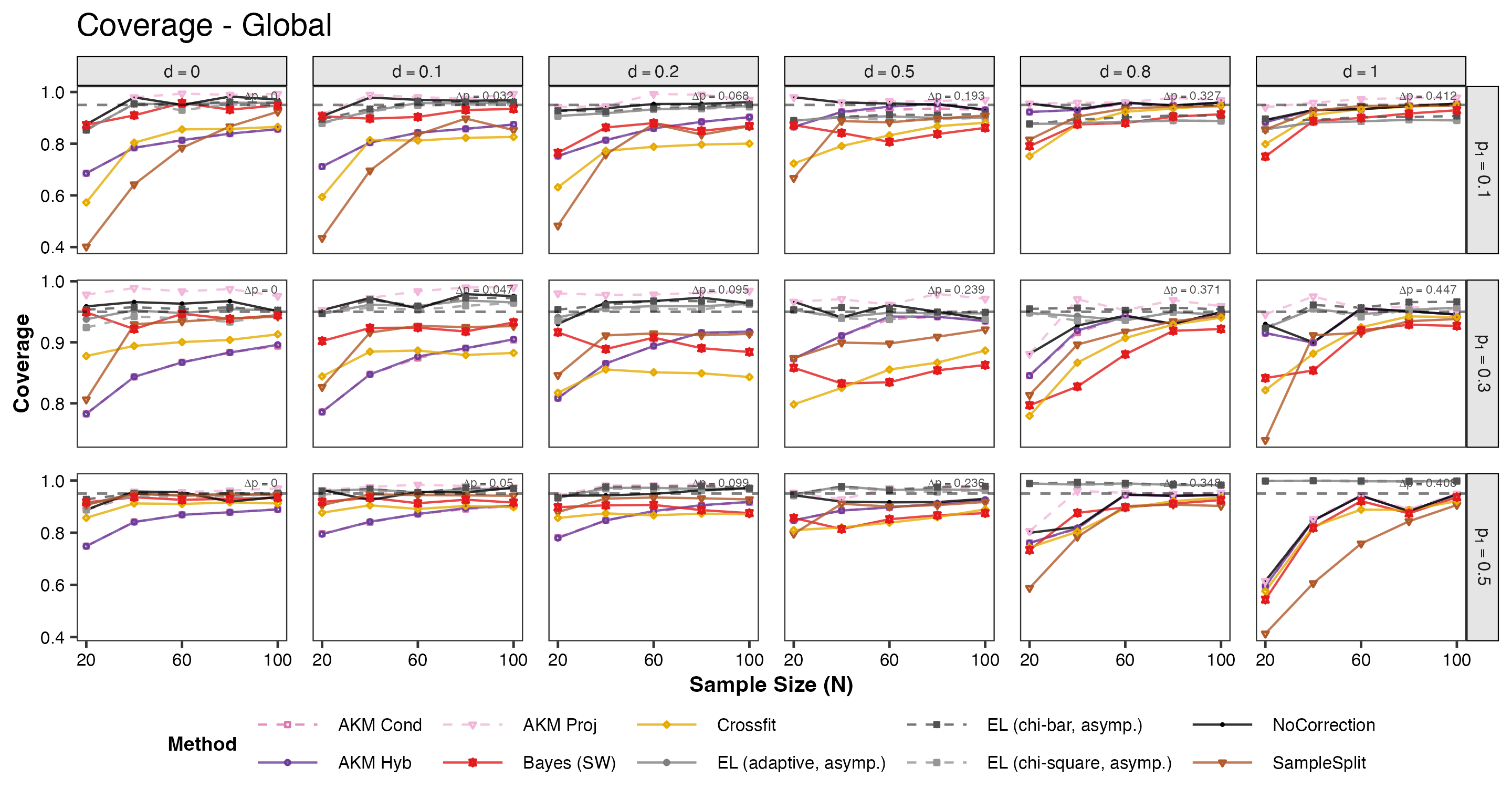}
    \label{fig:dgrid-cov-global}                               
    \footnotesize
    
    \textit{Notes:} Empirical coverage of 95\% two-sided CIs for $\mu_{k^*}$,   
  i.e.,  
    $\Pr(\mu_{k^*}\in CI_\alpha(\hat{\mu}_{\hat k}))$, as a function of sample  
  size $N$.  Panels vary by Cohen's $d$ (columns) and $\sigma$ (rows).                                 
  \end{figure}                                                 
                                      
  \begin{figure}[!htbp]                                        
    \centering                                                             \caption{Coverage (Select)}                        
    \includegraphics[width=0.9 \textwidth]{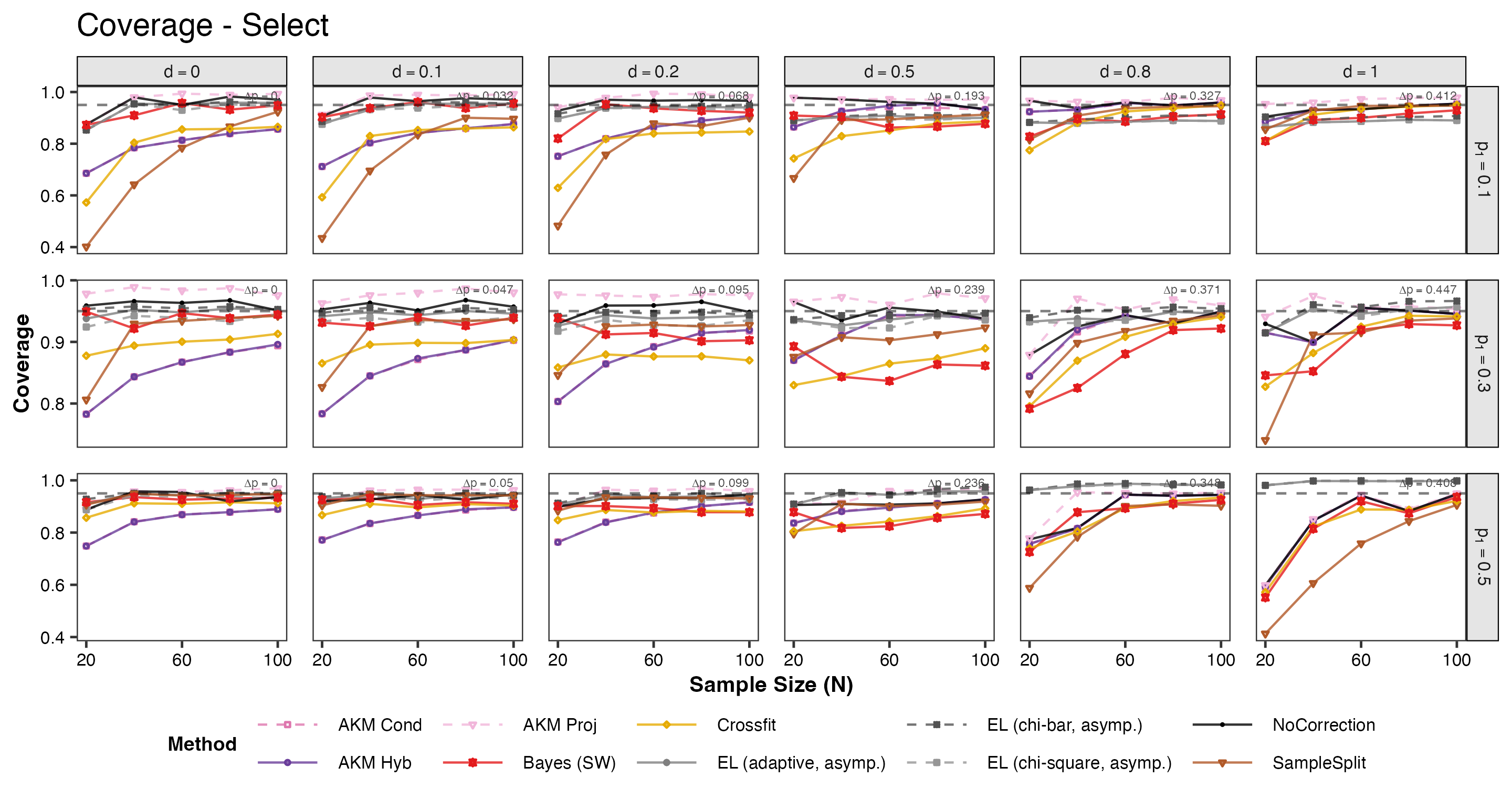}
    \label{fig:dgrid-cov-select}          
    \footnotesize

    \textit{Notes:} Empirical coverage of 95\% two-sided CIs for                
  $\mu_{\hat{k}}$, i.e.,                      
    $\Pr(\mu_{\hat k}\in CI_\alpha(\hat{\mu}_{\hat k}))$, as a function of      
  sample size $N$.                                             
    Panels vary by Cohen's $d$ (columns) and $\sigma$ (rows).                   
  \end{figure}   

\clearpage

\section{Multiple Arms}

\begin{table}[!ht]
\centering
\caption{Key Modifications for $K$-Arm Generalization}
\label{tab:k-arm-changes}
\footnotesize
\renewcommand{\arraystretch}{1.1}
\begin{tabular}{p{2.2cm} p{13.2cm}}
\toprule
\textbf{Method} & \textbf{Modification for $K$ Arms} \\
\midrule

Standard nonparametric bootstrap
& Replace $\max$ over 2 with $\max$ over $K$. Stratified fallback if any arm is absent from a pooled resample. \\[2pt]
\hline \\[-6pt]

$m$-out-of-$n$
& $m = \max(K, \lfloor N^\gamma \rfloor)$. Proportional-allocation fallback ($\lfloor m/K \rfloor$ per arm) if any arm is absent. \\[2pt]
\hline \\[-6pt]

Parametric PB
& Arm-specific pairwise SE $\widehat{se}_{\text{diff},k} = \sqrt{\widehat{se}_{\hat k}^2 + \widehat{se}_k^2}$ and threshold $\eta_{N,k} = c_\eta \widehat{se}_{\text{diff},k}\sqrt{2\log\log N}$ for each $k \neq \hat k$. Arm $k$ joins the active set $\hat A$ if $|\bar Y_{\hat k} - \bar Y_k| \leq \eta_{N,k}$. Active-set arms shrunk to group mean $\tilde\tau_k = |\hat A|^{-1}\sum_{j \in \hat A}\bar Y_j$; others keep $\bar Y_k$. \\[2pt]
\hline \\[-6pt]

Fang \& Santos
& \textbf{Pre-test:} $S_k = (\bar Y_{\hat k} - \bar Y_k)/\widehat{se}_{\text{diff},k}$; active set $\hat A = \{k : S_k \leq \kappa_n\}$, $\kappa_n = \sqrt{\log(\min_k n_k)}$. Winner always in $\hat A$.
\textbf{Bootstrap:} Stratified; $h_k^{*(b)} = \sqrt{n_k}(\bar Y_k^{*(b)} - \bar Y_k)$.
\textbf{Derivative:} $d^{*(b)} = \max_{k \in \hat A} h_k^{*(b)}$.
\textbf{Correction:} $\hat\phi_{FS} = \hat\phi - \frac{1}{B}\sum_b d^{*(b)} / \sqrt{n_{\text{eff}}}$, $n_{\text{eff}} = \min_k n_k$. \\[2pt]
\hline \\[-6pt]

Hong \&  Li
& Stratified bootstrap perturbations $h_k^{*(b)} = \sqrt{N}(\bar Y_k^{*(b)} - \bar Y_k)$ for $k=1,\ldots,K$. Shifted functional $\hat\phi_\varepsilon^{*(b)} = \max_{k}\{\bar Y_k + \varepsilon h_k^{*(b)}\}$, $\varepsilon = N^{-1/4}$. Numerical derivative $d^{*(b)} = (\hat\phi_\varepsilon^{*(b)} - \hat\phi)/\varepsilon$. Correction: $\hat\phi_{HL} = \hat\phi - \frac{1}{B}\sum_b d^{*(b)}/\sqrt{N}$.  \\[2pt]
\hline \\[-6pt]

Sample split
& $\hat k_A = \arg\max_{k=1,\ldots,K}\bar Y_k(D_A)$. Retry splitting to ensure both folds contain all $K$ arms with $\geq 2$ obs each. \\[2pt]
\hline \\[-6pt]

Cross-fit
& Same as sample split in both directions, with $\arg\max$ over $K$. \\[2pt]
\hline \\[-6pt]

AKM Cond.
& $x_L = \bar Y_{(2)}$ (second-highest mean). Independence $\Rightarrow$ $\{\hat k \text{ wins}\}$ reduces to $\bar Y_{\hat k} \geq x_L$. Truncated normal on $[x_L, \infty)$ unchanged. \\[2pt]
\hline \\[-6pt]

AKM Proj.
& $c_\alpha(K) = \Phi^{-1}\!\bigl(\frac{1+(1-\alpha)^{1/K}}{2}\bigr)$. CI: $\bar Y_{\hat k} \pm c_\alpha(K)\cdot \widehat{se}_{\hat k}$. \\[2pt]
\hline \\[-6pt]

AKM Hybrid
& Projection bounds use $c_\beta(K)$: $L(\mu) = \max\{x_L, \mu - c_\beta(K) sd_W\}$, $U(\mu) = \mu + c_\beta(K) sd_W$. \\[2pt]
\hline \\[-6pt]

EL (3 variants)
& \textbf{Profile deviance} over $K$ branches: $\text{dev}(\theta) = \min_k\{D_k(\theta) + \sum_{j\neq k}\mathbf{1}\{\bar Y_j > \theta\}D_j(\theta)\}$.
\textbf{Chi-bar-sq CDF} for $J$ active arms: $\Pr(T \leq t) = \sum_{k=1}^{J}\binom{J}{k}(\frac{1}{2})^{J} F_k(t) + (\frac{1}{2})^{J}[1-(1-F_1(t))^J]$, extending Eq.~(23).
Active set: $J_{\text{eff}} = |\{k : S_k \leq \kappa_n\}|$, $S_k = (\bar Y_{\hat k} - \bar Y_k)/(\hat\sigma\sqrt{2/n_{\min}})$. Variants: \emph{adaptive} ($J_{\text{eff}}$), \emph{chi-bar} ($J=K$), $\chi^2$ ($J=1$). \\[2pt]
\hline \\[-6pt]

Bayes (S\&W, 2006)
& Grand mean $\bar\mu = K^{-1}\sum_k \bar Y_k$. MOM prior variance $\hat\sigma^2_\mu = \max\bigl(0, \frac{1}{K-1}\sum_k(\bar Y_k - \bar\mu)^2 - \frac{1}{K}\sum_k \widehat{se}_k^2\bigr)$. Shrinkage $\alpha_k = \hat\sigma^2_\mu/(\hat\sigma^2_\mu + \widehat{se}_k^2)$; posterior $\hat v_k = \alpha_k\bar Y_k + (1-\alpha_k)\bar\mu$. Select arm with highest $\hat v_k$. \\

\bottomrule
\end{tabular}

\vspace{3pt}
{\scriptsize \textit{Notes:} $\hat k = \arg\max_k \bar Y_k$, $n_{\min} = \min_k n_k$, $\widehat{se}_k = \hat\sigma_k/\sqrt{n_k}$, $\widehat{se}_{\text{diff},k} = \sqrt{\widehat{se}_{\hat k}^2 + \widehat{se}_k^2}$, and $\hat\sigma = \sqrt{K^{-1}\sum_k\hat\sigma_k^2}$. All other details ($B=200$, split $1/2$, $\alpha=0.05$, $\beta = \alpha/10$) identical to the two-arm case in Tables~3 and~5.}
\end{table}
\begin{figure}[!ht]
      \centering
      \caption{Winner's Curse Bias (Global) for $K$ Arms}
      \includegraphics[width=0.8 \textwidth]{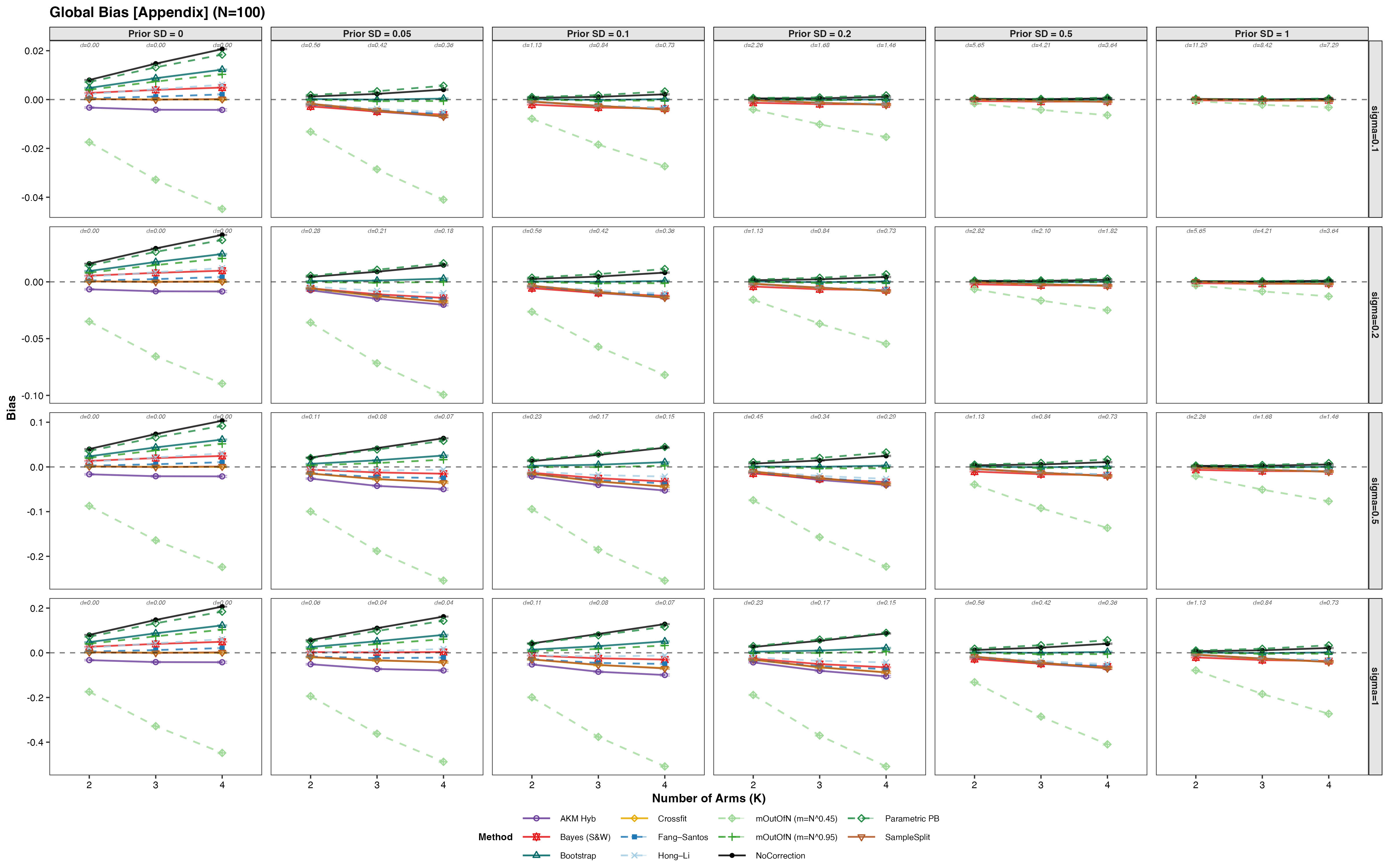}
      \label{fig:prior_K_bias_global_app}

       \footnotesize
       Notes: WC Global as a function of number of arms $K$. Each arm's true
  mean is drawn from $\tau_k \sim \mathcal{N}(1, \sigma_0^2)$ with $N=100$.
  Panels vary by Prior SD $\sigma_0$ (columns) and $\sigma$ (rows). 
  \end{figure}

  \begin{figure}[!ht]
      \centering
      \caption{Winner's Curse Bias (Select) for $K$ Arms}
      \includegraphics[width=0.8 \textwidth]{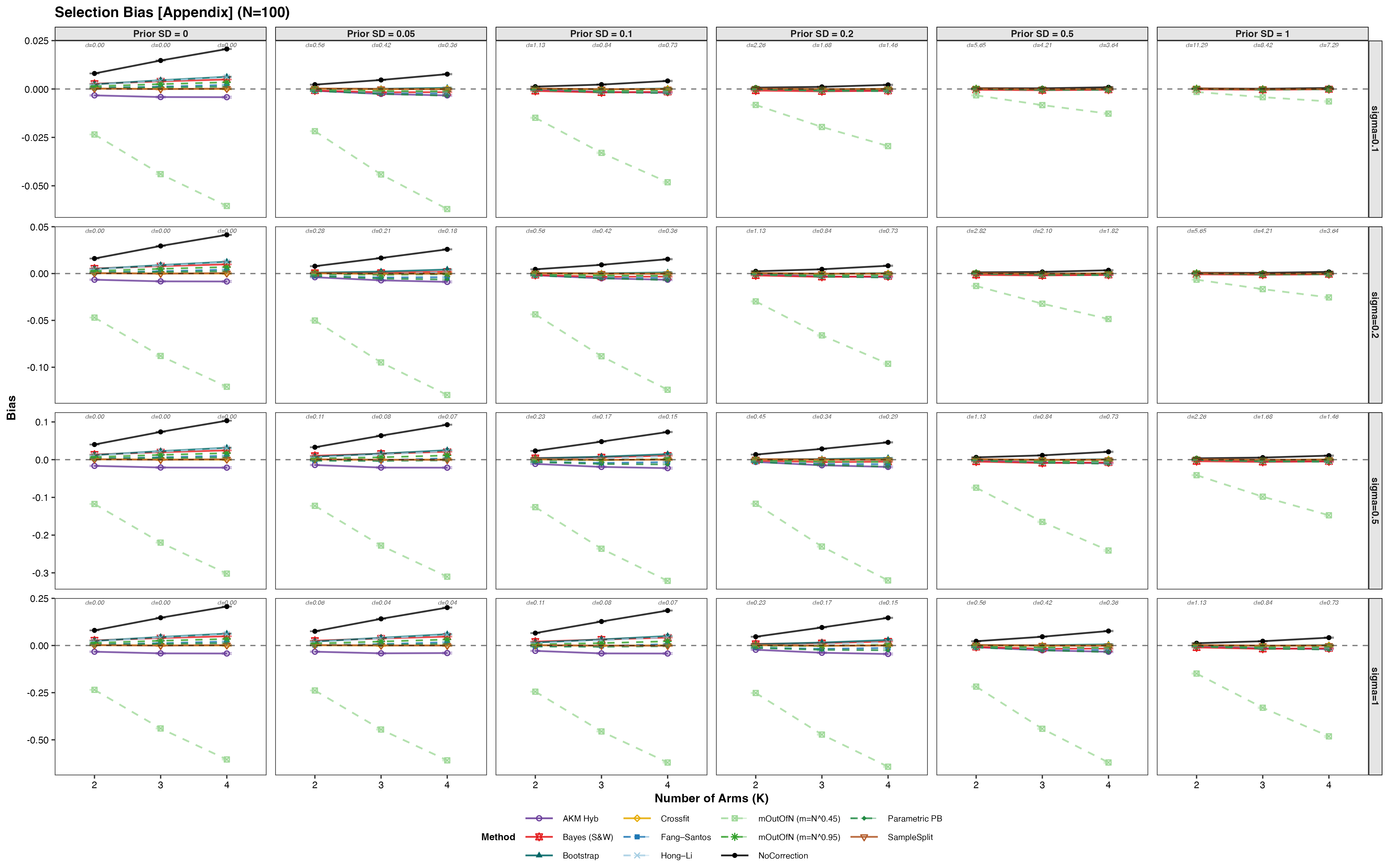}
      \label{fig:prior_K_bias_select_app}

       \footnotesize
       Notes: WC Select bias as a function of number of arms $K$. Panels vary by
   Prior SD $\sigma_0$ (columns) and $\sigma$ (rows). 
  \end{figure}

\begin{figure}[!ht]
      \centering
      \caption{Mean Squared Error (Global) for $K$ Arms}

  \includegraphics[width=0.8 \textwidth]{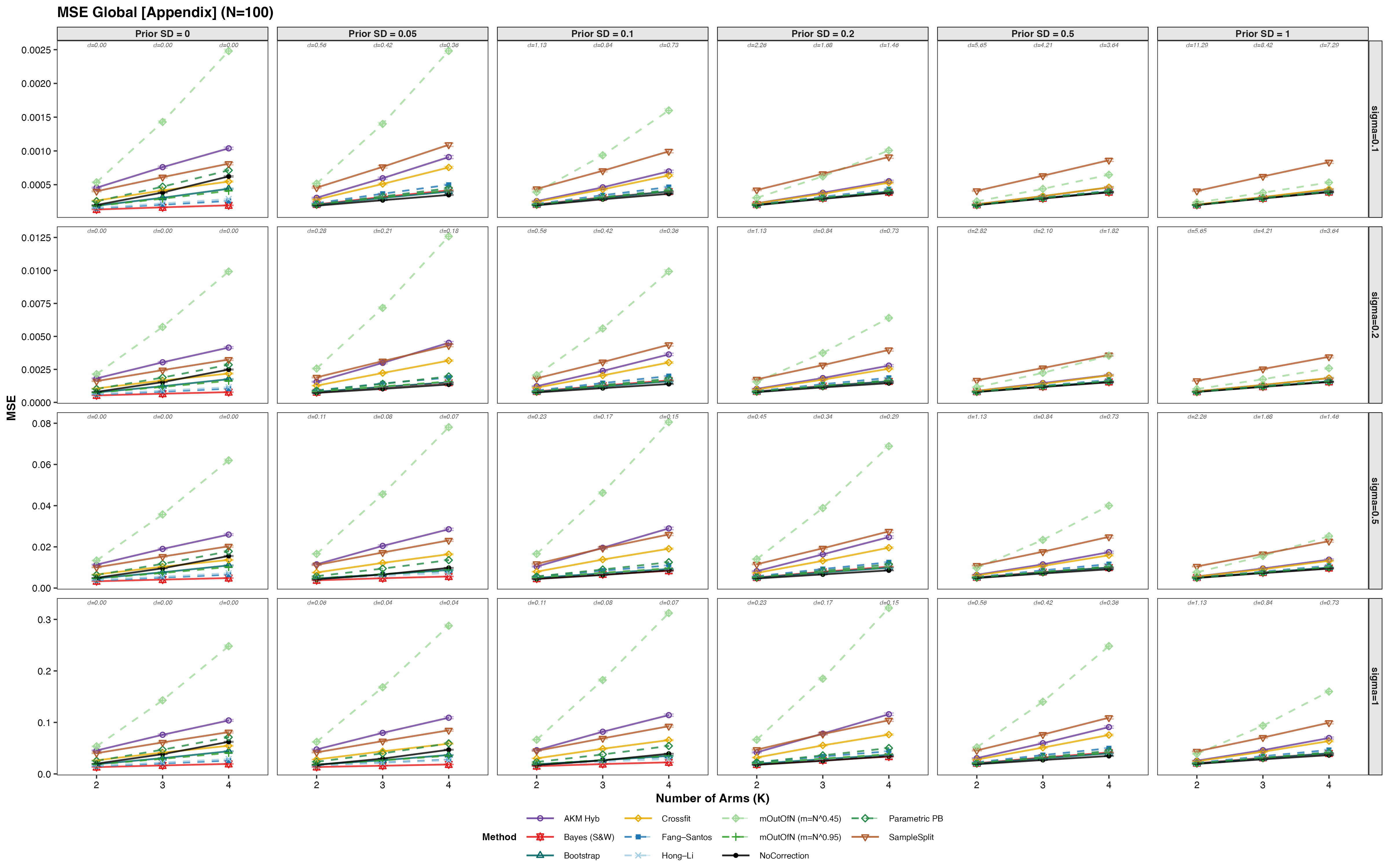}
      \label{fig:prior_K_mse_global_app}

       \footnotesize
       Notes: Mean squared error (MSE) for WC Global as a function of number of
  arms $K$. Panels vary by Prior SD $\sigma_0$ (columns) and $\sigma$ (rows).
  $N=100$.
  \end{figure}

  \begin{figure}[!ht]
      \centering
      \caption{Mean Squared Error (Select) for $K$ Arms}

  \includegraphics[width=0.8 \textwidth]{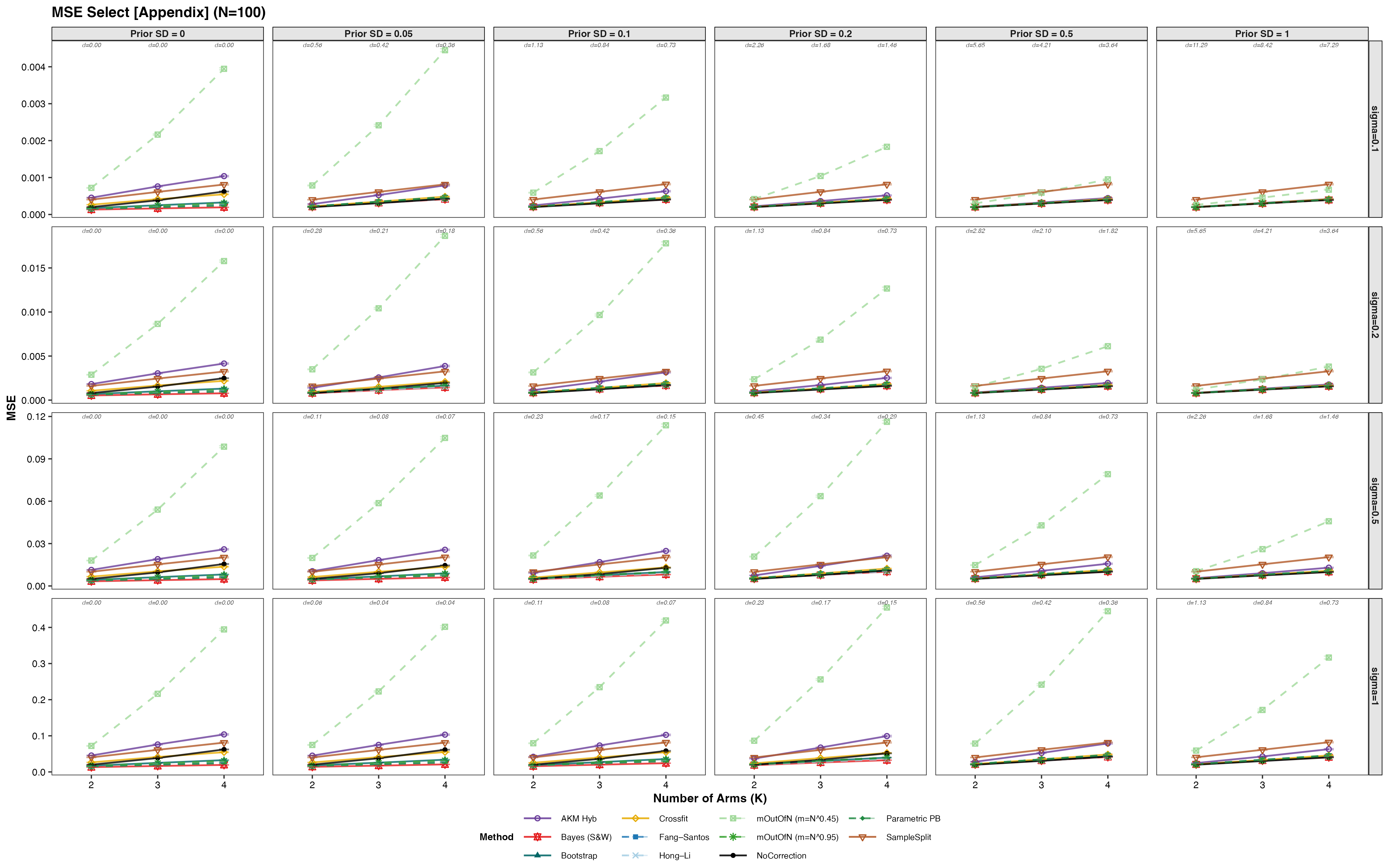}
      \label{fig:prior_K_mse_select_app}

       \footnotesize
       Notes: Mean squared error (MSE) for WC Select as a function of number of
  arms $K$. Panels vary by Prior SD $\sigma_0$ (columns) and $\sigma$ (rows).
  $N=100$. 
  \end{figure}

  \begin{figure}[!ht]
      \centering
      \caption{Coverage (Global) for $K$ Arms}

  \includegraphics[width=0.8 \textwidth]{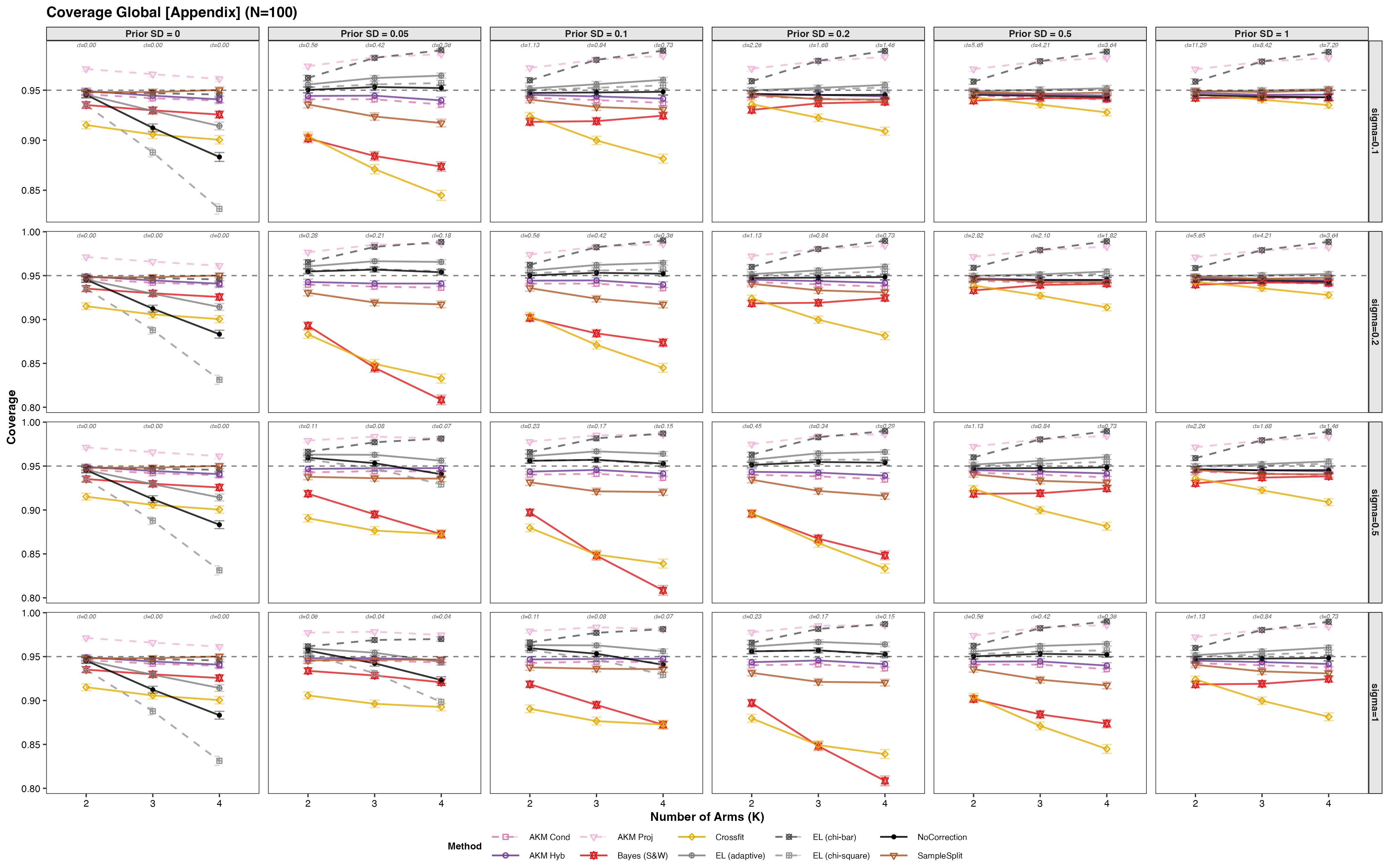}
      \label{fig:prior_K_cov_global_app}

       \footnotesize
       Notes: Empirical coverage of $\alpha$-level two-sided confidence
  intervals for $\mu_{k^*}$, i.e., $\Pr\left(\mu_{k^*}\in
  CI_\alpha(\hat\mu_{\hat k})\right)$, as a function of number of arms $K$.
  Panels vary by Prior SD $\sigma_0$ (columns) and $\sigma$ (rows). 
  \end{figure}

  \begin{figure}[!ht]
      \centering
      \caption{Coverage (Select) for $K$ Arms}

  \includegraphics[width=0.8 \textwidth]{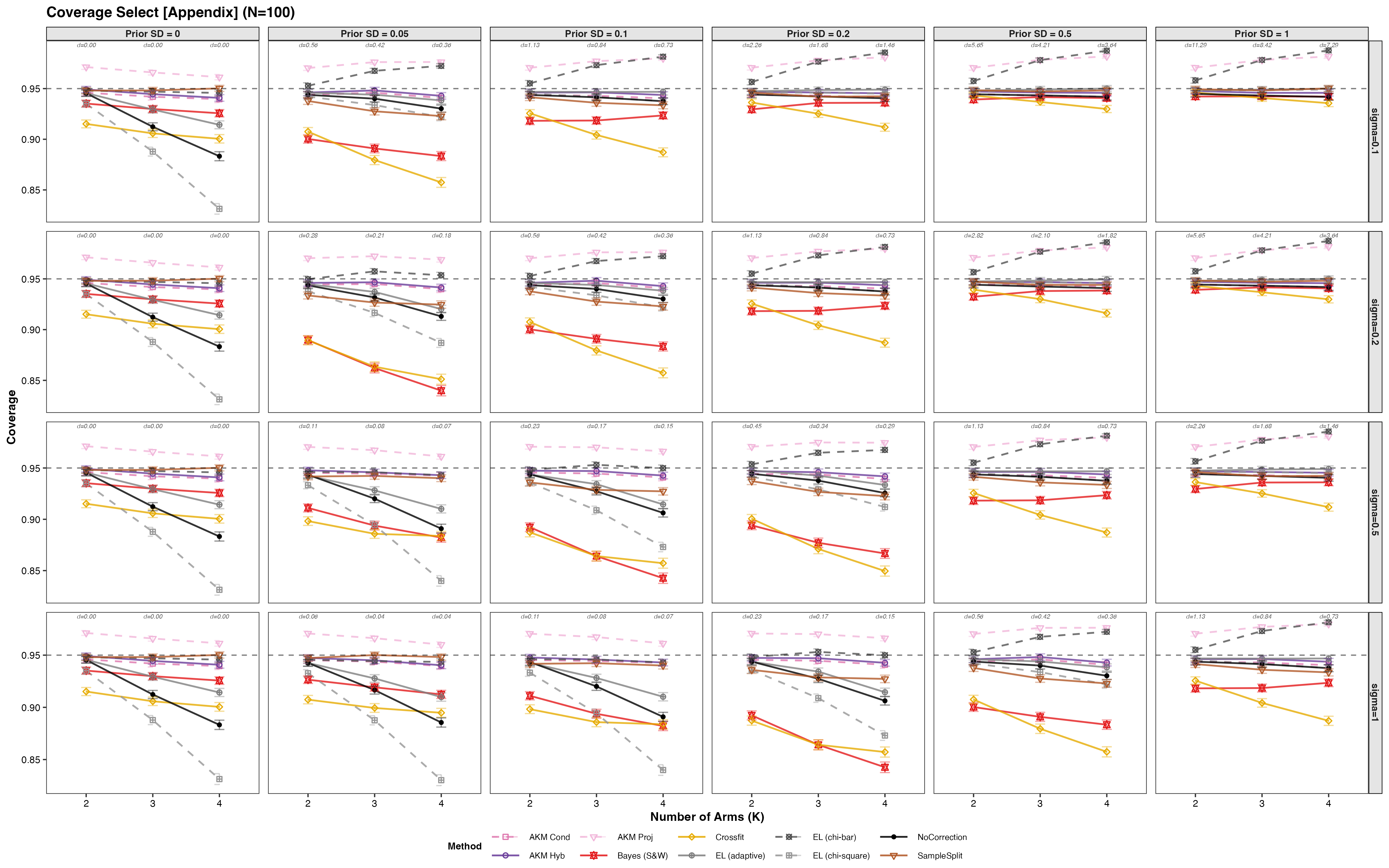}
      \label{fig:prior_K_cov_select_app}

       \footnotesize
       Notes: Empirical coverage of $\alpha$-level two-sided confidence
  intervals for $\mu_{\hat k}$, i.e., $\Pr\left(\mu_{\hat k}\in
  CI_\alpha(\hat\mu_{\hat k})\right)$, as a function of number of arms $K$.
  Panels vary by Prior SD $\sigma_0$ (columns) and $\sigma$ (rows). 
  \end{figure}

\end{document}